\documentclass[11pt]{article}

\usepackage{a4wide}
\usepackage{amssymb}
\usepackage{amsmath}
\usepackage{amsfonts}
\usepackage{amsthm}

\usepackage{graphicx}

\usepackage{stmaryrd}

\usepackage{enumerate}

\usepackage{stackrel}

\usepackage{url}

\usepackage{pj-macros}

\begin{document}

\title{Equivalence of pushdown automata via first-order grammars}

\author{Petr Jan\v{c}ar
\\
	{\small	Dept of Computer Science, Faculty of Science, Palack\'y
	University in Olomouc, 
Czechia}
\\
{\small petr.jancar@upol.cz}
}

\date{}

\maketitle

\begin{abstract}
\noindent	
A decidability proof for bisimulation
equivalence of first-order grammars is given. It is an alternative proof for
a result by S\'enizergues (1998, 2005) that subsumes his affirmative solution 
of the famous decidability question for deterministic pushdown
automata.

The presented proof is conceptually simpler, 
and a particular novelty is that it is not given as two
semidecision procedures but it provides an explicit algorithm that 
might be amenable to a complexity analysis.
\end{abstract}

\section{Introduction}\label{sec:intro}

Decision problems for  semantic equivalences 
have been a frequent topic in computer science. 
For pushdown automata (PDA) 
\emph{language equivalence} was quickly shown undecidable,
while 
the decidability in the case of deterministic PDA (DPDA) 
is 
a famous result 
by S\'enizergues~\cite{Senizergues:TCS2001}.
A finer equivalence, called \emph{bisimulation equivalence} or
\emph{bisimilarity}, has emerged as another fundamental behavioural
equivalence~\cite{Milner89}; for deterministic systems it essentially coincides with
language equivalence.
By~\cite{BBK2} we can exemplify the 
first decidability results for infinite-state systems (a subclass of
PDA, in fact),
and 
refer to~\cite{Srba:Roadmap:04} for a survey of results in a relevant area.

One of the most involved results in this area
shows the decidability of bisimilarity
of equational graphs with finite out-degree, which are equivalent
to PDA with alternative-free $\varepsilon$-steps (if an $\varepsilon$-step is enabled, then it has no alternative); 
S\'enizergues~\cite{Seni05} has thus generalized his
decidability result for DPDA.

We recall that the complexity of the DPDA problem  remains 
far from clear, the problem is known to be 
PTIME-hard and to be in TOWER
(i.e., in the first complexity class beyond elementary 
in the terminology of~\cite{DBLP:journals/toct/Schmitz16}); the upper bound
was shown 
by Stirling~\cite{Stir:DPDA:prim}
(and formulated more explicitly in~\cite{DBLP:conf/fossacs/Jancar14}). 
For PDA the bisimulation equivalence problem is known to be 
nonelementary~\cite{BGKM12} (in fact, TOWER-hard), 
even for real-time PDA, i.e. PDA with no $\varepsilon$-steps.
For the above mentioned PDA with alternative-free $\varepsilon$-steps
the 
problem is even not primitive recursive;
its Ackermann-hardness was shown in~\cite{DBLP:conf/fossacs/Jancar14}.

The decidability proofs, both for DPDA and PDA, are involved and hard
to understand.
This paper aims to contribute to 
a clarification of the more general decidability proof, showing an
algorithm deciding 
bisimilarity of PDA with alternative-free $\varepsilon$-steps.

The proof is shown in the framework 
of labelled transition systems generated by \emph{first-order
grammars} (FO-grammars), which seems to be a particularly convenient
formalism; it is called \emph{term context-free grammars}
in~\cite{Cau95}. Here the states 
(or configurations) 
are first-order terms over a
specified finite set of function symbols (or ``nonterminals''); the
transitions are induced by a \emph{first-order grammar}, which is
a finite set
of labelled rules for rewriting the roots of terms.
This framework 
is equivalent
to the framework of~\cite{Seni05}; cf.,
e.g.,~\cite{CourcelleHandbook,Cau95} and the
references therein, or also~\cite{JancarLICS12}
for a concrete transformation of PDA 
to FO-grammars.
The proof here
is in principle based on the high-level ideas from the proof
in~\cite{Seni05} but with various simplifications and new
modifications.
The presented proof has resulted by a thorough reworking of the
conference paper~\cite{DBLP:conf/icalp/Jancar14}, aiming to get an
algorithm that might be amenable to a complexity analysis.

\emph{Proof overview.} We give a flavour of the process that is
formally realized in the paper.
It is standard to characterize bisimulation equivalence (also called
bisimilarity) in terms of a
turn-based game between Attacker and Defender, say.
If two PDA-configurations, modelled by first-order terms $E,F$ in our
framework, are non-bisimilar, then
Attacker can force his win within $k$ rounds of the game, for some
number $k\in\Nat$; in this case $k{-}1$ for the least such $k$ can be
viewed as the \emph{equivalence-level} $\eqlevel(E,F)$ of terms $E,F$: we
write $E\sim_{k-1}F$ and $E\not\sim_k F$.
If $E,F$ are bisimilar, i.e. $E\sim F$, then Defender has a winning strategy and we
put $\eqlevel(E,F)=\omega$.
A natural idea is to search for a computable function $f$ attaching 
a number $f(\calG,E,F)\in\Nat$ to terms $E,F$
and a grammar $\calG$ 
so that it is guaranteed that $\eqlevel(E,F)\leq
f(\calG,E,F)$ or $\eqlevel(E,F)=\omega$; this immediately yields an
algorithm that computes  $\eqlevel(E,F)$ (concluding that
$\eqlevel(E,F)=\omega$
when finding that $\eqlevel(E,F)>f(\calG,E,F)$).

We will show such a computable function $f$ by analysing optimal plays
from $E_0\not\sim F_0$; such an optimal play gives rise 
to a sequence $(E_0,F_0)$, $(E_1,F_1)$, $\dots$, $(E_k,F_k)$
of pairs of terms where $\eqlevel(E_i,F_i)=\eqlevel(E_{i-1},F_{i-1})-1$
for $i=1,2,\dots,k$, and  $\eqlevel(E_k,F_k)=0$ (hence
$\eqlevel(E_0,F_0)=k$).
This sequence is then suitably modified to yield a certain sequence
\begin{equation}\label{eq:modseq}
(E'_0,F'_0), (E'_1,F'_1), \dots, (E'_k,F'_k)
\end{equation}
such that 
$(E'_0,F'_0)=(E_0,F_0)$ and $\eqlevel(E'_i,F'_i)=\eqlevel(E_i,F_i)$
for all $i=1,2,\dots,k$; here we use simple congruence properties 
(if $E'$ arises from $E$ by replacing a subterm $H$ with $H'$ such
that $H\sim_k H'$, then $E\sim_k E'$). Doing this modification
carefully, adhering to a sort of ``balancing policy''
(inspired by one crucial ingredient
in~\cite{Senizergues:TCS2001,Seni05},
used also in~\cite{Stirling:TCS2001})
we derive that if $k$ is ``large'', 
then the sequence~(\ref{eq:modseq}) contains a ``long'' subsequence
\begin{equation}\label{eq:extrnsgseq}
(\overline{E}_1\sigma,\overline{F}_1\sigma),(\overline{E}_2\sigma,
\overline{F}_2\sigma),\dots,(\overline{E}_z\sigma,\overline{F}_z\sigma),
\end{equation}
called an
$(n,s,g)$-sequence, where the variables in all ``tops''
$\overline{E}_j$, 
$\overline{F}_j$ are from the set $\{x_1,\dots,x_n\}$, $\sigma$ is the
common ``tail'' substitution (maybe with ``large'' terms $x_i\sigma$), and the size-growth of the
tops is bounded:  $\pressize(\overline{E}_j,\overline{F}_j)\leq s+g\cdot(j{-}1)$ for
$j=1,2,\dots,z$. The numbers $n,s,g$ are elementary in the size of the
grammar $\calG$.
Then another fact is used (whose analogues in
 different frameworks
could be traced back to~\cite{Senizergues:TCS2001,Seni05} and other related
works):
if $\eqlevel(\overline{E}_1, \overline{F}_1)= e<\ell=
\eqlevel(\overline{E}_1\sigma, \overline{F}_1\sigma)$,
then 
there is $i\in\{1,2,\dots,n\}$ and a term
$H\neq x_i$ reachable from $\overline{E}_1$ or 
$\overline{F}_1$ within $e$ moves (i.e. root-rewriting steps) such
that $x_i\sigma\sim_{\ell-e} H\sigma$.
This entails that for $j=e{+}2,e{+}3,\dots,z$ 
the tops
$(\overline{E}_j,\overline{F}_j)$ in~(\ref{eq:extrnsgseq})
can be replaced with
$(\overline{E}_j[x_i/H'],\overline{F}_j[x_i/H'])$,
where $H'$ is the regular term 
$H[x_i/H][x_i/H][x_i/H]\cdots$, 
without changing the equivalence-level;
hence
$\eqlevel(\overline{E}_j\sigma,\overline{F}_j\sigma)=
\eqlevel(\overline{E}_j[x_i/H']\sigma,\overline{F}_j[x_i/H']\sigma)$.
Though $H'$ might be an infinite regular term, its natural graph
presentation is not larger than the presentation of $H$. Moreover,
$x_i$ does not occur in $H'$, and thus the term $x_i\sigma$ ceases to
play any role in the pairs 
$(\overline{E}_j[x_i/H']\sigma,\overline{F}_j[x_i/H']\sigma)$
($j=e{+}2,e{+}3,\dots,z$). 

By continuing this reasoning inductively (``removing'' one $x_i\sigma$
in each of at most $n$ phases), we note that the length of
$(n,s,g)$-sequences~(\ref{eq:extrnsgseq}) is bounded by a (maybe
large) constant determined by the grammar $\calG$. 
By a careful analysis we then show that such a constant is, in fact,
computable when a grammar is given.

\emph{Further remarks on related research.}
Further work is needed to fully understand the bisimulation problems on
PDA and their subclasses, also regarding
their computational complexity.
E.g., even the case 
of BPA processes, generated by real-time PDA with
a single control-state, is not quite clear.
Here the bisimilarity
problem is EXPTIME-hard~\cite{Kiefer13} and in 2-EXPTIME~\cite{DBLP:conf/mfcs/BurkartCS95} 
(proven explicitly in~\cite{Jan12b}); for the subclass of normed BPA
the problem is polynomial~\cite{HiJeMo96}
(see~\cite{CzLa10} for the best published upper bound).
Another issue is the precise decidability border. This was also
studied in~\cite{DBLP:journals/jacm/JancarS08}; 
allowing that $\varepsilon$-steps can have alternatives 
(though they are restricted to be stack-popping)
leads to undecidability of bisimilarity. 
This aspect has been also refined, for 
 branching bisimilarity~\cite{DBLP:journals/corr/YinFHHT14}.
For second-order PDA the undecidability is established without 
$\varepsilon$-steps~\cite{DBLP:conf/fsttcs/BroadbentG12}.
We can refer to the survey
papers~\cite{DBLP:conf/lics/Ong15,DBLP:journals/siglog/Walukiewicz16}
for the work on higher-order PDA, and in particular 
mention that the decidability of
equivalence of \emph{deterministic} higher-order PDA remains open; 
 some progress in this direction was made by Stirling 
 in~\cite{DBLP:conf/concur/Stirling06}.

Finally we remark that recently (while this paper was under review)
the author cooperated with Sylvain Schmitz on developing a concrete
version of the algorithm suggested here, and its complexity analysis
has revealed an Ackermannian upper bound; with the lower bound 
from~\cite{DBLP:conf/fossacs/Jancar14} this yields the Ackermann-completeness of the studied equivalence
problem~\cite{DBLP:conf/lics/JancarS19}.

\emph{Organization of the paper.}
After the preliminaries in Section~\ref{sec:prelim} we state the main
theorem in Section~\ref{sec:maintheorem}. The theorem is proven in
Section~\ref{sec:finalproof}, using the notions and results discussed
in Sections~\ref{sec:boundnsg},~\ref{sec:balance},
and~\ref{sec:analysis}; each of these sections starts with an informal
summary.

\section{Basic Notions and Facts}\label{sec:prelim}

In this section we define basic notions and observe their simple
properties.
Some standard definitions are restricted 
when we do not need full generality.

By $\Nat$ and $\Natpos$ we denote the 
sets
of nonnegative integers and of positive integers, respectively.
By $[i,j]$, for $i,j\in\Nat$, we denote the set $\{i,i{+}1,\dots,j\}$.
For a set $\calA$, by $\calA^*$ we denote the set of finite
sequences of elements of $\calA$, which are also called \emph{words}
(over $\calA$).
By $|w|$ we denote the
\emph{length} of 
$w\in \calA^*$, and 
by $\varepsilon$ the \emph{empty sequence};
hence $|\varepsilon|=0$. We put
$\calA^+=\calA^*\smallsetminus\{\varepsilon\}$.

\subparagraph*{Labelled transition systems.}
A \emph{labelled transition system}, an \emph{LTS} for short,
is a tuple 
$\calL=(\calS,\act,(\gt{a})_{a\in{\act}})$
where $\calS$ is a finite or countable
set of \emph{states},
$\act$ is a finite or countable
set of \emph{actions} 
and $\gt{a}\mathop{\subseteq} \calS\times\calS$ is a set of
\emph{$a$-transitions} (for each $a\in\act$). 
We say that $\calL$ is 
a \emph{deterministic LTS} if for each pair 
$s\in\calS$, $a\in\act$ there is 
at most one $s'$ such that $s\gt{a}s'$ (which stands for
$(s,s')\mathop{\in}\gt{a}$).
By $s\gt{w}s'$, where 
$w=a_1a_2\dots a_n\in
\act^*$,
we denote 
that there is a \emph{path}
$s=s_0\gt{a_1}s_1\gt{a_2}s_2\cdots\gt{a_n}s_n=s'$;
the length of such a path is $n$, which is zero for the  
(trivial) path $s\gt{\varepsilon} s$.
If $s\gt{w}s'$, then  
$s'$ is \emph{reachable from} $s$. By $s\gt{w}$ we denote that $w$ is
\emph{enabled in} $s$, or $w$ is \emph{performable from} $s$,
i.e., $s\gt{w}s'$ for some $s'$.
If $\calL$ is deterministic, then  the expressions 
$s\gt{w}s'$ and $s\gt{w}$ also denote a unique path.
 
\subparagraph*{Bisimilarity, eq-levels.}
Given $\calL=(\calS,\act,(\gt{a})_{a\in\act})$, 
a \emph{set} $\calD\subseteq \calS\times\calS$  
\emph{covers} 
$(s,t)\in  \calS\times\calS$ if 
for any $s\gt{a}s'$ there is $t\gt{a}t'$ such that 
$(s',t')\in \calD$, and for any   $t\gt{a}t'$ there is $s\gt{a}s'$
such that 
$(s',t')\in \calD$.
For $\calD, \calD'\subseteq \calS\times\calS$
we say that $\calD'$ \emph{covers} $\calD$ if $\calD'$
covers each $(s,t)\in \calD$.
A set $\calD\subseteq \calS\times\calS$
is a \emph{bisimulation} if $\calD$ covers $\calD$.
States $s,t\in\calS$ are \emph{bisimilar},
written $s\sim t$, if there is a bisimulation
$\calD$ containing $(s,t)$. 
A standard 
fact is that 
$\sim\,\subseteq \calS\times\calS$ is an equivalence relation,
and it is the largest
bisimulation, namely the union of all bisimulations.

We also
put $\sim_0=\calS\times\calS$, and define 
$\sim_{k+1}\subseteq\calS\times\calS$
(for $k\in\Nat$)
as the set of pairs 
covered by $\sim_{k}$. 
It is obvious that $\sim_k$ are equivalence relations, and
that $\sim_0\,\supseteq\,
\sim_{1}\,\supseteq\,\sim_2\,\supseteq\,\cdots\cdots\supseteq \sim$.
For the (first limit) ordinal $\omega$ we put 
$s\sim_\omega t$ if $s\sim_k t$ for all $k\in\Nat$; hence 
$\sim_\omega=\bigcap_{k\in\Nat}\sim_k$.
We will only consider  \emph{image-finite} LTSs, where 
 the set $\{s'\mid s\gt{a}s'\}$ is finite for each 
 pair $s\in\calS$, $a\in\act$.
In this case 
 $\bigcap_{k\in\Nat}\sim_k$ is a bisimulation 
 (for each $(s,t)\in \bigcap_{k\in\Nat}\sim_k$ and $s\gt{a}s'$,
 in the finite set $\{t'\mid t\gt{a}t'\}$
 there must be one  $t'$ such that
$s'\sim_k t'$ for infinitely many $k$, which entails
$(s',t')\in \bigcap_{k\in\Nat}\sim_k$),
 and thus $\sim\,=\bigcap_{k\in\Nat}\sim_k\,=\,\sim_\omega$.

To each pair of states $s,t$ 
we attach their \emph{equivalence level} (eq-level): 
\begin{center}
$\eqlevel(s,t)=\max\,\{k\in\Nat\cup\{\omega\}\mid s\sim_k t\}$.
\end{center}
Hence $\eqlevel(s,t)=0$ iff $\{a\in\act\mid s\gt{a}\}\neq 
\{a\in\act\mid t\gt{a}\}$ (i.e., $s$ and $t$ enable different sets of
actions).
The next proposition captures a few additional simple facts;
we should add that we handle $\omega$ as an infinite amount,
stipulating $\omega>n$ and $\omega+n=\omega-n=\omega$ for all $n\in\Nat$.

\begin{proposition}\label{prop:elkeep}
\begin{enumerate}
	\item
		If $\eqlevel(t,t')>\eqlevel(s,t)$, then
	$\eqlevel(s,t)=\eqlevel(s,t')$.
\item
If $\omega>\eqlevel(s,t)>0$, then there is either a transition 
$s\gt{a}s'$ such that  for all transitions $t\gt{a}t'$ we have
$\eqlevel(s',t')\leq \eqlevel(s,t)-1$, or 
 a transition 
$t\gt{a}t'$ such that  for all transitions $s\gt{a}s'$ we have
$\eqlevel(s',t')\leq \eqlevel(s,t)-1$.
\item
	If $|w|\leq \eqlevel(s,t)$ and $s\gt{w}s'$, then
	$t\gt{w}t'$ for $t'$ such that $\eqlevel(s',t')\geq
	\eqlevel(s,t)-|w|$.
\end{enumerate}
\end{proposition}
\begin{proof}
1.	If  $s\sim_k t$, $s\not\sim_{k+1} t$, and $t\sim_{k+1}t'$, then
$s\sim_k t'$ and  $s\not\sim_{k+1} t'$.

The points 2 and 3 trivially follow from the definition of $\sim_k$
(for $k\in\Nat\cup\{\omega\}$).
\end{proof}

\subparagraph{First-order terms, regular terms, finite graph presentations.}
We will consider LTSs 
in which the states
are
first-order regular terms. 

The terms are built from \emph{variables}
taken from a fixed countable set
$$\var=\{x_1,x_2,x_3,\dots\}$$ and from 
\emph{function symbols}, also called \emph{(ranked) nonterminals},
from some specified finite set $\calN$; each $A\in\calN$ has 
$\arity(A)\in\Nat$.
We reserve symbols $A,B,C,D$ to range over nonterminals, and 
$E,F,G,H,T,U,V,W$
to range over 
terms. 
An example of a finite term is $E_1=A(D(x_5,C(x_2,B)),x_5,B)$,
where the arities of nonterminals $A,B,C,D$ are $3,0,2,2$,
respectively.
Its syntactic tree is depicted on the left of Fig.\ref{fig:basicterm}.

\begin{figure}[t]
\centering
\includegraphics[scale=0.4]{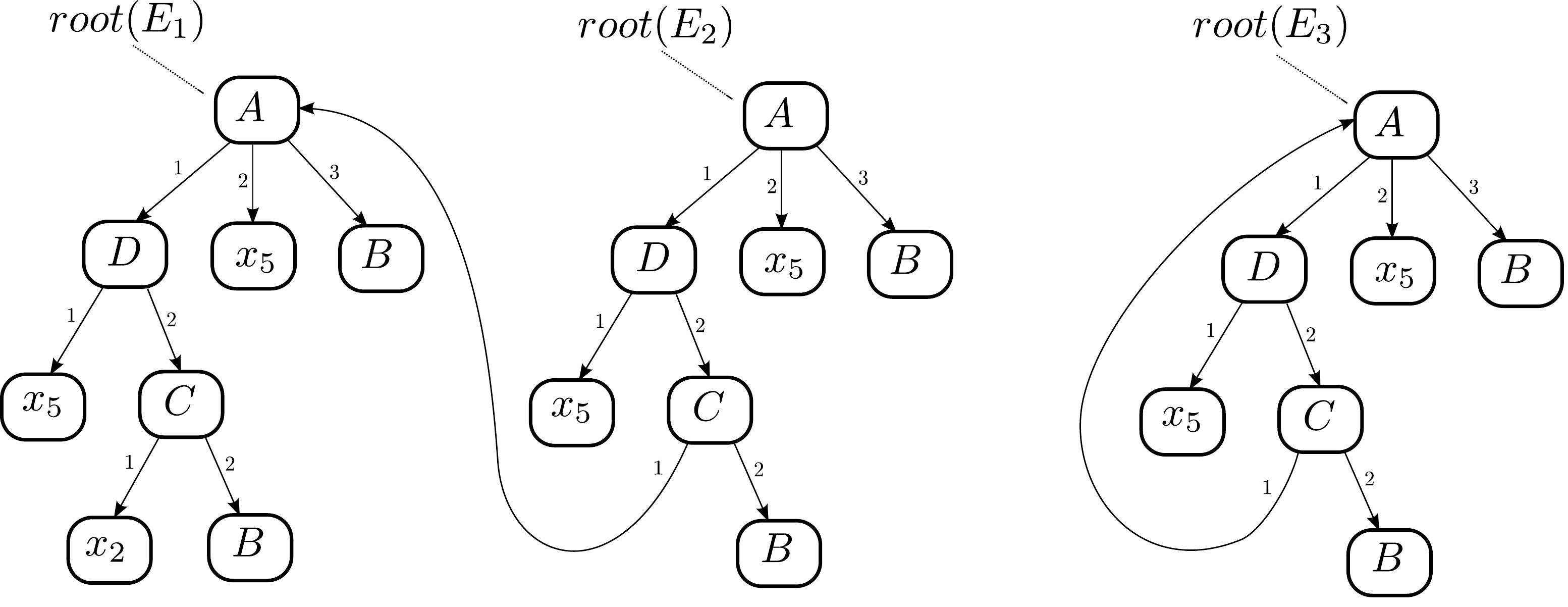}
\caption{Finite terms $E_1$, $E_2$, and 
a graph presenting a regular infinite term $E_3$}\label{fig:basicterm}
\end{figure}

We identify terms with their syntactic trees.
Thus a \emph{term over}
$\calN$ is (viewed as) a rooted, ordered, finite or infinite tree where each node
has a label from $\calN\cup\var$; if the label of a node is $x\in\var$, 
then the node has no successors, and if the label is $A\in\calN$, then 
it has $m$ (immediate) successor-nodes where $m=\arity(A)$.
A subtree of a term 
$E$ is also called a
\emph{subterm} of $E$. We make no difference between isomorphic
(sub)trees, and thus a subterm can have more (maybe infinitely
many) \emph{occurrences} in $E$.  Each \emph{subterm-occurrence} has
its (nesting) \emph{depth in} $E$, which is its (naturally defined) 
distance from the root of $E$.
E.g., $C(x_2,B)$ is a depth-2 subterm of $E_1$; $x_5$ is a subterm
with a depth-1 and a depth-2 occurrences.

We also use the standard notation for terms: 
we write $E=x_i$ or $E=A(G_1,\dots, G_m)$ with the obvious meaning; in
the latter case $\termroot(E)=A\in\calN$, $m=\arity(A)$, and
$G_1,\dots,G_m$ are the ordered  depth-$1$ occurrences of  subterms of
$E$, which are also called the \emph{root-successors} in $E$.

A \emph{term} is \emph{finite} if the respective tree is finite. 
A (possibly infinite) \emph{term} is \emph{regular}
if it has only finitely many subterms (though the subterms may be infinite and
may have infinitely many occurrences). 
We note that any regular term has at least one \emph{graph
presentation}, i.e. a finite directed graph with a designated root,
where each node has a
label from $\calN\cup\var$; if the label of a node is $x\in\var$,
then the node has no outgoing arcs, if the label is  $A\in\calN$,
then it has $m$
 ordered outgoing arcs where $m=\arity(A)$.
 We can see an example of such a graph presenting a term $E_3$ on the right in
Fig.~\ref{fig:basicterm}. 
The standard tree-unfolding of the
graph is the respective term, which is infinite if there are cycles in
the graph. There is a bijection between 
the nodes in the \emph{least} graph presentation of $E$ 
and (the roots of) the subterms of $E$.

\subparagraph*{Sizes, heights, and variables of terms.}
By $\trees_\calN$ we denote the set of all regular terms over $\calN$
(and $\var$); we do not consider non-regular terms.
By a ``term'' we mean a general regular term unless the context makes
clear that the term
is finite.

By $\pressize(E)$ we mean the number of nodes
in the least graph presentation of $E$. E.g., in
Fig.\ref{fig:basicterm} $\pressize(E_1)=6$ ($E_1$ has six subterms)
and  $\pressize(E_3)=5$. 
By $\pressize(\{E_1,E_2,\dots,E_n\})$ 
we mean the number of nodes
in the least graph presentation 
in which
a distinguished node $\textsc{r}_i$ corresponds to the (root
of the) term $E_i$, for each $i\in[1,n]$.
(Since $E_1,E_2,\dots,E_n$ can share some subterms, $\pressize(\{E_1,E_2,\dots,E_n\})$
can be smaller than
$\sum_{i\in[1,n]}\pressize(E_i)$.)
We usually write $\pressize(E,F)$ instead of $\pressize(\{E,F\})$. 
E.g.,  $\pressize(E_1,E_2)=9$ in Fig.~\ref{fig:basicterm}.

For a finite term $E$ we define $\height(E)$ as the maximal depth of a
subterm; e.g., $\height(E_1)=3$ in Fig.\ref{fig:basicterm}.

We put $\varin(E)=\{x\in\var\mid x$ occurs in $E\}$ and 
 $\varin(E,F)=\{x\in\var\mid x$ occurs in $E$ or $F\}$. 
 E.g., $\varin(E_1,E_2)=\{x_2,x_5\}$ in  Fig.\ref{fig:basicterm}.

\subparagraph*{Substitutions, associative composition, iterated
substitutions.}

A \emph{substitution} $\sigma$ is a mapping
$\sigma:\var\rightarrow\trees_{\calN}$ whose 
\emph{support}
\begin{center}
$\support(\sigma)=\{x\in\var\mid \sigma(x)\neq x\}$ 
\end{center}
is \emph{finite};
we reserve the symbol $\sigma$ for substitutions.
By \emph{applying a substitution} $\sigma$ {to 
a term} $E$ we get the term $E\sigma$ 
that arises from $E$ by replacing each occurrence of $x\in\var$ with
$\sigma(x)$; given graph presentations, 
in the graph of $E$ we just redirect each arc
leading to a node labelled with $x$ towards the root of $\sigma(x)$ (which includes the
special ``root-designating arc'' when $E=x$). Hence $E=x$ implies
$E\sigma=x\sigma=\sigma(x)$.
The natural \emph{composition of substitutions}, where
$\sigma=\sigma_1\sigma_2$ is defined by
$x\sigma=(x\sigma_1)\sigma_2$,
can be easily verified to be
associative. We thus write $E\sigma_1\sigma_2$ instead of 
$(E\sigma_1)\sigma_2$ or  $E(\sigma_1\sigma_2)$. 
For $i\in\Nat$ we define $\sigma^i$ inductively: $\sigma^0$ is the
empty-support substitution, and $\sigma^{i+1}=\sigma\sigma^i$. 

By $[x_{i_1}/H_1, x_{i_2}/H_2,\dots, x_{i_k}/H_k]$, 
where $i_j\neq i_{j'}$ for $j\neq j'$,
we denote the
substitution $\sigma$ such that $x_{i_j}\sigma=H_j$ for all $j\in[1,k]$ 
and $x\sigma=x$ for all
$x\in\var\smallsetminus\{x_{i_1},x_{i_2},\dots,x_{i_k}\}$.
We will use
$\sigma^\omega=\sigma\sigma\sigma\cdots$ just for the special case
$\sigma=[x_i/H]$, where $\sigma^\omega$
is clearly well-defined; a graph presentation
of the term $x_i\sigma^\omega$ arises from a graph presentation of $H$
by redirecting each arc leading to $x_i$ (if any exists)
towards the root; we have $x_i\sigma^\omega=H$ if
$x_i\not\in\varin(H)$, or if $H=x_i$. 
In Fig.\ref{fig:basicterm}, for $\sigma=[x_2/E_1]$ we have
$E_2=E_1\sigma$ and $E_3=E_1\sigma^\omega$. 

By $\sigma\remxi$ we denote the substitution arising from $\sigma$ by
removing $x_i$ from its support (if it is there): hence $x_i\sigma\remxi=x_i$ and
 $x\sigma\remxi=x\sigma$ for all $x\in\var\smallsetminus\{x_i\}$.

 We note a trivial fact (for later use):

\begin{proposition}\label{prop:Hremxi}
If $H\neq x_i$, then for the term 
$H'=H[x_i/H][x_i/H][x_i/H]\cdots$
we have 
$x_i\not\in\varin(H')$, and thus
 $H'\sigma=H'\sigma\remxi$ for any $\sigma$.
 We also have $\pressize(H')\leq \pressize(H)$.
\end{proposition}

\subparagraph*{First-order grammars.}
A \emph{first-order grammar}, 
or just a \emph{grammar} for short, is a tuple
$\calG=(\calN,\act,\calR)$ where $\calN$
is a finite nonempty set of 
ranked \emph{nonterminals}, viewed as function symbols with arities, 
$\act$
is a finite nonempty set of \emph{actions} (or ``letters''), 
and $\calR$
is
a finite nonempty set of 
\emph{rules} of the form
\begin{equation}\label{eq:rewrule}
A(x_1,x_2,\dots, x_m)\gt{a} E
\end{equation}
where $A\in \calN$, $\arity(A)=m$, 
$a\in\act$,
and $E$ is a
finite
term over $\calN$ 
in which each occurring variable 
is
from the set $\{x_1,x_2,\dots,x_m\}$; we can have $E=x_i$ for some
$i\in[1,m]$.

\subparagraph{LTSs generated by rules, and by actions, of grammars.}
Given $\calG=(\calN,\act,\calR)$, 
by $\calL^{\ltsrul}_{\calG}$ we denote the (\emph{rule-based}) LTS
$\calL^{\ltsrul}_{\calG}=(\trees_{\calN},\calR,(\gt{r})_{r\in\calR})$
where each rule $r$ of the form
$A(x_1,x_2,\dots, x_m)\gt{a} E$ 
induces
transitions $A(x_1,\dots, x_m)\sigma\gt{r}E\sigma$
for all substitutions $\sigma$.
The transition induced by  $\sigma$
with $\support(\sigma)=\emptyset$ is
$A(x_1,\dots, x_m)\gt{r}E$.

Using terms from Fig.\ref{fig:basicterm} as examples, if a rule $r_1$ 
is $A(x_1,x_2,x_3)\gt{b}x_2$, then we have $E_3\gt{r_1}x_5$
(since $E_3$ can be written as $A(x_1,x_2,x_3)\sigma$ where
$x_2\sigma=x_5$); the action $b$ only plays a role 
 in the LTS $\calL^{\ltsact}_{\calG}$ defined below
 (where we have $E_3\gt{b}x_5$).
For a rule $r_2: A(x_1,x_2,x_3)\gt{a}C(x_2,D(x_2,x_1))$ we deduce
$E_1\gt{r_2}C(x_5,D(x_5,D(x_5,C(x_2,B))))$; we note that the
third root-successor in $E_1$ thus ``disappears'' since
$x_3\not\in\varin(C(x_2,D(x_2,x_1)))$.

By definition, the LTS $\calL^{\ltsrul}_{\calG}$ is deterministic
(for each $F$ and $r$ there is at most one $H$ such that $F\gt{r}H$).
We note that \emph{variables} are \emph{dead} (have no outgoing
transitions).
We also note  
that $F\gt{w}H$ implies $\varin(H)\subseteq\varin(F)$
(each variable occurring in $H$ also occurs in $F$)
but not $\varin(F)\subseteq\varin(H)$ in general.

\medskip

\emph{Remark.}
Since the rhs (right-hand sides) $E$ in the rules~(\ref{eq:rewrule}) 
are finite, all
terms reachable from a finite term are finite.
The ``finite-rhs version'' with general regular terms in LTSs has been
chosen for technical convenience. 
This is not crucial, since the equivalence problem for 
the ``regular-rhs version''
can be easily reduced to the problem for our
finite-rhs version.

\medskip

The deterministic rule-based LTS $\calL^{\ltsrul}_{\calG}$ is helpful
technically,
but we are primarily interested in the (image-finite nondeterministic)
\emph{action-based} LTS 
$\calL^{\ltsact}_{\calG}=(\trees_{\calN},\act,(\gt{a})_{a\in\act})$
where each rule $A(x_1,\dots, x_m)\gt{a}E$
induces  the transitions
$A(x_1,\dots, x_m)\sigma\gt{a}E\sigma$ for all substitutions $\sigma$.
(Hence the rules $r_1$ and $r_2$ in the 
above examples induce $E_3\gt{b}x_5$ and  
$E_1\gt{a}C(x_5,D(x_5,D(x_5,C(x_2,B))))$.)

Fig.\ref{fig:pathinLTSG} sketches a path in some LTS 
$\calL^{\ltsrul}_\calG$ 
where we have, e.g., $r_1:A(x_1,x_2,x_3)\gt{a_1}C(D(x_2,x_3),x_3)$
and $r_2:C(x_1,x_2)\gt{a_2}x_2$ for some actions $a_1,a_2$ (which
would replace $r_1,r_2$ in the  LTS 
$\calL^{\ltsact}_\calG$). In the rectangle just a part of a regular-term 
presentation is sketched. Hence the initial root-node $A$ might be
accessible from later roots due to its possible undepicted 
ingoing arcs. On the other hand, the
root-node
$D$ after the steps $r_1r_2r_3$ is not accessible (and can be omitted)
in the presentation of the final term.

\begin{figure}[t]
\centering
\includegraphics[scale=0.53]{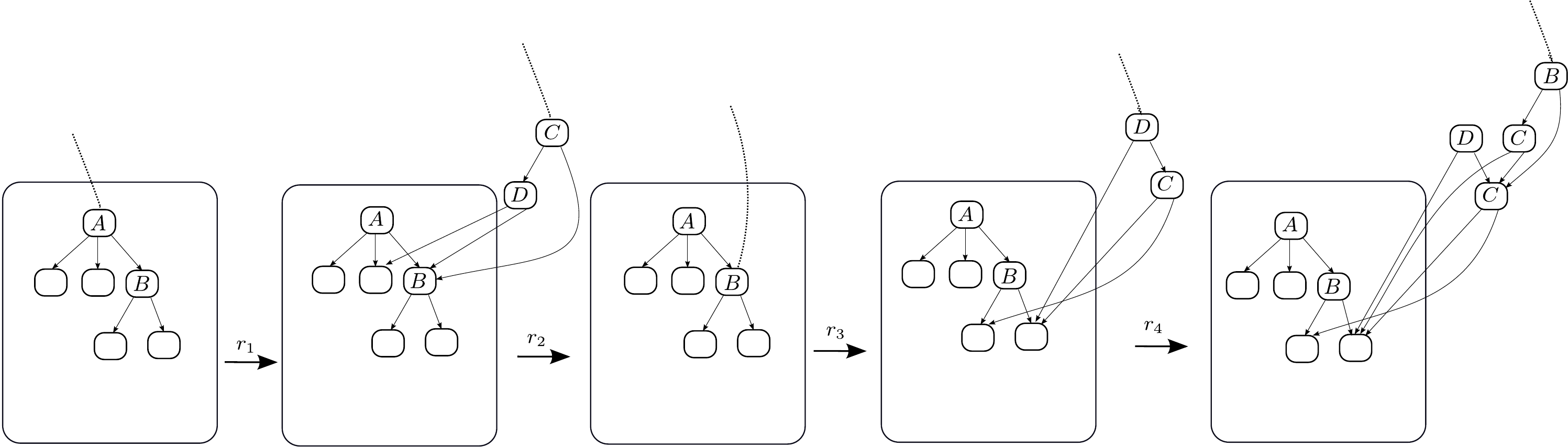}
\caption{Path $A(T_1,T_2,B(T_3,T_4))\gt{r_1}C(D(T_2,B(T_3,T_4)),B(T_3,T_4))
	\gt{r_2}B(T_3,T_4)\gt{r_3r_4}$ in $\calL^{\ltsrul}_\calG$ 
}\label{fig:pathinLTSG}
\end{figure}

\subparagraph{Eq-levels of pairs of terms.}
Given a grammar $\calG=(\calN,\act,\calR)$,
by $\eqlevel(E,F)$ we refer to the equivalence level of (regular) terms $E,F$
in $\calL^{\ltsact}_{\calG}$, with the following adjustment:
though variables $x_i$ are handled as dead also in
$\calL^{\ltsact}_{\calG}$, we stipulate $\eqlevel(x_i,H)=0$ if
$H\neq x_i$ (while $\eqlevel(x_i,x_i)=\omega$);
this would be achieved automatically if 
we enriched $\calL^{\ltsact}_{\calG}$ with transitions
$x\gt{a_x}x$ where $a_x$ is a special action added
to each variable $x\in\var$.
This adjustment gives us the point $1$ in the next
proposition on compositionality.

We put $\sigma\sim_k\sigma'$ if $x\sigma\sim_k x\sigma'$ for all
$x\in\var$, and 
define
\begin{center}
$\eqlevel(\sigma,\sigma')=\max\big\{k\in\Nat\cup\{\omega\}\mid
\sigma\sim_k\sigma'\big\}$.
\end{center}
\begin{proposition}\label{prop:congruence}
	For all $\sigma,\sigma',\sigma'', E,F$, and
	$k\in\Nat\cup\{\omega\}$ the following conditions hold:
	
\begin{enumerate}
	\item
If $\sigma'\sim_k \sigma''$, then $\sigma'\sigma\sim_k
		\sigma''\sigma$. 
		Hence	
		$\eqlevel(\sigma',\sigma'')\leq\eqlevel(\sigma'\sigma,\sigma''\sigma)$.
		
In particular, $\eqlevel(E,F)\leq\eqlevel(E\sigma,F\sigma)$.
	\item
		If $\sigma'\sim_k\sigma''$, then 
	$\sigma\sigma'\sim_k	\sigma\sigma''$.
		Hence
		$\eqlevel(\sigma',\sigma'')\leq\eqlevel(\sigma\sigma',\sigma\sigma'')$.

		In particular, $\eqlevel(\sigma',\sigma'')\leq
\eqlevel(E\sigma',E\sigma'')$.
\end{enumerate}
\end{proposition}
\begin{proof}
	It suffices to prove the claims for $k\in\Nat$, since
	$\sim_{\omega}=\bigcap_{k\in\Nat}\sim_k$. We use an
	induction on $k$, noting that for 
	$k=0$ the claims are trivial. 
	
	Assuming $k>0$
and $E\sim_k F$, we show that $E\sigma\sim_k F\sigma$:
We cannot have $\{E,F\}=\{x_i,H\}$ for some $H\neq x_i$ (since then 
$\eqlevel(E,F)=0$ by our definition).
Hence either
$E=F=x$ for some $x\in\var$, in which case $E\sigma=F\sigma$,
or $E\not\in\var$ and $F\not\in\var$. In the latter case  every
transition $E\sigma\gt{a}G$ ($F\sigma\gt{a}G$) 
is, in fact, $E\sigma\gt{a}E'\sigma$ ($F\sigma\gt{a}F'\sigma$) where
$E\gt{a}E'$ ($F\gt{a}F'$),
and there must be a corresponding transition $F\gt{a}F'$ ($E\gt{a}E'$)
such that $E'\sim_{k-1}F'$
(by Proposition~\ref{prop:elkeep}(3)); by the induction hypothesis 
$E'\sigma\sim_{k-1}F'\sigma$, which shows that $E\sigma\sim_k
F\sigma$ (since $(E\sigma,F\sigma)$ is covered by $\sim_{k-1}$). 

This gives us the point $1$. For the point $2$ we note that 
$\sigma'\sim_k\sigma''$ implies $E\sigma'\sim_k E\sigma''$, which is
even more straightforward to verify. 
\end{proof}

The next lemma shows a simple but important fact (whose analogues in
 different frameworks
could be traced back to~\cite{Senizergues:TCS2001,Seni05} and other related
works). Its claim is sketched in a~part of Figure~\ref{fig:removeone}.
(We recall that $E,F$ denote general regular terms when we do not say
that they are finite.)

\begin{lemma}\label{lem:getequation}
If $\eqlevel(E, F)= k<\ell=
\eqlevel(E\sigma, F\sigma)$,
then 
there are $x_i\in\support(\sigma)$, $H\neq x_i$,
and $w\in\act^*$, $|w|\leq k$, such that
$E\gt{w}x_i$, $F\gt{w}H$ or $E\gt{w}H$, $F\gt{w}x_i$, and
$x_i\sigma\sim_{\ell-k} H\sigma$.
\end{lemma}

\begin{proof}
We assume $\eqlevel(E, F)= k<\ell= \eqlevel(E\sigma, F\sigma)$
and use an induction on $k$.
If $k=0$, then necessarily
$\{E,F\}=\{x_i,H\}$ for some $x_i\neq H$ (since 
 $E\not\in\var$, $F\not\in\var$ would 
imply
$\eqlevel(E\sigma,F\sigma)=0$ as well); the claim is thus trivial
(if $x_i\not\in\support(\sigma)$, i.e. $x_i\sigma=x_i$,
then $H=x_j$ and $x_j\sigma=x_i$,
which entails that $x_j\in\support(\sigma)$).

For $k>0$ we must have $E\not\in\var$, $F\not\in\var$.
There must be a transition $E\gt{a}E'$ (or $F\gt{a}F'$) such that for
all $F\gt{a}F'$ (for all $E\gt{a}E'$) we have 
$\eqlevel(E',F')\leq k{-}1$ (by Proposition~\ref{prop:elkeep}(2)). 
On the other hand, for each
$E\sigma\gt{a}G_1$ (and each $F\sigma\gt{a}G_2$) there is 
$F\sigma\gt{a}G_2$ ($E\sigma\gt{a}G_1$) such that 
$\eqlevel(G_1,G_2)\geq \ell{-}1$ 
(by Proposition~\ref{prop:elkeep}(3)); since 
$E\not\in\var$ and $F\not\in\var$, the transitions
$E\sigma\gt{a}G_1$, $F\sigma\gt{a}G_2$ can be written
$E\sigma\gt{a}E'\sigma$, $F\sigma\gt{a}F'\sigma$, respectively,
where $E\gt{a}E'$, $F\gt{a}F'$.
Hence  there is a pair of transitions
$E\gt{a}E'$, $F\gt{a}F'$ such that $\eqlevel(E',F')=k'\leq k{-}1$ and
$\eqlevel(E'\sigma,F'\sigma)=\ell'\geq \ell{-}1$. 
We apply the induction
hypothesis and deduce that there are 
$x_i\in\support(\sigma)$, $H\neq x_i$,
and $w\in\act^*$, $|w|\leq k'$, such that
$E'\gt{w}x_i$, $F'\gt{w}H$ or $E'\gt{w}H$, $F'\gt{w}x_i$, and
$x_i\sigma\sim_{\ell'-k'} H\sigma$, which entails
$x_i\sigma\sim_{\ell-k} H\sigma$ (since $\ell-k=(\ell-1)-(k-1)\leq
\ell'-k'$). Since 
$E\gt{aw}x_i$, $F\gt{aw}H$ or $E\gt{aw}H$, $F\gt{aw}x_i$,
we are done.
\end{proof}

\subparagraph{Bounded growth of sizes and heights.}
We fix a grammar
$\calG=(\calN,\act,\calR)$,
and note a few simple facts to aid later analysis;
we also introduce the constants
$\stepinc$ (size increase), $\hinc$ (height increase)
related to $\calG$.
We recall that the rhs-terms $E$ in the rules~(\ref{eq:rewrule}) are
finite, and we put 
\begin{equation}\label{eq:hinc}
	\textnormal{
$\hinc=\max\big\{\height(E){-}1\mid  E$ is the rhs of a rule
in $\calR  \big\}$.
}
\end{equation}
We add that in this paper we stipulate $\max\emptyset = 0$.

By $\propsize(E)$ we mean the number of nonterminal nodes in the
least graph presentation of $E$ (hence the number of non-variable subterms of $E$).
We put
\begin{equation}\label{eq:stepinc}
\textnormal{
$\stepinc=\max\big\{\propsize(E)\mid  E$ is the rhs of a rule
in $\calR  \big\}$.
}
\end{equation}

The next proposition 
shows (generous) upper bounds on the size
and height increase
caused by (sets of) transition sequences.
(It is helpful to recall Fig.~\ref{fig:pathinLTSG}, assuming that the
rectangle contains a presentation of $G$.)

\begin{proposition}\label{prop:sizeinc}
\hfill
	\begin{enumerate}
		\item			
	If $G\gt{w}F$, then  $\pressize(F)\leq
	\pressize(G)+|w|\cdot\stepinc$.
\item	
If $G\gt{w}F$ where $G$ is a finite term, then $\height(F)\leq
	\height(G)+|w|\cdot\hinc$.
\item
	If $G\gt{v_1}F_1$, 
	$G\gt{v_2}F_2$, $\cdots$, $G\gt{v_p}F_p$,
	where $|v_i|\leq d$ for all $i\in[1,p]$, then 
	$\pressize(\{F_1,F_2,\dots,F_p\})\leq
	\pressize(G)+p\cdot d\cdot\stepinc$.
\end{enumerate}
\end{proposition}
\begin{proof}
The points $1$ and $2$ are immediate. 
A ``blind'' use of $1$ in the point $3$ would yield 
	$\pressize(\{F_1,F_2,\dots,F_p\})\leq
	p\cdot\big(\pressize(G)+d\cdot\stepinc\big)$. But since 
	the terms $F_i$ can share subterms of $G$, 
	we get the
	stronger bound  $\pressize(G)+p\cdot d\cdot\stepinc$.
\end{proof}

\subparagraph{Shortest sink words.}
If $A(x_1,\dots,x_{\arity(A)})\gt{w}x_i$ in $\calL^\ltsrul_\calG$
(hence $w\in\calR^+$), then we call
$w$ an  \emph{$(A,i)$-sink word}.
We note that such $w$ can be written $rw'$ where 
$A(x_1,\dots,x_{\arity(A)})\gt{r}E\gt{w'}x_i$; hence $w'$ ``sinks'' along a branch
of $E$ to $x_i$, or $w'=\varepsilon$ when $E=x_i$.
This 
suggests a standard dynamic programming
approach to find and fix some shortest $(A,i)$-sink words $w_{[A,i]}$ 
for all elements $(A,i)$ of the set 
$\textsc{NA}=\{(B,j)\mid B\in\calN, j\in[1,\arity(B)]\}$ 
 for which such words exist.
 We can clearly (generously) bound the lengths of $w_{[A,i]}$ by 
 $h^{|\textsc{NA}|}$ where 
 $h=2+\hinc$ (i.e., $h=1+\max\big\{\height(E)\mid E$ is the
 rhs of a rule in $\calR \big\}$).
We put 
\begin{equation}\label{eq:Mzero}
	\textnormal{	
$d_0=1+\max\,\{\,|w_{[A,i]}|; A\in\calN, i\in
[1,\arity(A)]\,\}$.
}
\end{equation}
The above discussion entails that $d_0$ is a (quickly) computable number,
whose value is at most exponential in the size of
the given grammar $\calG$.

\medskip

\emph{Remark.}
For any grammar $\calG$ we can construct a
``normalized'' grammar $\calG'$ in which  
$w_{[A,i]}$ exists for each
$(A,i)\in\textsc{NA}$, while the LTSs $\calL^{\ltsact}_\calG$ and 
$\calL^{\ltsact}_{\calG'}$ are isomorphic. 
(We can refer to~\cite{DBLP:journals/jcss/Jancar20} for more details.) 
We do not need such normalization in this paper.

\medskip

\noindent
\textbf{Convention.}
When having a fixed grammar $\calG=(\calN,\act,\calR)$, 
we also put  
\begin{equation}\label{eq:marity}
m=\max\big\{\arity(A)\mid A\in\calN\big\}
\end{equation}
but we will often write $A(x_1,\dots,x_m)$
even if $\arity(A)$ might not be maximal. This is harmless
since such $m$ could be always replaced with $\arity(A)$
if we wanted to be pedantic. 
(In fact, the grammar could be also normalized
so that the arities of nonterminals are 
the same~\cite{DBLP:journals/jcss/Jancar20} 
	but this is a superfluous technical
issue here.)

\section{Main Result (Computability of Equivalence
Levels)}\label{sec:maintheorem}

\subparagraph*{Small numbers.}
We use the notion of ``small''
numbers determined by a grammar $\calG$;
by saying that a \emph{number} $d\in\Nat$ is \emph{small} 
we mean that it is a computable number (for a given grammar $\calG$)
that is elementary  in the size of $\calG$.

E.g., the numbers $m$, $\hinc$, $\stepinc$
(defined by~(\ref{eq:marity}),~(\ref{eq:hinc}),~(\ref{eq:stepinc}))
are trivially small, and we have also shown that
$d_0$ (defined by~(\ref{eq:Mzero})) is small. 
In what follows we will also introduce further specific small numbers,
summarized in Table~\ref{tab:constants} at the end of the paper.

\subparagraph*{Main theorem.}

We first note a fact that is obvious
(by induction on $k$):

\begin{proposition}\label{prop:stratdecid}
There is an algorithm that, given a grammar $\calG$, terms $T,U$, and 
$k\in\Nat$, decides if $T\sim_k U$
in the LTS $\calL^{\ltsact}_\calG$. 
\end{proposition}
Hence 
the next theorem adds the decidability of 
$\sim$ (i.e., of $\sim_{k}$ for $k=\omega$).

\begin{theorem}\label{th:computingelbound}
For any grammar $\calG=(\calN,\act,\calR)$
there is a small number $c$ and a computable (not
necessarily small) number $\calE$ such that for all
$T,U\in\trees_\calN$ we have:
\begin{equation}\label{eq:elbound}
	\textnormal{
	if $T\not\sim U$ then $\eqlevel(T,U)\leq c\cdot\big(\calE\cdot
\pressize(T,U)+(\pressize(T,U))^2\big)$.
}
\end{equation}
\end{theorem}

\begin{corollary}\label{cor:decidability}
It is decidable, given $\calG$, $T$, $U$, if $T\sim U$ in 
$\calL^{\ltsact}_\calG$.
\end{corollary}

Theorem~\ref{th:computingelbound} is proven in
Section~\ref{sec:finalproof}; the proof uses the notions and 
results from
Sections~\ref{sec:boundnsg},~\ref{sec:balance},
and~\ref{sec:analysis}.
Each section starts with an informal summary, and the collection
of these summaries yields a more detailed informal overview of the proof
than that given in the introduction.

\begin{figure}[t]
\centering
\includegraphics[scale=0.6]{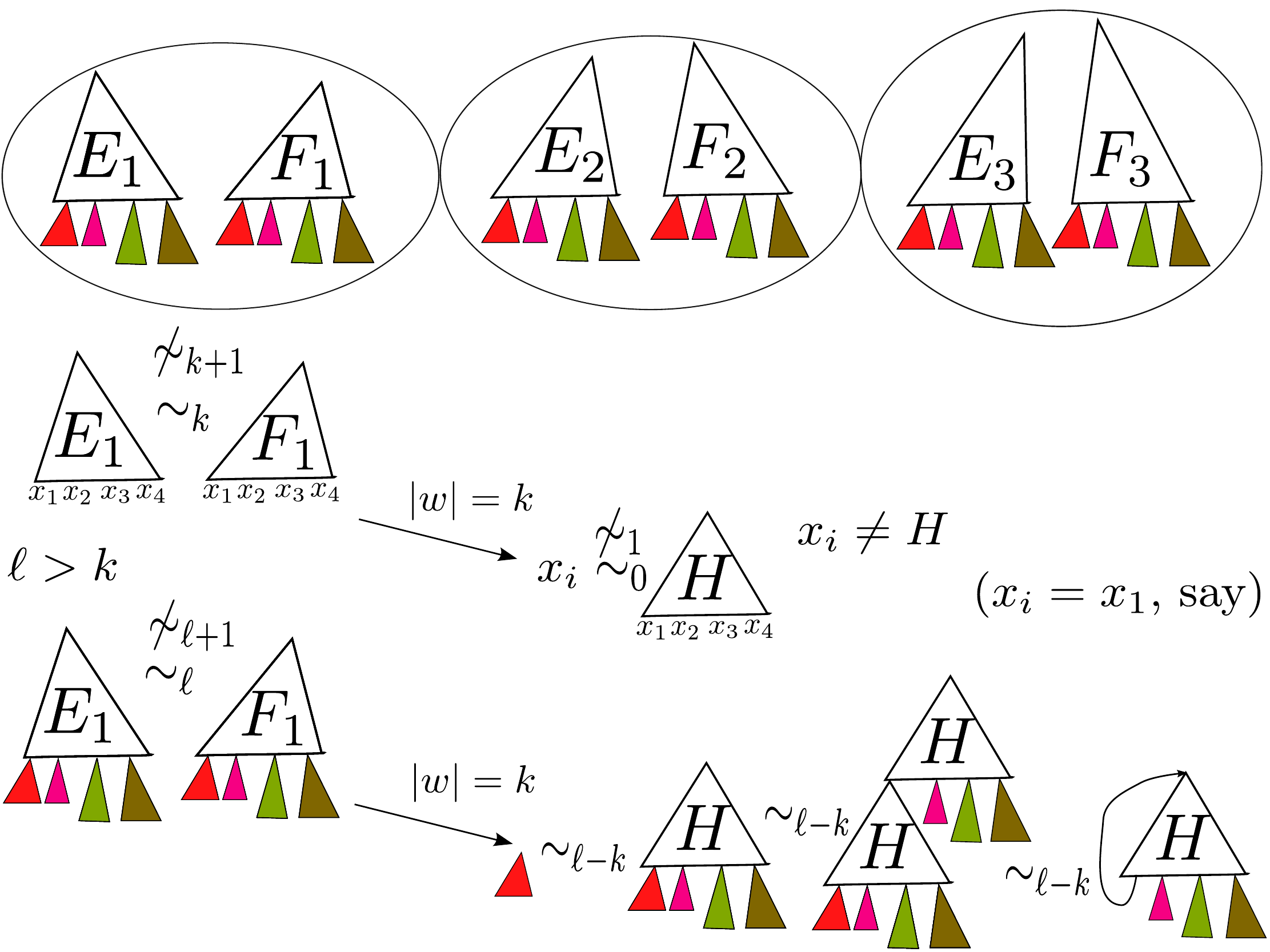}
	\caption{In an $(n,s,g)$-sequence, $(E_1,F_1)$ helps to get
	rid of one term in the tail-substitution $\sigma$}\label{fig:removeone}
\end{figure}

\begin{figure}[t]
\centering
\includegraphics[scale=0.6]{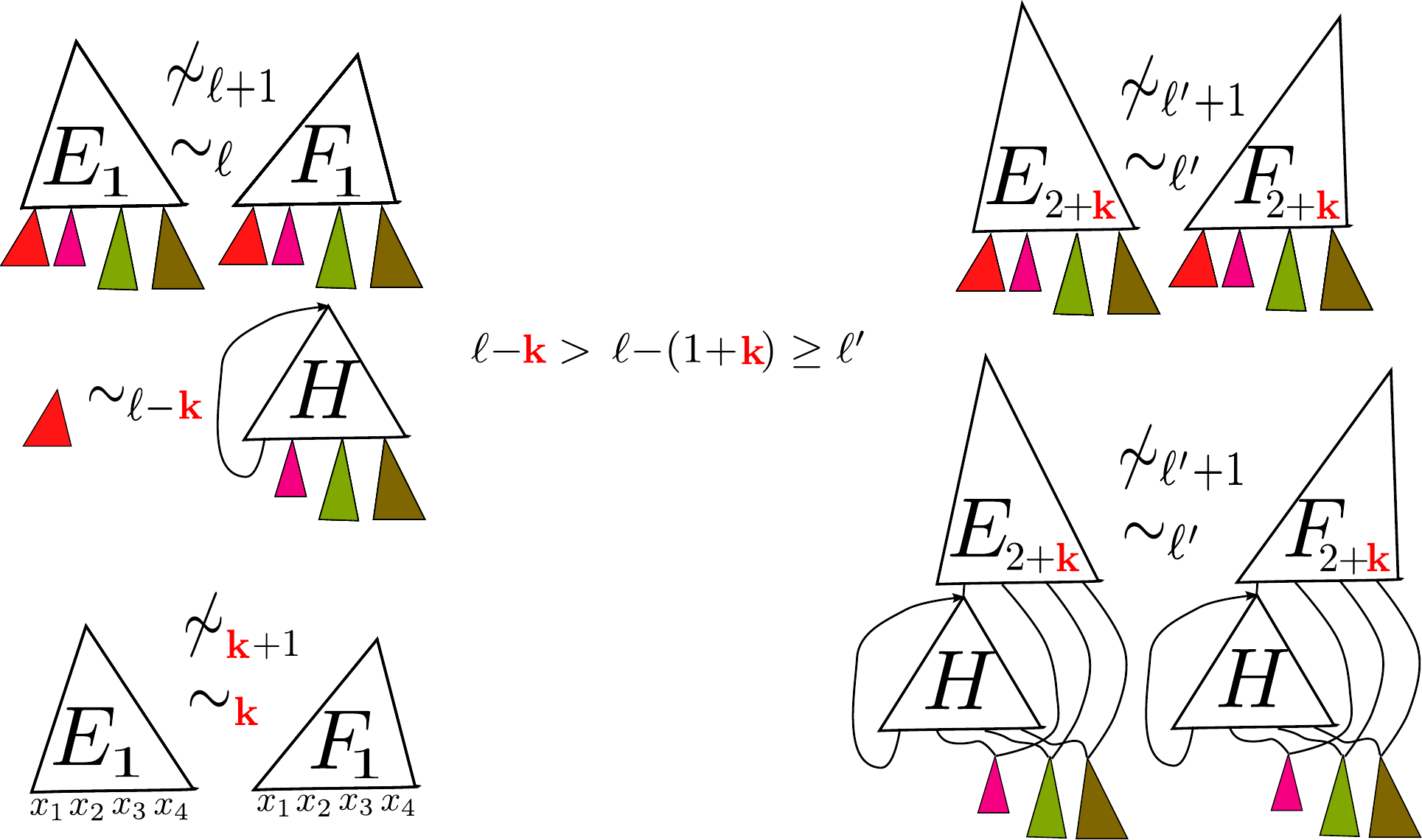}
\caption{Support of $\sigma$ can be safely decreased after the
	eq-level drops sufficiently}\label{fig:boundonnonsink}
\end{figure}

\section{Bounding the Lengths of ``$(n,s,g)$-sequences''}\label{sec:boundnsg}

The top of Figure~\ref{fig:removeone} depicts (a prefix of) a sequence of the
form
\begin{center}
$(E_1\sigma,F_1\sigma)$, $(E_2\sigma,F_2\sigma)$, $\dots$,
$(E_z\sigma,F_z\sigma)$
\end{center}
($E_i,F_i$ being regular terms)
where we assume that the eq-levels are finite and
decreasing:
\begin{center}
$\omega>\eqlevel(E_1\sigma,F_1\sigma)>\eqlevel(E_2\sigma,F_2\sigma)>\cdots>
	\eqlevel(E_z\sigma,F_z\sigma)$. 
\end{center}
	We then have $\eqlevel(E_1,F_1)=k\leq\ell=\eqlevel(E_1\sigma,F_1\sigma)$
(for some $k,\ell\in\Nat$),
by Proposition~\ref{prop:congruence}. If $\sigma$ is the empty-support
substitution, then $k=\ell$ and the sequence length $z$ is bounded by
$1+k$. If $k<\ell$, then Lemma~\ref{lem:getequation} yields some $x_i$
and
$H\neq x_i$ ($x_i=x_1$ in Figure~\ref{fig:removeone})
where $x_i\sigma\sim_{\ell-k}H\sigma$; hence 
in each pair $(E_{j}\sigma,F_{j}\sigma)$ where
$j\in[2+k,z]$ we can (repeatedly) replace
$x_i\sigma$ with $H\sigma$ without changing 
the eq-level
 of the pair. 
This is depicted in Figure~\ref{fig:boundonnonsink};
since $\eqlevel(E_{2+k}\sigma,F_{2+k}\sigma)=\ell'<\ell-k$, the
respective eq-levels do not change 
due to Propositions~\ref{prop:congruence} and~\ref{prop:elkeep}.

If, moreover, we are guaranteed that the size growth of $(E_j,F_j)$ is
controlled, i.e., 
$$\pressize(E_j,F_j)\leq s+g\cdot(j{-}1)$$
for some
fixed constants $s$ and $g$ (and $j\in[1,z]$),
and $\varin(E_j,F_j)\subseteq \{x_1,\dots,x_n\}$
for some fixed $n$ (which bounds the support of $\sigma$), then 
a bound on the lengths $z$ of such
$(n,s,g)$-sequences
is determined
by the respective grammar $\calG$ (independently of the sizes of terms
$x_i\sigma$). This is straightforward, as we now show.

Given $n,s,g$, the number of respective pairs
$(E_1,F_1)$ is bounded, and there is thus $e\in\Nat$ that is the largest
$\eqlevel(E_1,F_1)$ for such pairs (we recall that
$\omega>\eqlevel(E_1\sigma,F_1\sigma)\geq\eqlevel(E_1,F_1)$); at this
moment we do
not claim that $e$ is computable.
For each $(n,s,g)$-sequence $(E_1\sigma,F_1\sigma)$,
$(E_2\sigma,F_2\sigma)$, $\dots$,
$(E_z\sigma,F_z\sigma)$  where $z>1{+}e$ we have
($\eqlevel(E_1,F_1)<\eqlevel(E_1\sigma,F_1\sigma)$ and)
either $E_1\gt{w}x_i$ and  $F_1\gt{w}H$, or $E_1\gt{w}H$ and
$F_1\gt{w}x_i$, $|w|\leq e$, for the
respective $x_i,H$ discussed above and illustrated in
Figures~\ref{fig:removeone} and~\ref{fig:boundonnonsink}; hence 
$\pressize(H)\leq\pressize(E_1,F_1)+e\cdot\stepinc$ (by
Proposition~\ref{prop:sizeinc}(1)).
This entails that replacing $x_i\sigma$ with $H'\sigma\remxi$ 
where $H'=H[x_i/H][x_i/H][x_i/H]\cdots$ (recall
Proposition~\ref{prop:Hremxi}) in the pairs
 $(E_{j}\sigma,F_{j}\sigma)$ for
$j=1{+}e{+}1,1{+}e{+}2,\dots,z$ gives us 
an $(n{-}1,s',g)$-sequence of length $z{-}(1{+}e)$, where 
$s'=s+g\cdot(1{+}e)+s+e\cdot\stepinc$ (which bounds the size of
terms $E_{2+e}, F_{2+e}$ extended by a shared subterm $H'$).
To be precise, 
for the terms $E'_j=E_j[x_i/H']$ and $F'_j=F_j[x_i/H']$ we only have
$\varin(E'_j,F'_j)\subseteq\{x_1,\dots,x_{n}\}\smallsetminus\{x_i\}$,
and $x_n$ can occur in them (when $x_i\neq x_n$).
In this case we just replace $x_n$ with $x_i$ 
in all $E'_j, F'_j$ ($j=1{+}e{+}1,1{+}e{+}2,\dots,z$) and use 
the tail-substitution $\sigma'$ that arises
from $\sigma$ by putting $x_i\sigma'=x_n\sigma$ 
and $x_n\sigma'=x_n$. An inductive argument thus establishes 
that there is indeed a claimed bound (on the lengths of
$(n,s,g)$-sequences) determined by the grammar.

We will later show that such a bound is even computable
when $\calG,n,s,g$ are given. Moreover, we will also show how to compute small
$n,s,g$ to a given $\calG$ so that the computable bound on the length
of $(n,s,g)$-sequences gives us the number $\calE$ in 
Theorem~\ref{th:computingelbound}. 

In the rest of this section
we formalize the above ideas showing that $(n,s,g)$-sequences are bounded.
In this formalization
 we also define the notion of 
\emph{$(n,s,g)$-candidates}, candidates for ``non-equivalence
bases''; 
intuitively, the \emph{base} $\calB_{n,s,g}$ is intended to collect all possible ``tops''
$(E_j,F_j)$, $(E_j[x_i/H'],F_j[x_i/H'])$, $\dots$
from all (eqlevel-decreasing) $(n,s,g)$-sequences
that 
undergo the above described inductive transformation.

\subparagraph*{Eqlevel-decreasing $(n,s,g)$-sequences.}

We fix a grammar $\calG=(\calN,\act,\calR)$.
By an \emph{eqlevel-decreasing sequence} we mean 
a sequence $(T_1,U_1),(T_2,U_2),\dots,(T_z,U_z)$
of pairs of terms 
(where $z\in\Natpos$) such that 
$\omega>\eqlevel(T_1,U_1)>\eqlevel(T_2,U_2)>\cdots
>\eqlevel(T_z,U_z)$.
The length $z$ of such a sequence is obviously at most
$1+\eqlevel(T_1,U_1)$.

For $n,s,g\in\Nat$
we say that an eqlevel-decreasing sequence
in the form
\begin{equation}\label{eq:nsgseq}
(E_1\sigma,F_1\sigma),(E_2\sigma,F_2\sigma),\dots,(E_z\sigma,F_z\sigma)
\end{equation}
is an \emph{$(n,s,g)$-sequence}  if 
$\varin(E_j,F_j)\subseteq \{x_1,\dots,x_n\}$  and
$\pressize(E_j,F_j)\leq s+g\cdot(j{-}1)$
for all $j\in[1,z]$. (The size of ``tops'' $(E_j,F_j)$ is at most 
$s$ at the start,
and $g$ bounds the ``growth-rate'' of tops; the terms $x_i\sigma$,
$i\in[1,n]$, 
might be large but the  ``tail substitution'' $\sigma$ 
is the same in all elements of the sequence.)

\subparagraph*{Candidates for (non-equivalence) bases.}

To show a bound on the lengths of $(n,s,g)$-sequences 
in a convenient form (in Lemma~\ref{lem:ELdecreasbound}), we introduce further
notions; we start with a piece of notation.
For any $n,s\in\Nat$ we put
\begin{itemize}
	\item		
$\pairsvar{n}=\big\{(E,F)\in\trees_\calN\times\trees_\calN\mid
	\varin(E,F)=\{x_1,\dots,x_n\}\big\}$,
\item
	$\pairssize{s}=\big\{(E,F)\in\trees_\calN\times\trees_\calN\mid
	\pressize(E,F)\leq s\big\}$,
\item
	$\pairsvs{n}{s}=\pairsvar{n}\cap\pairssize{s}$.
\end{itemize}

Given $n,s,g\in\Nat$, 
we say that $\calB\subseteq\trees_\calN\times\trees_\calN$ is 
an \emph{$(n,s,g)$-candidate} (intended to collect the tops of
$(n,s,g,)$-sequences that undergo the above described inductive transformation)
if the following
conditions $1$--$3$ hold (in which an implicit induction
on $n$ is used):

\begin{enumerate}
\item
$\calB\subseteq \big(\pairsvar{0}\cup
\pairsvar{1}\cup\cdots\cup\pairsvar{n}\big)\mathop{\cap}\not\sim$\,.
\item	
	$(\calB\cap\pairsvar{n})\subseteq\pairssize{s}$.
\item
If $n>0$, then 
the set 
$\calB'=\calB\smallsetminus \pairsvar{n}$ 
is an $(n{-}1,s',g)$-candidate
where 
\begin{equation}\label{eq:nextsize}
\textnormal{
$s'=2s+g\cdot(1{+}e)+e\cdot\stepinc$ for
$e=\max \big\{\eqlevel(E,F)\mid (E,F)\in\calB\cap\pairssize{s}\}$.
}
\end{equation}
\end{enumerate}

Every $(n,s,g)$-candidate $\calB$ 
yields a \emph{bound} 
$\calE^{n,s,g}_{\calB}\in\Natpos$, denoted just $\calE_\calB$ when 
 $n,s,g$ are clear from the context; 
 in the above notation (around~(\ref{eq:nextsize})) 
 we define 
 $\calE^{n,s,g}_\calB$ 
 as follows:
 \begin{center}
	 if $n=0$, then $\calE^{n,s,g}_{\calB}=1+e$;
if  $n>0$, then
$\calE^{n,s,g}_{\calB}=1+e+\calE^{n-1,s',g}_{\calB'}$.
\end{center}

An \emph{$(n,s,g)$-candidate} $\calB$ is \emph{full below} an
\emph{eq-level} $\overline{e}\in\Nat\cup\{\omega\}$ 
if 
each pair $(E,F)\in \big(\pairsvar{0}\cup
\pairsvar{1}\cup\cdots\cup\pairsvar{n}\big)\mathop{\cap}\pairssize{s}$
such that $\eqlevel(E,F)<\overline{e}$ belongs to $\calB$,
and, moreover, in the case $n>0$ the $(n{-}1,s',g)$-candidate	
$\calB'$ is full below $\overline{e}$.
We say that $\calB$ is \emph{full} if it is full below $\omega$
(in which case $\calB$ contains \emph{all} relevant non-equivalent
pairs).

\begin{proposition}\label{prop:nsgcand}
For  any $n,s,g$ there is the unique full $(n,s,g)$-candidate, denoted
 $\calB_{n,s,g}$.
\end{proposition}	
\begin{proof}
Given $n,s,g$, the full
$(n,s,g)$-candidate $\calB=\calB_{n,s,g}$ is defined as follows:
$\calB\cap\pairssize{s}=\big(\pairsvar{0}\cup
\pairsvar{1}\cup\cdots\cup\pairsvar{n}\big)\mathop{\cap}\pairssize{s}
\mathop{\cap}\not\sim$
 and, 
moreover, in the case
$n>0$
the set $\calB'=\calB\smallsetminus\pairsvs{n}{s}$ is the full $(n{-}1,s',g)$-candidate
(where $s'$ is defined as in~(\ref{eq:nextsize})).
\end{proof}

The unique full $(n,s,g)$-candidate $\calB_{n,s,g}$ will be also
called the \emph{$(n,s,g)$-base}.

\subparagraph*{The $(n,s,g)$-sequences have bounded lengths.}

We show the announced bound.

\begin{lemma}\label{lem:ELdecreasbound}
If $(E_1\sigma,F_1\sigma),(E_2\sigma,F_2\sigma),\dots,(E_z\sigma,F_z\sigma)$
is an $(n,s,g)$-sequence and $\calB$ is an $(n,s,g)$-candidate that is
full below $1+\eqlevel(E_1\sigma,F_1\sigma)$, then $z\leq
\calE_\calB$;
in particular, $z\leq \calE_{\calB_{n,s,g}}$.
\end{lemma}

\begin{proof}
We consider an $(n,s,g)$-sequence
$(E_1\sigma,F_1\sigma),(E_2\sigma,F_2\sigma),\dots,(E_z\sigma,F_z\sigma)$
as in~(\ref{eq:nsgseq}), and   an $(n,s,g)$-candidate $\calB$ that is
full below $1+\eqlevel(E_1\sigma,F_1\sigma)$.
Since 
$\omega>\eqlevel(E_1\sigma,F_1\sigma)\geq\eqlevel(E_1,F_1)$ 
(by Proposition~\ref{prop:congruence}(1)), 
we have 
$(E_1,F_1)\in\calB\cap\pairssize{s}$.
This entails 	
\begin{center}
$\eqlevel(E_1,F_1)=k\leq 
e=\max \big\{\eqlevel(E,F)\mid (E,F)\in\calB\cap\pairssize{s}\}$.
\end{center}
	If $\eqlevel(E_1\sigma,F_1\sigma)=\eqlevel(E_1,F_1)=k$, which is surely
the case when $n=0$
(in this case $(E_1\sigma,F_1\sigma)=(E_1,F_1)$),
then $z\leq 1+k$, due to the required eqlevel-decreasing
property of $(n,s,g)$-sequences;
in this case $z\leq 1+e\leq\calE_\calB$.

We proceed inductively (on $n$), assuming  
 $n>0$ and $\eqlevel(E_1\sigma,F_1\sigma)=\ell>k=\eqlevel(E_1,F_1)$.
 By Lemma~\ref{lem:getequation} there is $x_i\in\support(\sigma)$,
 $i\in[1,n]$, and $H\neq x_i$
such that $E_1\gt{w}x_i$ and $F_1\gt{w}H$, or $E_1\gt{w}H$ and
$F_1\gt{w}x_i$, for some $w\in\act^*$ with $|w|\leq k$, where 
	$x_i\sigma\sim_{\ell-k}H\sigma$.
Hence $\sigma\sim_{\ell-k}[x_i/H]\sigma$,
which entails that 
$\sigma\sim_{\ell-k}[x_i/H]^j\sigma$ for all
$j\in\Nat$ (by applying Proposition~\ref{prop:congruence}(2)
repeatedly).
We can also easily check that 
$[x_i/H]^{\ell-k}\sigma\sim_{\ell-k}[x_i/H]^{\omega}\sigma$
(by induction on $\ell-k$), hence
\begin{center}
	$x_i\sigma\sim_{\ell-k}H'\sigma\remxi$ where 
	$H'=H[x_i/H][x_i/H][x_i/H]\cdots$
\end{center}
(We also recall Proposition~\ref{prop:Hremxi}.)
We note that 
\begin{center}
$\pressize(H')\leq\pressize(H)\leq
\max\{\pressize(E_1),\pressize(F_1)\}+k\cdot\stepinc\leq s+e\cdot
\stepinc$
\end{center}
(by using Proposition~\ref{prop:sizeinc}(1)).
For each $j\in[k{+}2,z]$ we now put
\begin{center}
$(E'_j,F'_j)=(E_j[x_i/H'],F_j[x_i/H'])$,
hence $(E'_j\sigma,F'_j\sigma)=(E'_j\sigma\remxi, F'_j\sigma\remxi)$,
\end{center}
and note that
$\eqlevel(E_j\sigma,F_j\sigma)=
\eqlevel(E'_j\sigma\remxi, F'_j\sigma\remxi)$,
since $\eqlevel(E_j\sigma,F_j\sigma)<\ell{-}k$ (for each $j\geq
k{+}2$); here we use that  $\eqlevel(E_j\sigma,E_j[x_i/H']\sigma)\geq
\ell{-}k$ and $\eqlevel(F_j\sigma,F_j[x_i/H']\sigma)\geq
\ell{-}k$, and we recall Proposition~\ref{prop:elkeep}(1).
We also note that for each  $j\in[k{+}2,z]$ we have
\begin{center}
$\pressize(E'_j,F'_j)\leq \pressize(E_j,F_j)+\pressize(H')\leq
s+g\cdot(j{-}1)+s+e\cdot\stepinc=2s+g\cdot(j{-}1)+e\cdot\stepinc$.
\end{center}
Hence $\pressize(E'_{k+2},F'_{k+2})\leq
2s+g\cdot(1+k)+e\cdot\stepinc\leq 2s+g\cdot(1+e)+e\cdot\stepinc=s'$
(recall $s'$ from~(\ref{eq:nextsize})). 
Thus the sequence 
\begin{center}
$(E'_{k{+}2}\sigma\remxi, F'_{k+2}\sigma\remxi)$, 
$(E'_{k{+}3}\sigma\remxi, F'_{k+3}\sigma\remxi)$$\dots$,
$(E'_{z}\sigma\remxi, F'_{z}\sigma\remxi)$
\end{center}
is ``almost'' an
$(n{-}1,s',g)$-sequence.
The only problem is that $x_n$ can occur in $E'_j,F'_j$. 
But we use the fact that $x_i$ does not occur in $E'_j,F'_j$, and we
replace $x_n$ with $x_i$, while replacing $\sigma\remxi$ 
with $\sigma'$ where $x_n\sigma'=x_n$,
$x_i\sigma'=x_n\sigma\remxi$, and $x\sigma'=x\sigma\remxi$ for all
$x\in\var\smallsetminus\{x_i,x_n\}$.

We note that the $(n{-}1,s',g)$-candidate 
$\calB'=\calB\smallsetminus \pairsvar{n}$ is full below
$1+\eqlevel(E'_{k+2}\sigma\remxi,F'_{k+2}\sigma\remxi)$
(since $\calB'$ is full below $1+\eqlevel(E_1\sigma,F_{1}\sigma)$,
and
$\eqlevel(E'_{k+2}\sigma\remxi,F'_{k+2}\sigma\remxi)
=\eqlevel(E_{k+2}\sigma,F_{k+2}\sigma)
<\eqlevel(E_1\sigma, F_1\sigma)$).
By the induction hypothesis $z{-}(k{+}1)\leq \calE_{\calB'}$,
and thus 
$z\leq 1{+}k{+}\calE_{\calB'}\leq 1{+}e{+}\calE_{\calB'}=
\calE_\calB$.
\end{proof}

In the final argument of the proof of Theorem~\ref{th:computingelbound}
(in Section~\ref{sec:finalproof})
we will use $\calE_{\calB_{n,s,g}}$ as $\calE$
in~(\ref{eq:elbound}), for some specific small $n,s,g$.
Though we have defined the 
$(n,s,g)$-base $\calB_{n,s,g}$ only semantically, it will turn out that it coincides 
with an effectively constructible ``sound'' $(n,s,g)$-candidate.
But we first need some further technicalities to clarify the specific  
$n,s,g$ (as well as $c$ in~(\ref{eq:elbound})).

\section{Plays (of Bisimulation Game) and their
Balancing}\label{sec:balance}

In Section~\ref{sec:intro} we discussed the notion of
optimal plays, which we make more precise now.
We assume a given grammar $\calG=(\calN,\act,\calR)$;
for $r\in\calR$ of the form $A(x_1,\dots,x_m)\gt{a}E$
we put $\lab(r)=a$.
For technical convenience, by a \emph{play} we only mean
an optimal play from a non-equivalent pair, i.e., a sequence
\begin{equation}\label{eq:prelimoptplay}
\textnormal{
$\toppair{T_0}{U_0}\bothgt{r_1}{r'_1}\toppair{T_1}{U_1}
\bothgt{r_2}{r'_2}\toppair{T_2}{U_2}\cdots
\bothgt{r_k}{r'_k}\toppair{T_k}{U_k}$
}
\end{equation}
where for each $i\in[1,k]$ we have $r_i,r'_i\in\calR$, 
$\lab(r_i)=\lab(r'_i)$,
$T_{i-1}\gt{r_i}T_i$, $U_{i-1}\gt{r'_i}U_i$ (in the LTS
$\calL^{\ltsrul}_{\calG}$); moreover, 
$\omega>\eqlevel(T_{0},U_{0})$ and 
$\eqlevel(T_{i},U_{i})=\eqlevel(T_{i-1},U_{i-1})-1$ for  each $i\in[1,k]$
(in the LTS $\calL^{\ltsact}_{\calG}$).
If $\eqlevel(T_k,E_k)=0$, then
it is a \emph{completed play}, in which case $\eqlevel(T_0,U_0)=k$.
(We recall that $T_0,
U_0$ can be regular terms of a large size.)
The length of the play~(\ref{eq:prelimoptplay}) is (defined to be)
$k$, and another presentation of the play is
$\toppair{T_0}{U_0}\bothgt{u}{u'}\toppair{T_k}{U_k}$, or also just
$\toppair{T_0}{U_0}\bothgt{u}{u'}$,
where
$u=r_1r_2\cdots r_k$ and $u'=r'_1r'_2\cdots r'_k$.

Our aim is to bound
the lengths of completed plays in the way stated in Theorem~\ref{th:computingelbound}.  
To facilitate this task,
in this section we show a particular transformation of a completed 
play~(\ref{eq:prelimoptplay}) into a sequence of plays of the same
overall length (i.e., the sum of lengths) that
are connected by so-called \emph{eqlevel-concatenation}
$\econc$; such concatenation 
\[
	\left[\toppair{T}{U}\bothgt{u_1}{u'_1}\toppair{T'}{U'}\right]\econc
	\left[\toppair{T''}{U''}\bothgt{u_2}{u'_2}\toppair{T'''}{U'''}\right]
\]
is defined if (and only if) $\eqlevel(T',U')=\eqlevel(T'',U'')$,
though the pairs $(T',U')$ and $(T'',U'')$ can differ.
The overall length of this concatenation is $|u_1|+|u_2|$; if 
$\toppair{T''}{U''}\bothgt{u_2}{u'_2}\toppair{T'''}{U'''}$ is a
completed play, then this length ($|u_1|+|u_2|$) is obviously the same as the length
of any completed play starting with $(T,U)$. 

In the first phase of the mentioned transformation of a completed
play~(\ref{eq:prelimoptplay}) we will replace it with the concatenation 
of two plays in the
form 
\[
\lbrac\toppair{T_0}{U_0}\bothgt{v_0u_1}{v'_0u'_1}
	\toppair{T'_1}{U'_1}\rbrac\econc
	\lbrac\toppair{T''_1}{U''_1}\bothgt{v}{v'}\rbrac
\]
where  $v_0u_1$ is a certain prefix of $u=r_1r_2\cdots r_k$,
 $v'_0u'_1$ is a prefix of $u'=r'_1r'_2\cdots r'_k$ (of the same
 length as $v_0u_1$),
and $\lbrac\toppair{T''_1}{U''_1}\bothgt{v}{v'}\rbrac$ is a completed play
(while $v$ ($v'$) is generally \emph{not} a suffix of $u$ ($u'$)). Further we
replace $\lbrac\toppair{T''_1}{U''_1}\bothgt{v}{v'}\rbrac$ with 
$\lbrac\toppair{T''_1}{U''_1}\bothgt{v_1u_2}{v'_1u'_2}
	\toppair{T'_2}{U'_2}\rbrac\econc
	\lbrac\toppair{T''_2}{U''_2}\bothgt{\bar{v}}{\bar{v}'}\rbrac$
where  $v_1u_2$ is a certain prefix of $v$,
$v'_1u'_2$ is a prefix of $v'$,
and
$\lbrac\toppair{T''_2}{U''_2}\bothgt{\bar{v}}{\bar{v}'}\rbrac$ is a
completed play;
we continue in this way, doing $\ell$ phases for a certain number
$\ell$, until finally getting
\begin{equation}\label{eq:prelimmodifplay}
\textnormal{
	$\lbrac\toppair{T_0}{U_0}\bothgt{v_0}{v'_0}\toppair{\bar{T}_1}{\bar{U}_1}\bothgt{u_1}{u'_1}
	\toppair{T'_1}{U'_1}\rbrac\econc
	\lbrac\toppair{T''_1}{U''_1}\bothgt{v_1}{v'_1}\toppair{\bar{T}_2}{\bar{U}_2}\bothgt{u_2}{u'_2}\toppair{T'_2}{U'_2}\rbrac\econc
	\lbrac\toppair{T''_2}{U''_2}
\bothgt{v_2}{v'_2}
	\cdots\cdots
	\toppair{\bar{T}_\ell}{\bar{U}_\ell}\bothgt{u_\ell}{u'_\ell}
	\toppair{T'_\ell}{U'_\ell}\rbrac\econc
	\lbrac\toppair{T''_\ell}{U''_\ell}
	\bothgt{v_\ell}{v'_\ell}\toppair{\bar{T}_{\ell+1}}{\bar{U}_{\ell+1}}\rbrac$
}
\end{equation}
where $\lbrac\toppair{T''_\ell}{U''_\ell}
	\bothgt{v_\ell}{v'_\ell}\toppair{\bar{T}_{\ell+1}}{\bar{U}_{\ell+1}}\rbrac$ is
	completed and ``non-transformable'';
	the overall length of~(\ref{eq:prelimmodifplay}) is thus equal to
	$k=\eqlevel(T_0,U_0)$.	
In fact, we have $\ell=0$ when already~(\ref{eq:prelimoptplay}) is
non-transformable; we thus put
$(T''_0,U''_0)=(T_0,U_0)$ for convenience.
(Later we repeat~(\ref{eq:prelimmodifplay}) as~(\ref{eq:modifplay})
without the bars in the notation $\bar{T}_j, \bar{U}_j$; now the bars
are added to avoid the confusion with ${T}_j, {U}_j$
in~(\ref{eq:prelimoptplay}).)

More concretely, we will perform the transformation so that
for each phase $j\in[1,\ell]$ we have $|u_j|=d_0$
(for $d_0$ defined by~(\ref{eq:Mzero})), one of the terms 
 $\bar{T}_j, \bar{U}_j$ is the \emph{pivot} $W_j$, and the 
 pair $(T''_j,U''_j)$ is 
the  \emph{balancing
result}, or the \emph{bal-result} for short, \emph{related to the pivot} $W_j$.

In fact, if $W_j=\bar{U}_j$, then we have $U''_j=U'_j$ (and $T''_j\neq
T'_j$); in this case the $j$-th phase consists in replacing the
completed play 
$\lbrac
\toppair{T''_{j-1}}{U''_{j-1}}\bothgt{v}{v'}\rbrac$ with
the eqlevel-concatenation
$\lbrac
\toppair{T''_{j-1}}{U''_{j-1}}\bothgt{v_{j-1}}{v'_{j-1}}\toppair{\bar{T}_j}{\bar{U}_j}\bothgt{u_j}{u'_j}
	\toppair{T'_j}{U'_j}\rbrac\econc\lbrac
	\toppair{T''_j}{U'_j}\bothgt{\bar{v}}{\bar{v}'}\rbrac$
(where 
$\lbrac
\toppair{T''_{j-1}}{U''_{j-1}}\bothgt{v_{j-1}}{v'_{j-1}}\toppair{\bar{T}_j}{\bar{U}_j}\bothgt{u_j}{u'_j}
	\toppair{T'_j}{U'_j}\rbrac$ is a prefix of 
$\lbrac
\toppair{T''_{j-1}}{U''_{j-1}}\bothgt{v}{v'}\rbrac$);
 this is called a \emph{left balancing step} (the
\emph{left}
term in $(T'_j,U'_j)$ has been replaced with $T''_j$ so that
$\eqlevel(T'_j,U'_j)=\eqlevel(T''_j,U'_j)$).
Similarly,  if $W_j=\bar{T}_j$, then we have $T''_j=T'_j$, and we have
performed a  \emph{right balancing step}, replacing 
$U'_j$ with $U''_j$.

We thus have pivots $W_1, W_2,\dots, W_\ell$, each having its related
bal-result.
Since the sequence 
\begin{center}
$(T''_1,U''_1), (T''_2,U''_2), \dots,
(T''_\ell,U''_\ell)$
\end{center}	
of bal-results is eqlevel-decreasing, 
no pair can repeat in the sequence.

We will ``balance'' in a way that will also yield a \emph{pivot path} 
\begin{equation}\label{eq:prelimpivotpath}
W_0\gt{w_0}W_1\gt{w_1}W_2\gt{w_2}\cdots
W_\ell\gt{w_\ell}W_{\ell+1}
\end{equation}
where $w_0\in\calR^*$, $w_j\in\calR^+$ for $j\in[1,\ell]$,
 $W_0\in\{T_0,U_0\}$, $W_{\ell+1}\in\{\bar{T}_{\ell+1},\bar{U}_{\ell+1}\}$, and 
we will guarantee the following properties: 
\begin{enumerate}
\item
There is some small $n$ such that for each $j\in[1,\ell]$ there are
		small finite terms $G,E,F$, with 
$\varin(E,F)\subseteq \varin(G)\subseteq\{x_1,\dots,x_n\}$,
		such that  
		\begin{center}		
		$W_j=G\sigma$ and
		$(T''_{j},U''_j)=(E\sigma,F\sigma)$, 
		\end{center}			
			for a
		substitution $\sigma$ (with
		$\support(\sigma)\subseteq\{x_1,\dots,x_n\}$).
		(Hence the terms in the bal-result arise from the
		pivot $W$ by replacing a small top of $W$ by other
		small tops.)
This is depicted in
Figure~\ref{fig:twoboundonnonsink} for some $W_{j}$ and $W_{j+1}$ (and
in more detail in Figure~\ref{fig:balstep}).
\item
Each pivot-path segment $W_j\gt{w_j}W_{j+1}$ 
(for $j\in[0,\ell]$)
		is either short (i.e., its length is small), 
or it has a short prefix and a short suffix while the middle part
		is ``quickly sinking'' (to a deep
		subterm of $W_j$ if this part is long). 
\end{enumerate}	
We note that 		
we do not exclude that a pivot $W$ occurs more than once in the
pivot path
($W=W_j$ and $W=W_{j'}$ for $j\neq j'$), but the number of its occurrences
must be small; this follows from the point $1$ which
entails that there is only a small number of possible 
bal-results related to one pivot, and from the fact 
that the bal-results cannot repeat.

Figure~\ref{fig:twoboundonnonsink} depicts a ``non-sinking 
segment'' on the pivot path. (In such a segment, no root-successor of
the starting term is exposed.) By the above point $2$ it is
intuitively clear
that any long non-sinking segment must contain a
large number of pivots, and that the possible increase of (the tops
of) the pivots is controlled. 
Hence any long non-sinking segment of the pivot path gives rise to a
long  $(n,s,g)$-sequence, for some small $n,s,g$; here we use the
point $1$ (and recall Figure~\ref{fig:twoboundonnonsink}).
This is a crucial fact for our proof of
Theorem~\ref{th:computingelbound}.

In this section, our task is to show a transformation that guarantees
a suitable pivot path~(\ref{eq:prelimpivotpath})
and the above properties $1$ and $2$.

A concrete way how we do a \emph{left balancing step} is captured by
Figure~\ref{fig:balstep}.
Informally speaking, if the left-hand side
does not sink to a root-successor
within less than $d_0$ moves 
(for $d_0$ defined by~(\ref{eq:Mzero})),
which is the case 
in Figure~\ref{fig:balstep}
due to $A(x_1,\dots,x_m)\gt{u}E'$, then the other side ($U=G\sigma$
in Figure~\ref{fig:balstep}) can become a pivot, and the bal-result
can be created as depicted; the original root-successors in the
left-hand side are replaced 
by suitable terms that are shortly reachable from the pivot, so that 
the respective eq-level does not change
($\eqlevel(E'\sigma',U')=\eqlevel(E'\sigma'',U')$ in
Figure~\ref{fig:balstep}). 
The existence of such a transformation (we claim nothing about its
effectiveness) is clear by 
Propositions~\ref{prop:elkeep} and~\ref{prop:congruence}.
\emph{Right balancing steps} are analogous (they are elligible when
the right-hand side does not sink within less than $d_0$ moves).

We observe that any path $W\gt{v}W'$ can sink to some depth-$|v|$
subterm of $W$ at most, surely not deeper; hence $W'$ arises from $W$
by replacing its ``$|v|$-top'' with another top; the size of these tops 
is small
when $v$ is short.
This observation now easily entails
the above property $1$, guaranteed by our transformation.

To guarantee a suitable pivot path~(\ref{eq:prelimpivotpath}) and its property $2$,
as a first attempt we consider the following procedure in the $j$-th
phase (of the transformation of~(\ref{eq:prelimoptplay})
into~(\ref{eq:prelimmodifplay})): when we are about to replace the completed play
$\lbrac
\toppair{T''_{j-1}}{U''_{j-1}}\bothgt{v}{v'}\rbrac$, we use its shortest
prefix of the form
$\lbrac
\toppair{T''_{j-1}}{U''_{j-1}}\bothgt{v_{j-1}}{v'_{j-1}}\toppair{\bar{T}_j}{\bar{U}_j}\bothgt{u_j}{u'_j}
	\toppair{T'_j}{U'_j}\rbrac$ where 
$\lbrac
\toppair{\bar{T}_j}{\bar{U}_j}\bothgt{u_j}{u'_j}
	\toppair{T'_j}{U'_j}\rbrac$
	enables a (left or right) balancing step (i.e., some side
	does not sink to a root successor within less than $d_0$
	moves).

But doing this balancing as suggested 
would complicate our task of creating a suitable
pivot path~(\ref{eq:prelimpivotpath}), as we now discuss.
First we note that we can smoothly define $W_0\gt{w_0}W_{1}$:
it is $U_0\gt{v_0}U_1$
if $W_1=U_1$, and $T_0\gt{v_0}T_1$ if $W_1=T_1$.
Similarly we define $W_\ell\gt{w_\ell}W_{\ell+1}$ as 
 $\bar{U}_{\ell}\gt{u'_\ell v'_\ell}\bar{U}_{\ell+1}$
 if $W_\ell=\bar{U}_\ell$, and as 
 $\bar{T}_{\ell}\gt{u_\ell v_\ell}\bar{T}_{\ell+1}$
 if $W_\ell=\bar{T}_\ell$.
If in the consecutive phases $j$ and $j{+}1$ we have the
pivot on the same side, say $W_j=\bar{U}_j$ and
$W_{j+1}=\bar{U}_{j+1}$, then we have no problem either: we define
$W_j\gt{w_j}W_{j+1}$ simply as $\bar{U}_j\gt{u'_jv'_j}\bar{U}_{j+1}$
(which is legal since $U'_j=U''_j$). 

A problem to define $W_j\gt{w_j}W_{j+1}$ arises when there is a switch
of balancing sides. Hence we add a simple condition to be satisfied when
such a switch is allowed to occur.
Suppose $W_j=U_j$, and let
Figure~\ref{fig:balstep} describe the respective left balancing step. 
In the $(j{+}1)$-th phase of the transformation we have
\begin{center}
	$\lbrac\toppair{T_0}{U_0}\bothgt{v_0}{v'_0}\toppair{\bar{T}_1}{\bar{U}_1}\bothgt{u_1}{u'_1}
	\toppair{T'_1}{U'_1}\rbrac\econc\cdots\econc
	\lbrac\toppair{T''_{j-1}}{U''_{j-1}}\bothgt{v_{j-1}}{v'_{j-1}}\toppair{\bar{T}_j}{W_j}\bothgt{u_j}{u'_j}\toppair{E'\sigma'}{U'_j}\rbrac\econc
	\lbrac\toppair{E'\sigma''}{U'_{j}}
\bothgt{v}{v'}\rbrac$
\end{center}
and we are about to replace the (current) completed play 
$\lbrac\toppair{E'\sigma''}{U'_{j}}\bothgt{v}{v'}\rbrac$.
We would prefer to do another left balancing, ideally for a short
prefix of this completed play. This is not possible only if 
the path $E'\sigma''\gt{v}$ is quickly sinking in the beginning, i.e.,
within each segment of length $d_0$ a root-successor of the term
starting the segment is exposed;
the path  $E'\sigma''\gt{v}$ thus has a short prefix 
$E'\sigma''\gt{v_{j1}}x_i\sigma''$ for some $x_i$ (since $E'$ is
a small finite term and the path sinks along one of its branches).
But $x_i\sigma''$ is reachable from the last pivot $W_j$ ($W_j=U_j$)
by a short word $\bar{v}$ (e.g., if $x_i\sigma''=V_2$ in
Figure~\ref{fig:balstep}, then we use the path $U\gt{\bar{v}_2}V_2$).
Hence only after such a short prefix 
$\lbrac\toppair{E'\sigma''}{U'_{j}}\bothgt{v_{j1}}{v'_{j1}}\toppair{x_i\sigma''}{\overline{U}}\rbrac$
we allow to balance on both sides (if a left balancing is not possible
earlier). 

If a switch of balancing sides indeed happens in our discussed case, then we can write the
$j$-th and the $(j{+}1)$-th play in the
sequence~(\ref{eq:prelimmodifplay}) in the form
\begin{center}
$\lbrac\toppair{T''_{j-1}}{U''_{j-1}}\bothgt{v_{j-1}}{v'_{j-1}}\toppair{\bar{T}_j}{W_j}\bothgt{u_j}{u'_j}\toppair{E'\sigma'}{U'_j}\rbrac\econc
	\lbrac\toppair{E'\sigma''}{U'_{j}}
	\bothgt{v_{j1}}{v'_{j1}}\toppair{x_i\sigma''}{\overline{U}}
		\bothgt{v_{j2}}{v'_{j2}}\toppair{W_{j+1}}{\bar{U}_{j+1}}\bothgt{u_{j+1}}{u'_{j+1}}\toppair{T'_{j+1}}{U'_{j+1}}\rbrac$
\end{center}
and define $W_j\gt{w_j}W_{j+1}$ as 
$W_j\gt{\bar{v}}x_i\sigma''\gt{v_{j_2}}W_{j+1}$ (where $\bar{v}$ is 
shorter than the short word $u_jv_{j1}$ but this does not matter).
To summarize: for the consecutive phases $j$ and $j{+}1$ where the
$j$-th phase is a left-balancing step captured by
Figure~\ref{fig:balstep}, we get 
\begin{center}
$\lbrac\toppair{T''_{j-1}}{U''_{j-1}}\bothgt{v_{j-1}}{v'_{j-1}}\toppair{\bar{T}_j}{W_j}\bothgt{u_j}{u'_j}\toppair{E'\sigma'}{U'_j}\rbrac\econc
	\lbrac\toppair{E'\sigma''}{U'_{j}}
	\bothgt{v_{j}}{v'_{j}}\toppair{\bar{T}_{j+1}}{\bar{U}_{j+1}}\bothgt{u_{j+1}}{u'_{j+1}}\toppair{T'_{j+1}}{U'_{j+1}}\rbrac$
\end{center}
where either $v_j$ is short and $W_{j+1}=\bar{U}_{j+1}$ or we can
write $v_j=v_{j1}v_{j2}$ where $v_{j1}$ is short and 
we have $E'\sigma''\gt{v_{j1}}x_i\sigma''$. In the latter case we can
write
\begin{equation}\label{eq:unclearandsink}
\lbrac\toppair{E'\sigma''}{U'_{j}}
	\bothgt{v_{j}}{v'_{j}}\toppair{\bar{T}_{j+1}}{\bar{U}_{j+1}}\bothgt{u_{j+1}}{u'_{j+1}}\toppair{T'_{j+1}}{U'_{j+1}}\rbrac=
\lbrac\toppair{E'\sigma''}{U'_{j}}
	\bothgt{v_{j1}}{v'_{j1}}\toppair{x_i\sigma''}{\overline{U}}
		\bothgt{v_{j2}}{v'_{j2}}\toppair{\bar{T}_{j+1}}{\bar{U}_{j+1}}\bothgt{u_{j+1}}{u'_{j+1}}\toppair{T'_{j+1}}{U'_{j+1}}\rbrac
\end{equation}	
where $W_{j+1}\in\{\bar{T}_{j+1},\bar{U}_{j+1}\}$, and 
both paths $x_i\sigma''\gt{v_{j2}}\bar{T}_{j+1}$ and
$\bar{U}\gt{v'_{j2}}\bar{U}_{j+1}$ (that may be long) are quickly
sinking (since there is no balancing possibility there).

Hence the above property $2$ of the pivot path is also clear
(including the case $W_{j+1}=\bar{U}_{j+1}$, where $W_j\gt{w_j}W_{j+1}$
is $\bar{U}_j\gt{u'_jv'_{j1}v'_{j2}}\bar{U}_{j+1}$).

Now we define the described transformation in a more formal way.

\begin{figure}[t]
\centering
\includegraphics[scale=0.6]{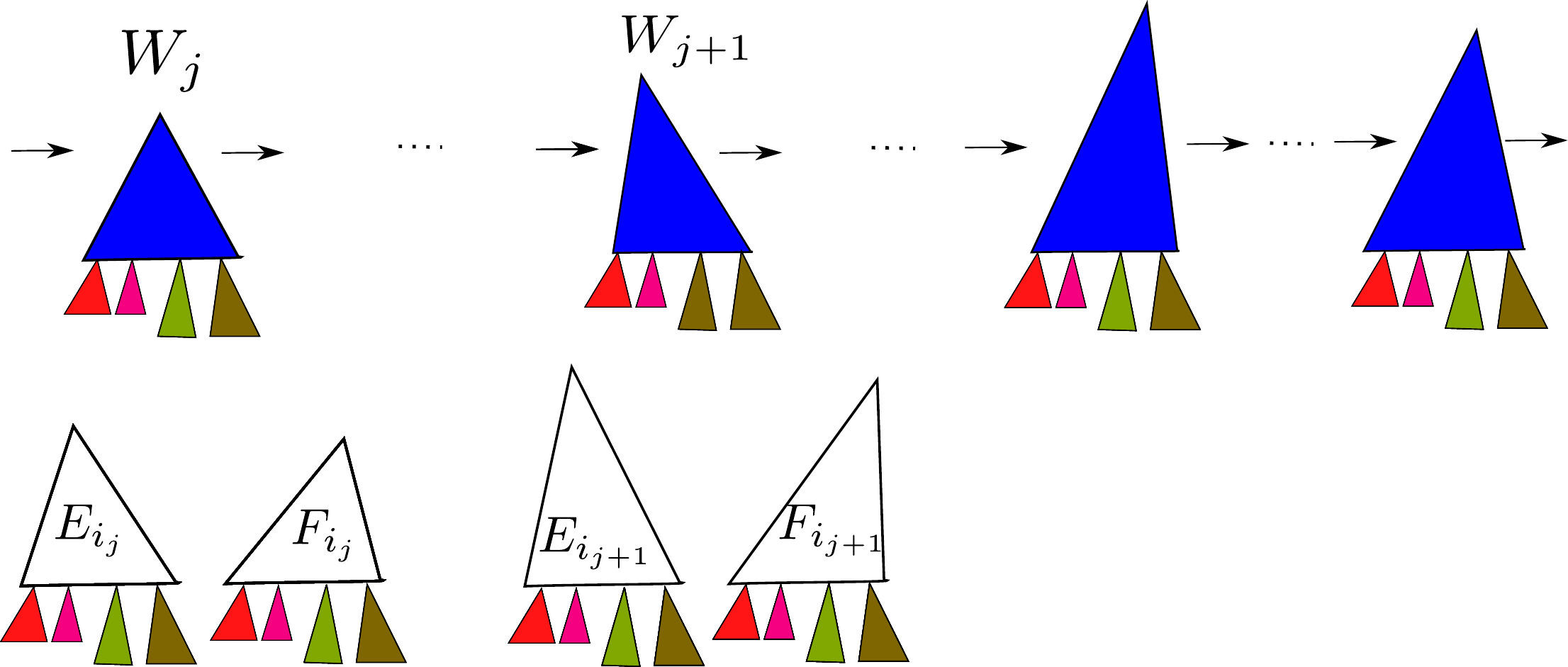}
	\caption{Non-sinking segment on the pivot path gives rise to
	an $(n,s,g)$-sequence}\label{fig:twoboundonnonsink}
\end{figure}

\subparagraph*{Modified optimal plays, and their eqlevel-concatenation.}
We still assume a fixed grammar $\calG=(\calN,\act,\calR)$.
Now we let 
$u,v,w$ (with subscripts etc.) range over $\calR^*$ (not over $\act^*$);
hence $E\gt{w}F$ determines one path in the LTS
$\calL^{\ltsrul}_\calG$. For $r\in\calR$ of the form $A(x_1,\dots,x_m)\gt{a}E$
we put $\lab(r)=a$;
this is extended to the
respective homomorphism $\lab:\calR^*\rightarrow\act^*$.

An \emph{optimal play}, or just a \emph{play} for short,
is a sequence 
\begin{center}
$(T_0,U_0)(r_1,r'_1)(T_1,U_1)(r_2,r'_2)(T_2,U_2)\cdots
(r_k,r'_k)(T_k,U_k)$,
\end{center}
denoted as 
\begin{equation}\label{eq:optplay}
\textnormal{
$\toppair{T_0}{U_0}\bothgt{r_1}{r'_1}\toppair{T_1}{U_1}
\bothgt{r_2}{r'_2}\toppair{T_2}{U_2}\cdots
\bothgt{r_k}{r'_k}\toppair{T_k}{U_k}$\,,
}
\end{equation}
where 
$T_0\not\sim U_0$ and for each $j\in[1,k]$ we have
$T_{j-1}\gt{r_j}T_j$,  $U_{j-1}\gt{r'_j}U_j$,
$\lab(r_j)=\lab(r'_j)$, and $\eqlevel(T_j,U_j)=\eqlevel(T_{j-1},U_{j-1})-1$.
It is clear (by Proposition~\ref{prop:elkeep}(2,3))
that for any $T_0\not\sim U_0$ 
there is a play of the form~(\ref{eq:optplay}) such that 
$k=\eqlevel(T_0,U_0)$ (and $\eqlevel(T_k,U_k)=0$).

A play $\mu$ of the form~(\ref{eq:optplay})
is a \emph{play from
$\startofplay(\mu)=(T_0,U_0)$ to $\finishofplay(\mu)=(T_k,U_k)$}, 
and is also written as $\toppair{T_0}{U_0}\bothgt{u}{u'}\toppair{T_k}{U_k}$,
or just as $\toppair{T_0}{U_0}\bothgt{u}{u'}$,
where
$u=r_1r_2\cdots r_k$ and $u'=r'_1r'_2\cdots r'_k$;
we put 
$\length(\mu)=k$ and $\pairs(\mu)=\{(T_i,U_i)\mid i\in[0,k]\}$.
We also consider the trivial plays of the form $(T_0,U_0)$ with the length $k=0$ (for
$T_0\not\sim U_0$).
A play~(\ref{eq:optplay}) 
is a 
\emph{completed play} if $\eqlevel(T_k,U_k)=0$. 

The \emph{standard concatenation} $\mu\nu$ of plays
$\mu = \toppair{T}{U}\bothgt{u}{u'}\toppair{T'}{U'}$
and $\nu=\toppair{T''}{U''}\bothgt{v}{v'}\toppair{T'''}{U'''}$ is
defined if (and only if) $(T',U')=(T'',U'')$; in this case $\mu\nu$
is the play $\toppair{T}{U}\bothgt{uv}{u'v'}\toppair{T'''}{U'''}$
(hence $\finishofplay(\mu)$ and $\startofplay(\nu)$ get merged).

We aim to show a bound of the form~(\ref{eq:elbound}) 
on the lengths of completed plays from $(T,U)$.
The use of  $(n,s,g)$-sequences,
bounded by Lemma~\ref{lem:ELdecreasbound},
will become clear after we introduce a special modification of plays.
Generally, 
\begin{center}
a \emph{modified play} $\pi$ is a sequence of plays
$\mu_1,\mu_2,\ldots,\mu_\ell$ ($\ell\geq 1$)
\end{center}
where for each $j\in[1,\ell{-}1]$ we have 
$\eqlevel(\finishofplay(\mu_j))=\eqlevel(\startofplay(\mu_{j+1}))$
but $\finishofplay(\mu_j)\neq \startofplay(\mu_{j+1})$;
it is a \emph{modified play from}
$\startofplay(\pi)=\startofplay(\mu_1)$
\emph{to} $\finishofplay(\pi)=\finishofplay(\mu_\ell)$, and it is
a \emph{completed modified play} if
$\eqlevel(\finishofplay(\mu_\ell))=0$.
(As expected, if $\finishofplay(\mu)=(T,U)$, then by 
$\eqlevel(\finishofplay(\mu))$ we refer to the eq-level 
$\eqlevel(T,U)$; similarly in the other cases.)

We put $\length(\pi)=\sum_{j\in[1,\ell]}\length(\mu_j)$, and 
$\pairs(\pi)=\bigcup_{j\in[1,\ell]}\pairs(\mu_j)$.
We do not consider peculiar modified plays 
where $\finishofplay(\mu_j)=\startofplay(\mu_{j+p})$ for $p\geq 2$,
in which case $\mu_{j+1},\mu_{j+2},\cdots,\mu_{j+p-1}$ are
zero-length plays; we implicitly deem the modified plays to be \emph{normalized} 
by (repeated) replacing such segments
$\mu_j,\mu_{j+1},\cdots,\mu_{j+p-1},\mu_{j+p}$
with $\mu_j\mu_{j+p}$. 
E.g., a modified play $\mu_1,\mu_2,\mu_3$ of
the form
$\toppair{T_0}{U_0}\bothgt{u_1}{u'_1}\toppair{T}{U},\toppair{T'}{U'},
\toppair{T}{U}\bothgt{u_2}{u'_2}\toppair{T''}{U''}$ (where
$\eqlevel(T,U)=\eqlevel(T',U')$) is replaced with
$\mu_1\mu_3=\toppair{T_0}{U_0}\bothgt{u_1u_2}{u'_1u'_2}\toppair{T''}{U''}$.

\begin{proposition}\label{prop:samelengthmodplays}
For any $T\not\sim U$ there is a completed play from $(T,U)$,
 and we have 
 $\length(\pi)=\eqlevel(T,U)$ for
each completed modified play $\pi$ from $(T,U)$;
moreover, no pair 
can appear at 
two different positions in $\pi$ (we thus have no repeat of a pair in
$\pi$).
\end{proposition}
\begin{proof}
The eq-levels of pairs in
$\pi=\mu_1,\mu_2,\dots,\mu_\ell$ are dropping in each $\mu_j$;
we have
$\eqlevel(\finishofplay(\mu_j))=\eqlevel(\startofplay(\mu_{j+1}))$
but  $\finishofplay(\mu_j)\neq\startofplay(\mu_{j+p})$ for $p\geq 1$
by definition (which includes the normalization).
\end{proof}

We also define a partial operation on the set of modified plays that
is called 
the  \emph{eqlevel-concatenation} and denoted by $\econc$\,. For
modified plays $\pi=\mu_1,\mu_2,\ldots,\mu_k$ and
$\rho=\nu_1,\nu_2,\ldots,\nu_\ell$, the eqlevel-concatenation $\pi\econc\rho$ is defined if (and only if) 
$\eqlevel(\finishofplay(\pi))=\eqlevel(\startofplay(\rho))$; we recall 
that $\finishofplay(\pi)=\finishofplay(\mu_k)$ and
$\startofplay(\rho)=\startofplay(\nu_1)$.
Suppose that $\pi\econc\rho$, in the above notation, is defined.
If $\finishofplay(\mu_k)\neq \startofplay(\nu_1)$, then 
$\pi\econc\rho= \mu_1,\mu_2,\ldots,\mu_k,
\nu_1,\nu_2,\ldots,\nu_\ell$; if 
$\finishofplay(\mu_k)=\startofplay(\nu_1)$, then
$\pi\econc\rho=
\mu_1,\mu_2,\ldots,\mu_{k-1},\mu_k\nu_1,\nu_2,\nu_3,\ldots,\nu_\ell$.
(We implicitly assume a normalization in the end, if necessary; but 
this will not be needed in our concrete cases.)

We note that the operation $\econc$ is 
associative.

In what follows, by writing the expression 
 $\pi\econc\rho$ for modified plays $\pi,\rho$
 we implicitly claim that $\pi\econc\rho$ is defined (and we refer to the
 resulting modified play $\pi\econc\rho$).
By writing $\pi\rho$ we implicitly claim that  
$\finishofplay(\pi)=\startofplay(\rho)$, and $\pi\rho$ refers to the
modified play $\pi\econc\rho$.

We now show a particular modification of plays, a first step towards
creating $(n,s,g)$-sequences.
In this process we will frequently replace a (sub)play of the type
 $\rho=\toppair{T}{U}\bothgt{u}{u'}\toppair{T'}{U'}$ 
 with a modified play
 $\rho'=\toppair{T}{U}\bothgt{u}{u'}\toppair{T'}{U'}\econc\toppair{T''}{U'}$ 
that has the same length by definition; essentially it means that we
have replaced $T'$ with $T''$ while guaranteeing that 
$\eqlevel(T',U')=\eqlevel(T'',U')$.

\begin{figure}[t]
	\begin{center}
\includegraphics[scale=0.50]{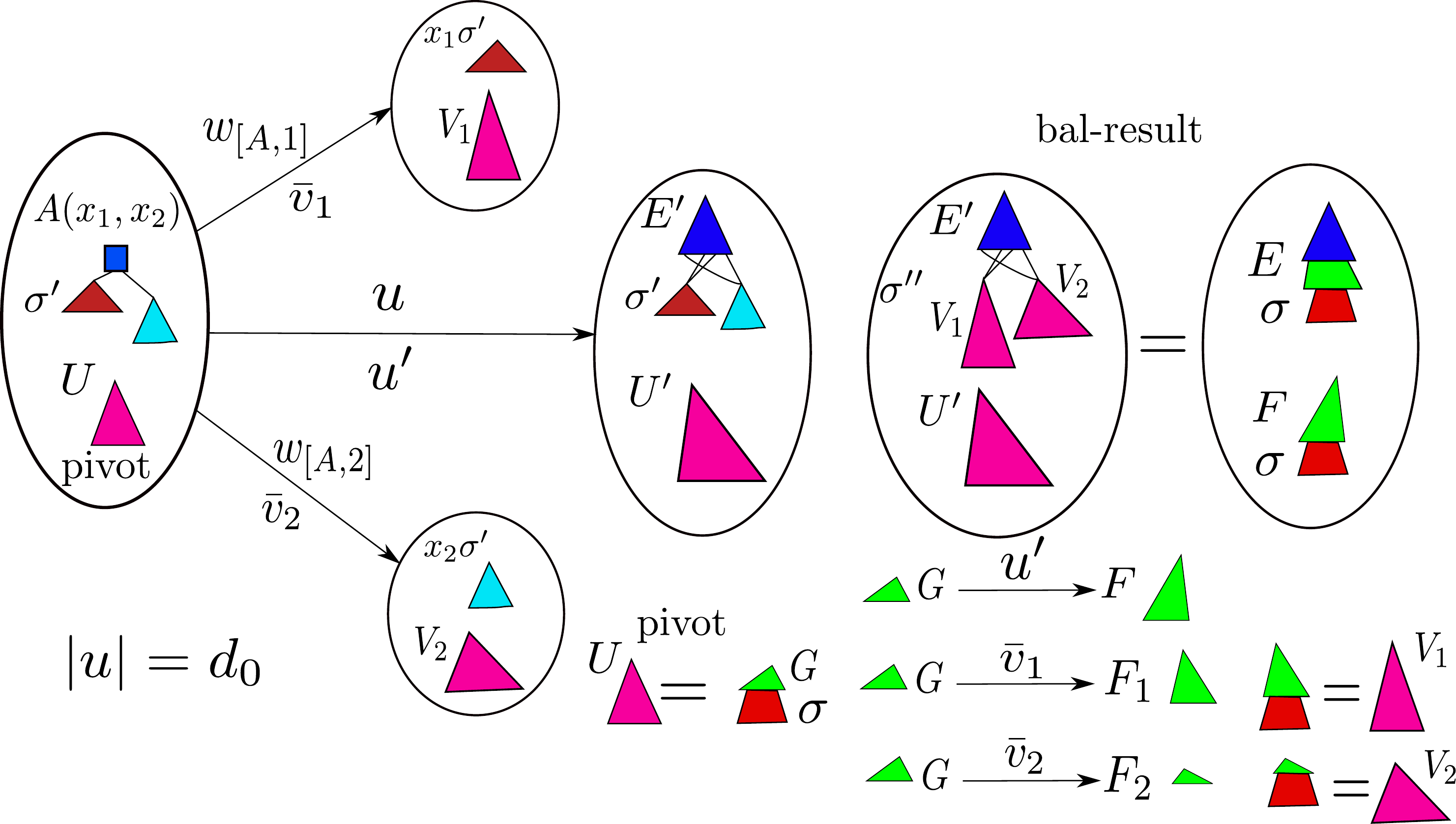}
\end{center}
\hspace*{3em}$\toppair{A(x_1,\dots,x_m)\sigma'}{}\bothgt{w_{[A,i]}}{}\toppair{x_i\sigma'}{}$
{\small (or $w_{[A,i]}$ does not exist)}
	
\hspace*{0.5em}$\rho\,=\,\toppair{A(x_1,\dots,x_m)\sigma'}{U}\hspace*{0.8em}\bothgt{u}{u'}
\hspace*{0.8em}\toppair{E'\sigma'}{U'}$
\hspace*{3em}
$\rho'\,=\,\toppair{A(x_1,\dots,x_m)\sigma'}{U}\bothgt{u}{u'}\toppair{E'\sigma'}{U'}\econc
\toppair{E'\sigma''}{U'}$
\hspace*{1em} {\small where $x_i\sigma''=V_i$}
\vspace*{0.5em}

\hspace*{7em}$\toppair{U}{}\bothgt{\bar{v}_i}{}\toppair{V_i}{}$
{\small ($V_i=U$ when $w_{[A,i]}$ does not exist)}
\caption{Balancing step
$\rho\der_L\rho'$ ($|u|=d_0$, $i$ ranges over $[1,m]$, $\eqlevel(x_i\sigma',V_i)>
\eqlevel(E'\sigma',U')$)}\label{fig:balstep}
\end{figure}

\subparagraph*{Balancing steps, their pivots and balanced results.}
Informally speaking,  
a play $\toppair{T}{U}\bothgt{u}{u'}\toppair{T'}{U'}$
enables left balancing if 
$T\gt{u}T'$ misses the opportunity to sink to a root-successor
as quickly as possible (recall $w_{[A,i]}$ and 
$d_0$ defined around~(\ref{eq:Mzero})). 
A left balancing is illustrated in Figure~\ref{fig:balstep}
(in both a pictorial and a textual form).
We start with a simple example, 
and only then we give a formal 
definition. 

Let us consider a play of the form
\begin{center}
$\toppair{T=A(G_1,G_2)}
{U}\bothgt{r_1}{r_1'}\toppair{B(C(G_2,G_1),G_1)}{U_1}\bothgt{r_2}{r_2'}
\toppair{B'(G_1,C(G_2,G_1))=T'}{U'}
$
\end{center}
where $r_1$ is $A(x_1,x_2)\gt{a_1}B(C(x_2,x_1),x_1)$,
and $r_2$ is  $B(x_1,x_2)\gt{a_2}B'(x_2,x_1)$.
Let $r_3$ be $A(x_1,x_2)\gt{a_3}x_1$, hence we also have 
$A(G_1,G_2)\gt{a_3}G_1$. (Therefore the path $T\gt{r_1r_2}T'$ clearly 
missed the
opportunity to sink to $G_1$ as quickly as possible.)
Since $T\gt{a_3}G_1$,
there must be a transition $U\gt{a_3}V_1$, generated by a rule $r'_3$, such that 
 $\eqlevel(G_1,V_1)\geq \eqlevel(T,U)-1$ 
 (by Proposition~\ref{prop:elkeep}(3)); hence 
 $\eqlevel(G_1,V_1)>\eqlevel(T',U')$
 (since $\eqlevel(T',U')=\eqlevel(T,U)-2$ by
the definition of plays).
In $T'=B'(G_1,C(G_2,G_1))$ we can thus replace $G_1$ with $V_1$ without affecting
$\eqlevel(T',U')$; indeed, we have
$\eqlevel(T',B'(V_1,C(G_2,V_1))\geq\eqlevel(G_1,V_1)$
(using Proposition~\ref{prop:congruence}(2)), 
and $\eqlevel(G_1,V_1)>\eqlevel(T',U')$ thus entails that
 $\eqlevel(B'(V_1,C(G_2,V_1)),U')=\eqlevel(T',U')$
(by Proposition~\ref{prop:elkeep}(1)).
If also $G_2$ can be reached from $A(G_1,G_2)$ in less than two
steps, we similarly get $V_2$, where $U\gt{r'_4}V_2$ for some $r'_4$, so that
$\eqlevel(B'(V_1,C(V_2,V_1)),U')=\eqlevel(T',U')$; hence 
\begin{center}
$\toppair{A(G_1,G_2)}
{U}\bothgt{r_1r_2}{r_1'r'_2}\toppair{T'}{U'}\econc
\toppair{B'(V_1,C(V_2,V_1))}{U'}
$
\end{center}
is a well-defined modified play in this case.
Here $U$ is the ``pivot'', and we note that $U',V_1,V_2$ are all reachable
from $U$ in at most two steps. Hence if we present $U$
in a ``$2$-top
form'', say $U=G\sigma$ where 
$G=A_0(A_{1}(x_1,x_2),A_{2}(x_3,x_4))$, then 
we have $U'=F\sigma$, $V_1=F_1\sigma$, $V_2=F_2\sigma$ where 
$G\gt{r'_1r'_2}F$,  
$G\gt{r'_3}F_1$, $G\gt{r'_4}F_2$. 
Now the ``bal-result'' $(T'',U')=(B'(V_1,C(V_2,V_1)),U')$ can be
presented as $(E\sigma,F\sigma)$ where $E=B'(F_1,C(F_2,F_1))$; we note 
that in $U=G\sigma$ the top $G$ is small, hence also $E,F$ 
are small, while the terms $x\sigma$ might be large. We now formalize
(and generalize) the observation that has been exemplified.

\medskip

We again consider a fixed general grammar 
$\calG=(\calN,\act,\calR)$, and the numbers $m$ (\ref{eq:marity}) and $d_0$
(\ref{eq:Mzero}). We say 
that 
a \emph{play} $\rho=\toppair{T}{U}\bothgt{u}{u'}\toppair{T'}{U'}$
\emph{enables $L$-balancing} if 
$|u|=d_0$ (hence also $|u'|=d_0$) 
 and
$T\gt{u}T'$ is root-performable, i.e.,
$T=A(x_1,\dots,x_m)\sigma'$, $A(x_1,\dots,x_m)\gt{u}E'$,
and thus $T'=E'\sigma'$ (where $A\in\calN$, $E'\in\trees_\calN$,
$\varin(E')\subseteq\{x_1,\dots,x_m\}$).
We can thus write 
\begin{center}
$\rho\,=\,\toppair{T}{U}\bothgt{u}{u'}\toppair{T'}{U'}
\,=\,\toppair{A(x_1,\dots,x_m)\sigma'}{U}\bothgt{u}{u'}\toppair{E'\sigma'}{U'}$\,.
\end{center}
(We have not excluded that $E'=x_i$ for some $i\in[1,m]$.)

In the described case, in $T'$ we can replace each occurrence
of a root-successor of $T$ (which is $x_i\sigma'$ for $i\in[1,m]$)
with a term that is 
shortly reachable from $U$ so that 
$\eqlevel(T',U')$ is unaffected by this replacement; we now make this
claim more
precise, referring again to the illustration in Figure~\ref{fig:balstep}.

Suppose $A(x_1,\dots,x_m)\gt{w_{[A,i]}}x_i$ (cf. the definitions around~(\ref{eq:Mzero})), hence
$T\gt{w_{[A,i]}}x_i\sigma'$; since 
$\eqlevel(T,U)= d_0+\eqlevel(T',U')\geq d_0$ and $|w_{[A,i]}|<d_0$, 
there must be $\bar{v}_i\in\calR^+$ and a term $V_i$ such that
$|\bar{v}_i|=|w_{[A,i]}|$,
$\lab(\bar{v}_i)=\lab(w_{[A,i]})$, $U\gt{\bar{v}_i}V_i$, 
and $\eqlevel(x_i\sigma',V_i)\geq
\eqlevel(T,U)-|w_{[A,i]}|>\eqlevel(T,U)-d_0=
\eqlevel(T',U')$ (we use Proposition~\ref{prop:elkeep}(3)).
We can thus reason for all $i\in[1,m]$.
If there is no $w_{[A,i]}$ for some $i\in[1,m]$, then $x_i\sigma'$ is
not ``exposable'' in $T=A(x_1,\dots,x_m)\sigma'$, hence 
not in $T'=E'\sigma'$ either, and
$x_i\sigma'$ can
be replaced by any term without changing the equivalence class of $T'$;
in this case we put $V_i=U$, thus having $U\gt{\varepsilon} V_i$.
Therefore $\eqlevel(E'\sigma',U')=\eqlevel(E'\sigma'',U')$ where 
$x_i\sigma''=V_i$ for all $i\in[1,m]$ 
(by using 
Propositions~\ref{prop:congruence}(2) and~\ref{prop:elkeep}(1)).

Hence for a play $\rho=\toppair{T}{U}\bothgt{u}{u'}\toppair{T'}{U'}$
in the above notation we can soundly define
an \emph{$L$-balancing step}
$\rho\der_L\rho'$
where $\rho'$ is a
modified play
$\rho'=\rho\econc(E'\sigma'',U')$, depicted in Fig.\ref{fig:balstep}.
For such an $L$-balancing step  $\rho\der_L\rho'$,
the term $U$ is called the \emph{pivot} 
and the pair $(E'\sigma'',U')$ is called the \emph{bal-result}. 

An  \emph{$R$-balancing step} $\rho\der_R\rho'$ is defined symmetrically:
if in $\rho=\toppair{T}{U}\bothgt{u}{u'}\toppair{T'}{U'}$ we have 
$|u|=|u'|=d_0$ and
$U\gt{u'}U'$  is root-performable, 
and presented as $A(x_1,\dots,x_m)\sigma'\gt{u'}F'\sigma'$,
then we can soundly define 
\begin{center}
$\toppair{T}{A(x_1,\dots,x_m)\sigma'}\bothgt{u}{u'}\toppair{T'}{F'\sigma'}\,\der_R\,
\toppair{T}{A(x_1,\dots,x_m)\sigma'}\bothgt{u}{u'}\toppair{T'}{F'\sigma'}
\econc
\toppair{T'}{F'\sigma''}$\,;
\end{center}
here $T$ is the \emph{pivot} and $(T',F'\sigma'')$ is the
\emph{bal-result}.

\subparagraph*{Relation of the tops of the pivot and of the bal-result.}
We now look in more detail at the fact that the pivot of a balancing
step and the respective bal-result can be written $G\sigma$ and
$(E\sigma,F\sigma)$ for specifically related small ``tops'' $G,E,F$
(as is also depicted in Figure~\ref{fig:balstep}).

We say that a finite term $G$ is a \emph{$p$-top}, for $p\in\Natpos$,
if $\height(G)\leq p$, each depth-$p$ subterm is a variable,
and $\varin(G)=\{x_1,\dots,x_n\}$ for some $n\in\Nat$; hence $n\leq m^p$ 
(for $m$ being
the maximum arity of nonterminals~(\ref{eq:marity})).

We note that each term $W$ has a \emph{$p$-top form} $G\sigma$, i.e., 
$W=G\sigma$, $G$ is a $p$-top, $\support(\sigma)\subseteq\varin(G)$,
and
we have $x\sigma\in\var$ for each $x$
occurring in $G$ in depth less than $p$.
(Only a branch of $W$ that finishes with a variable in depth less than
$p$ gives rise to such a branch in $G$.)
E.g., a $2$-top form of $A(B(x_{9},C(x_3,x_6)),x_9)$
is $G\sigma$ where $G=A(B(x_{1},x_2),x_3)$ and 
$\sigma=[x_1/x_9, x_2/C(x_3,x_6), x_3/x_9]$; another $2$-top form of
this term is $G'\sigma'$ where $G'=A(B(x_{1},x_2),x_1)$ and 
$\sigma'=[x_1/x_9, x_2/C(x_3,x_6)]$.
(We could strengthen the
definition to get the unique $p$-top form to each term, but this is not
necessary.) 

We say that $G\sigma$ is a \emph{$p$-safe form of} $W$ if $W=G\sigma$ and 
$W\gt{v}$, $|v|\leq p$, implies $G\gt{v}$
(i.e., each word $v\in\calR^*$ of length at most $p$ that is performable from $W$ is
also performable from $G$).
We easily observe that each $p$-top form $G\sigma$ of $W$ is also
a $p$-safe form of $W$.

The next proposition follows immediately from the definition of balancing
steps.

\begin{proposition}\label{prop:relpivottobalres}
Let $W$ be the pivot and $(T'',U'')$ the bal-result of an
$L$-balancing step. 
Then for any $d_0$-safe form $G\sigma$ of $W$ we have
$(T'',U'')=(E\sigma,F\sigma)$ where 
\begin{itemize}
	\item
	$G\gt{u'}F$  for some $u'\in\calR^+$, $|u'|=d_0$; 
	\item
		$E=E'\overline{\sigma}$ where
		$A(x_1,\dots,x_m)\gt{u}E'$ for some $A\in\calN$,
		$u\in\calR^+$, $|u|=d_0$, and for all $i\in[1,m]$ we
		have
		$G\gt{\overline{v}_i}F_i$ where $F_i=x_i\overline{\sigma}$, for some
		$\overline{v}_i$, $|\overline{v}_i|<d_0$ (hence
		$T''=E\sigma=E'\overline{\sigma}\sigma=E'\sigma''$ 
		where 
		$W\gt{\overline{v}_i}x_i\sigma''$,
		for all
		$i\in[1,m]$).
\end{itemize}
A symmetric claim holds if $W$, $(T'',U'')$ correspond to an $R$-balancing
step.
\end{proposition}

We note a concrete consequence for future use.
(Fig.~\ref{fig:pathinLTSG} might be again helpful.)

\begin{corollary}\label{cor:sizebalresult}
Let $G\sigma$ be a $d_0$-safe form of $W$.
If $W$ is the pivot  of a balancing step, then the respective
bal-result can be written as $(E\sigma,F\sigma)$
where $\varin(E,F)\subseteq\varin(G)$ and
\begin{center}
	$\pressize(E,F)\leq \pressize(G)+(m{+}2)\cdot d_0\cdot\stepinc$.
\end{center}
\end{corollary}	
\begin{proof}
W.l.o.g. we assume an
$L$-balancing step, and use
$E=E'\overline{\sigma}$ and $F$ guaranteed 
by Proposition~\ref{prop:relpivottobalres},
where $x_i\overline{\sigma}=F_i$ for all $i\in[1,m]$.
We thus have 
$$\pressize(E,F)\leq \propsize(E')+\pressize(\{F,F_1,F_2,\dots,F_m\}),$$
since for presenting $E$ we redirect each arc in $E'$ that leads to $x_i$ towards the root
of $F_i$ (for $i\in[1,m]$). Since 
$A(x_1,\dots,x_m)\gt{u}E'$ where $|u|=d_0$, we have 
 $\propsize(E')\leq d_0\cdot\stepinc$.
Since all $F, F_1,F_2,\dots, F_m$ are reachable from $G$ in at most
$d_0$ steps, we get
$\pressize(\{F,F_1,F_2,\dots,F_m\})\leq\pressize(G)+(m+1)\cdot
d_0\cdot\stepinc$ by Proposition~\ref{prop:sizeinc}(3); moreover, all
sets
$\varin(F)$ and $\varin(F_i)$, $i\in[1,m]$, are thus subsets of
$\varin(G)$. The claim follows.
\end{proof}

We derive a small bound on the number
of bal-results when the pivot is fixed.
We put
\begin{equation}\label{eq:done}
	d_1= 2\cdot|\calN|\cdot (\max\{d_0,|\calR|^{d_0}\})^{m+2}
\end{equation}
(referring to
the grammar $\calG=(\calN,\act,\calR)$).

\begin{proposition}\label{prop:fewbaltopivot}
The number of bal-results related to a fixed pivot $W$ is at most
$d_1$.
\end{proposition}
\begin{proof}
Given $W$, we fix its $d_0$-safe form 
$G\sigma$ (e.g., a $d_0$-top form). Now we suppose that $W=G\sigma$ is the pivot of an $L$-balancing step;
let $(E\sigma,F\sigma)=(E'\overline{\sigma}\sigma,F\sigma)$
be the respective bal-result, as captured by 
Proposition~\ref{prop:relpivottobalres}.
We have at most $|\calR|^{d_0}$ options for $u'$ determining $F$, and
at most $|\calN|\cdot |\calR|^{d_0}$ options for $E'$. 
For each $i\in[1,m]$, we have at most 
$1+|\calR|^1+|\calR|^2\cdots+|\calR|^{d_0-1}\leq
\max\{d_0,|\calR|^{d_0}\}$ options for $F_i$.
Altogether we get no more than 
$|\calN|\cdot (\max\{d_0,|\calR|^{d_0}\})^{m+2}$ options for the
bal-result. 
The same number bounds the possible bal-results of $R$-balancing steps
with the pivot $W$, hence the claim follows. 
\end{proof}

\subparagraph*{Balanced modified plays, and pivot paths.}
We now describe a balancing policy, yielding a sequence of balancing steps 
that transform a completed play to a ``balanced'' modified play;
the idea of this policy
(in a different
framework)
can be traced back to S\'enizergues~\cite{Senizergues:TCS2001} (and was
also used by Stirling~\cite{Stirling:TCS2001}).

Let $T_0\not\sim U_0$ and let $\pi$ be a completed play $\pi$ from
$(T_0,U_0)$. We show a sequence of transformation phases; 
after $j$ phases we will get a completed modified  play 
 from $(T_0,U_0)$ of the
form 
\begin{center}
	$\pi_j=\mu_0\rho'_1\mu_1\rho'_2\cdots \mu_{j-1}\rho'_j\pi'_j$
\end{center}
where 
$\pi'_j$ is a play to be transformed in the $(j{+}1)$-th phase.
We start with $\pi_0=\pi'_0=\pi$.
In general 
$\pi'_j$ is not a suffix of $\pi$ but the lengths of the modified plays
$\pi_0,\pi_1,\pi_2,\dots$ are the same (recall
Proposition~\ref{prop:samelengthmodplays}).
In the end
we get a \emph{balanced modified play} 
$\pi_\ell=\mu_0\rho'_1\mu_1\rho'_2\cdots \mu_{\ell-1}\rho'_\ell\pi'_\ell$
(for some $\ell\geq 0$) where $\pi'_\ell$ is non-transformable;
this final modified play 
$\pi_\ell=\mu_0\rho'_1\mu_1\rho'_2\cdots
\mu_{\ell-1}\rho'_\ell\mu_{\ell}$ (where $\mu_{\ell}=\pi'_\ell$)
can be also presented as 
\[
\lbrac\toppair{T_0}{U_0}\bothgt{v_0}{v'_0}\toppair{T_1}{U_1}\rbrac
		\lbrac\toppair{T_1}{U_1}\bothgt{u_1}{u'_1}
	\toppair{T'_1}{U'_1}\econc\toppair{T''_1}{U''_1}\rbrac
	\lbrac\toppair{T''_1}{U''_1}\bothgt{v_1}{v'_1}\toppair{T_2}{U_2}\rbrac
	\lbrac\toppair{T_2}{U_2}
	\bothgt{u_2}{u'_2}\toppair{T'_2}{U'_2}\econc\toppair{T''_2}{U''_2}\rbrac
	\cdots
\lbrac\toppair{T_\ell}{U_\ell}\bothgt{u_\ell}{u'_\ell}
	\toppair{T'_\ell}{U'_\ell}\econc\toppair{T''_\ell}{U''_\ell}\rbrac
	\lbrac\toppair{T''_\ell}{U''_\ell}
	\bothgt{v_\ell}{v'_\ell}\toppair{T_{\ell+1}}{U_{\ell+1}}\rbrac
\]
where $\mu_j$ is $\toppair{T_0}{U_0}\bothgt{v_0}{v'_0}\toppair{T_1}{U_1}$
for $j=0$ and
$\toppair{T''_j}{U''_j}\bothgt{v_j}{v'_j}\toppair{T_{j+1}}{U_{j+1}}$
for $j\in[1,\ell]$, and 
$\rho'_j$ is $\toppair{T_j}{U_j}\bothgt{u_j}{u'_j}
\toppair{T'_j}{U'_j}\econc
\toppair{T''_j}{U''_j}$ (for $j\in[1,\ell]$).
By  $\rho_j$ we denote 
$\toppair{T_j}{U_j}\bothgt{u_j}{u'_j}
\toppair{T'_j}{U'_j}$, and we have either $\rho_j\der_L\rho'_j$ or 
$\rho_j\der_R\rho'_j$. (Hence all $\mu_j$ and $\rho_j$ are plays,
while $\rho'_j$ is a modified play resulting from $\rho_j$ by
a balancing step.) 
By our conventions (and associativity of $\econc$) we can present 
$\pi_\ell=\mu_0\rho'_1\mu_1\rho'_2\cdots
\mu_{\ell-1}\rho'_\ell\mu_{\ell}$ also as
\begin{equation}\label{eq:modifplay}
\textnormal{
	$\lbrac\toppair{T_0}{U_0}\bothgt{v_0}{v'_0}\toppair{T_1}{U_1}\bothgt{u_1}{u'_1}
	\toppair{T'_1}{U'_1}\rbrac\econc
	\lbrac\toppair{T''_1}{U''_1}\bothgt{v_1}{v'_1}\toppair{T_2}{U_2}\bothgt{u_2}{u'_2}\toppair{T'_2}{U'_2}\rbrac\econc
	\lbrac\toppair{T''_2}{U''_2}
\bothgt{v_2}{v'_2}
	\cdots\cdots
\toppair{T_\ell}{U_\ell}\bothgt{u_\ell}{u'_\ell}
	\toppair{T'_\ell}{U'_\ell}\rbrac\econc
	\lbrac\toppair{T''_\ell}{U''_\ell}
	\bothgt{v_\ell}{v'_\ell}\toppair{T_{\ell+1}}{U_{\ell+1}}\rbrac$.
}
\end{equation}
There are $\ell$ (occurrences of) pivots
$W_1,W_2,\cdots,W_\ell$ in~(\ref{eq:modifplay}), 
where $W_j\in\{T_j,U_j\}$
for each $j\in[1,\ell]$; the bal-result corresponding to $W_j$ is
$(T''_j,U''_j)$.
Though the pivots $W_j$ can be changing their sides (we can have, e.g.,
$W_j=U_j$ and $W_{j+1}=T_{j+1}$), they will be on one 
specific \emph{pivot path} in the  LTS $\calL^{\ltsrul}_\calG$,
denoted 
\begin{equation}\label{eq:pivotpath}
	\textnormal{
		$W_0\gt{w_0}W_1\gt{w_1}W_2\cdots\gt{w_{\ell-1}}W_\ell\gt{w_\ell}W_{\ell+1}$
}
\end{equation}
and defined below; 
we will have $W_0\in\{T_0,U_0\}$ and
$W_{\ell+1}\in\{T_{\ell+1},U_{\ell+1}\}$ but
$W_0, W_{\ell+1}$ are no pivots, except the case $w_0=\varepsilon$ and
$W_0=W_1$.
The pivot path will be a useful ingredient for applying our bound on
$(n,s,g)$-sequences
(Lemma~\ref{lem:ELdecreasbound}).

Now we describe the transformation phases (as non-effective
procedures), giving also a finer presentation of $\mu_j$ ($j\in[1,\ell]$) as 
$\mu_j=\mu^\uncl_j$ or $\mu_j=\mu^\uncl_j\mu^\dsink_j$
($\uncl$ for ``unclear'', $\dsink$ for ``sinking'') to be
discussed later.
The first phase, starting with $\pi_0=\pi$, works as follows:
\begin{enumerate}
\item
If possible, present $\pi_0$ as $\mu_0\rho_1\pi'$ where $\rho_1$ enables a balancing step
(on any side)
and $\mu_0\rho_1$ is the shortest possible. If there is no such
presentation of $\pi_0$, then put $\mu_0=\pi_0$ and \emph{halt} (here
$\ell=0$). In this case we do not need to define the
path~(\ref{eq:pivotpath}). 
\item
Replace $\rho_1$ with $\rho'_1$ where $\rho_1\der_L\rho'_1$
or $\rho_1\der_R\rho'_1$ (choosing arbitrarily when $\rho_1$ allows both 
$L$-balancing and $R$-balancing). Finally replace 
$\pi'$
with a completed play $\pi'_{1}$ from the bal-result,
i.e., from $\finishofplay(\rho'_1)$, thus getting
$\pi_1=\mu_0\rho'_1\pi'_1$ where
$\mu_0\rho'_1=\toppair{T_0}{U_0}\bothgt{v_0}{v'_0}\toppair{T_1}{U_1}\bothgt{u_1}{u'_1}
\toppair{T'_1}{U'_1}\econc
\toppair{T''_1}{U''_1}$.
We also define the prefix $W_0\gt{w_0}W_1$ of~(\ref{eq:pivotpath}):
if we have  $\rho_1\der_L\rho'_1$, hence $W_1=U_1$, then this prefix
is $U_0\gt{v'_0}U_1$; if $\rho_1\der_R\rho'_1$, hence $W_1=T_1$, then
the prefix 
is $T_0\gt{v_0}T_1$. 
\end{enumerate}
For $j\geq 1$, the $(j{+}1)$-th phase starts with 
$\pi_j=\mu_0\rho'_1\mu_1\rho'_2\cdots \mu_{j-1}\rho'_j\pi'_j$
where the last balancing step was either left, $\rho_j\der_L\rho'_j$, or
right, $\rho_j\der_R\rho'_j$.
 We 
describe the $(j{+}1)$-th phase for the case $\rho_j\der_L\rho'_j$;
the other case is symmetric. We recall Figure~\ref{fig:balstep} and 
present $\rho'_j\pi'_j$ as
\begin{center}
$\toppair{A(x_1,\dots,x_m)\sigma'}{U_j}\bothgt{u_j}{u'_j}\toppair{E'\sigma'}{U'_j}\econc
\toppair{E'\sigma''}{U'_j}
\bothgt{v}{v'}$\,.
\end{center}
We have also already defined the prefix
$W_0\gt{w_0}W_1\gt{w_1}W_2\cdots\gt{w_{j-1}}W_j$
of~(\ref{eq:pivotpath}), and we have $W_j=U_j$ in our considered case
 $\rho_j\der_L\rho'_j$.

Informally, 
the $(j{+}1)$-phase aims to make a balancing step in $\pi'_j$ as
early as possible but balancing at the opposite side than previously
is a bit constrained.
In our case a future right balancing would entail
that the next pivot is on the path $E'\sigma''\gt{v}$, and we first have  
to wait until a term $x_i\sigma''$ is exposed (i.e., until a prefix of
$v$ exposes one of $V_1, V_2$ in Figure~\ref{fig:balstep}, where $U$
represents the last pivot $U_j$). Only then a
right balancing is allowed. This exposing must obviously happen soon
(i.e., for a
short prefix of $v$) if a further left balancing is not enabled for a
while (since in this case $E'\sigma''\gt{v}$ must be quickly sinking along a
branch of $E'$).
Hence if even under this constraint the earliest next balancing will
be a right balancing, then the next pivot $T_{j+1}$ will be the final
term on a path
$E'\sigma''\gt{v_{j1}}x_i\sigma''\gt{v_{j2}}T_{j+1}$ where
$v_{j1}v_{j2}$ is a prefix of $v$ and $v_{j_1}$ is short.
Since $x_i\sigma''$ is reachable by a short $\bar{v}$
from the last pivot $U_j$ 
(e.g., if $x_i\sigma''=V_2$ in Figure~\ref{fig:balstep}, then
we use $U\gt{\bar{v}_2}V_2$),
we continue building the pivot
path smoothly:
in our case $W_j\gt{w_j}W_{j+1}$ will be defined as
$U_j\gt{\bar{v}}x_i\sigma''\gt{v_{j2}}T_{j+1}$. 
When a left (unconstrained) balancing is the earliest possibility, 
$W_j\gt{w_j}W_{j+1}$ will be defined simply as 
$U_j\gt{u'_jv'_j}U_{j+1}$ for the respective prefix $v'_j$ of $v'$.
(We note that the pivot path gets a bit
shorter than the modified play~(\ref{eq:modifplay}) whenever a switch of balancing
sides occurrs.) 
Now we describe the $(j{+}1)$-phase more formally
(assuming $\rho_j\der_L\rho'_j$).

\begin{enumerate}
	\item
If possible, present $\pi'_j=\toppair{E'\sigma''}{U'_j}
\bothgt{v}{v'}$ as $\mu_{j}\rho_{j+1}\pi'$ 
with the shortest possible 
$\mu_{j}\rho_{j+1}$ where 
\begin{enumerate}[a)]
\item
either $\rho_{j+1}$ enables $L$-balancing, 
\item or $\rho_{j+1}$ does not 
	enable $L$-balancing but it enables $R$-balancing
	and the
	path $E'\sigma''\gt{v_j}T_{j+1}$ in the 	
	play	
	$\mu_{j}=\toppair{E'\sigma''}{U'_j}\bothgt{v_{j}}{v'_{j}}\toppair{T_{j+1}}{U_{j+1}}$
	can be written 
	$E'\sigma''\gt{v_{j1}}x_i\sigma''\gt{v_{j2}}T_{j+1}$
	where  $E'\gt{v_{j1}}x_i$, for some $i\in[1,m]$.
	(We recall that $U_j\gt{\overline{v}}x_i\sigma''$ where
	$|\overline{v}|<d_0$.)

\end{enumerate}
If there is no such
presentation of $\pi'_{j}$, then  
put $\mu_{j}=\pi'_j$ and \emph{halt} (here $\ell=j$).
In this case we have 
$\rho'_\ell\mu_{\ell}=\toppair{A(x_1,\dots,x_m)\sigma'}{W_\ell}\bothgt{u_\ell}{u'_\ell}\toppair{E'\sigma'}{U'_\ell}\econc
\toppair{E'\sigma''}{U'_\ell}\bothgt{v_\ell}{v'_\ell}\toppair{T_{\ell+1}}{U_{\ell+1}}$
and we define $W_\ell\gt{w_\ell}W_{\ell+1}$ as 
$W_\ell\gt{u'_\ell v'_\ell}U_{\ell+1}$. 

In each case we get
$\mu_{j}=\toppair{E'\sigma''}{U'_j}\bothgt{v_{j}}{v'_{j}}\toppair{T_{j+1}}{U_{j+1}}$,
and if $E'\sigma''\gt{v_j}T_{j+1}$ can be written
 $E'\sigma''\gt{v_{j1}}x_i\sigma''\gt{v_{j2}}T_{j+1}$ where 
$E'\gt{v_{j1}}x_i$
(which holds in the case b) by definition),
then we put 
$\mu_j=\mu^\uncl_{j}\mu^\dsink_j$ where
$\mu^\uncl_{j}=\toppair{E'\sigma''}{U'_j}\bothgt{v_{j1}}{v'_{j1}}\toppair{x_i\sigma''}{\overline{U}_j}$
and 
$\mu^\dsink_{j}=\toppair{x_i\sigma''}{\overline{U}_j}\bothgt{v_{j2}}{v'_{j2}}\toppair{T_{j+1}}{U_{j+1}}$;
otherwise $\mu_j=\mu^\uncl_{j}$.

(We note that the ``unclear'' play $\mu^\uncl_{j}$ is always short. The
``sinking'' play $\mu^\dsink_j$ can be nonempty even if there is no
switch in balancing sides, and $\mu^\dsink_{j}$  can be long, but both paths in 
$\mu^\dsink_j$ are quickly sinking [since no balancing possibility
appears].)

\item
Replace $\rho_{j+1}$ with $\rho'_{j+1}$
where $\rho_{j+1}\der_{L}\rho'_{j+1}$ 
in the case a),
and $\rho_{j+1}\der_{R}\rho'_{j+1}$ in the case b).
Finally replace 
$\pi'$
with a completed play $\pi'_{j+1}$ from the bal-result,
i.e., from $\finishofplay(\rho'_{j+1})$, thus getting
$\pi_{j+1}=\mu_0\rho'_1\mu_1\rho'_2\cdots \mu_{j}\rho'_{j+1}\pi'_{j+1}$.

In the case $\rho_{j+1}\der_{L}\rho'_{j+1}$ we have 
$\rho'_j\mu_{j}=\toppair{A(x_1,\dots,x_m)\sigma'}{W_j}\bothgt{u_j}{u'_j}\toppair{E'\sigma'}{U'_j}\econc
\toppair{E'\sigma''}{U'_j}\bothgt{v_j}{v'_j}\toppair{T_{j+1}}{W_{j+1}}$
and we put $w_j=u'_j v'_j$, thus defining $W_j\gt{w_j}W_{j+1}$.

In the case $\rho_{j+1}\der_{R}\rho'_{j+1}$ we have 
$\rho'_j\mu_{j}=\toppair{A(x_1,\dots,x_m)\sigma'}{W_j}\bothgt{u_j}{u'_j}\toppair{E'\sigma'}{U'_j}\econc
\toppair{E'\sigma''}{U'_j}
\bothgt{v_{j1}}{v'_{j1}}\toppair{x_i\sigma''}{\overline{U}_j}
\bothgt{v_{j2}}{v'_{j2}}\toppair{W_{j+1}}{U_{j+1}}$

and we define $W_j\gt{w_j}W_{j+1}$ by 
putting $w_j=\overline{v}\, v_{j2}$
for a respective $\overline{v}$, $|\overline{v}|<d_0$, for which 
$W_j\gt{\overline{v}}x_i\sigma''$.
\end{enumerate}
As already mentioned, the work of the $(j{+}1)$-phase in the case 
$\rho_j\der_R\rho'_j$ is symmetric;
here we have $R$-balancing in the ``unconditional'' case a), and 
$L$-balancing in the case b) that now requires a prefix
$\mu^\uncl_{j}=\toppair{T'_j}{F'\sigma''}\bothgt{v_{j1}}{v'_{j1}}
\toppair{\overline{T}_j}{x_i\sigma''}$ (where $x_i\sigma''$ is shortly reachable
from the last pivot $T_j$).

\section{Analysis of Balanced Modified Plays}\label{sec:analysis}

Assuming a given grammar $\calG=(\calN,\act,\calR)$,
we have shown a transformation of a completed play $\pi$ starting
in $(T_0,U_0)$ (where $T_0, U_0$ can be large regular terms)
to
a balanced modified play
$\pi_\ell=\mu_0\rho'_1\mu_1\rho'_2\cdots
\mu_{\ell-1}\rho'_\ell\mu_\ell$ in the form~(\ref{eq:modifplay}),
repeated here:
\begin{equation}\label{eq:modifplaynext}
\textnormal{
	$\lbrac\toppair{T_0}{U_0}\bothgt{v_0}{v'_0}\toppair{T_1}{U_1}\bothgt{u_1}{u'_1}
	\toppair{T'_1}{U'_1}\rbrac\econc
	\lbrac\toppair{T''_1}{U''_1}\bothgt{v_1}{v'_1}\toppair{T_2}{U_2}\bothgt{u_2}{u'_2}\toppair{T'_2}{U'_2}\rbrac\econc
	\lbrac\toppair{T''_2}{U''_2}
\bothgt{v_2}{v'_2}
	\cdots\cdots
\toppair{T_\ell}{U_\ell}\bothgt{u_\ell}{u'_\ell}
	\toppair{T'_\ell}{U'_\ell}\rbrac\econc
	\lbrac\toppair{T''_\ell}{U''_\ell}
	\bothgt{v_\ell}{v'_\ell}\toppair{T_{\ell+1}}{U_{\ell+1}}\rbrac$.
}
\end{equation}
In this section we perform a technical analysis of such
$\pi_\ell$,
to verify that we indeed get specific small
numbers $n,s,g$ and $c$ yielding~(\ref{eq:elbound}), where
$\calE=\calE_{\calB}$ for  a ``sound'' $(n,s,g)$-candidate $\calB$
(which will turn out equal to the base
$\calB_{n,s,g}$, as discussed in Section~\ref{sec:finalproof}).
First we recall the
discussion at the beginning of Section~\ref{sec:balance} 
and give an informal 
overview of the future analysis.

We recall that in the pivot path
\begin{center}
$W_0\gt{w_0}W_1\gt{w_1}W_2\cdots\gt{w_{\ell-1}}W_\ell\gt{w_\ell}W_{\ell+1}$
\end{center}
we have $W_j\in\{T_j,U_j\}$ for $j\in[0,\ell{+}1]$, and the pivots 
$W_1,W_2,\dots$, $W_\ell$ have the respective related (eqlevel-decreasing)
bal-results $(T''_1,U''_1),(T''_2,U''_2),\dots,(T''_\ell,U''_\ell)$.
Referring to~(\ref{eq:modifplaynext}), we recall that $u_i$ (hence also $u'_i$) are short since
$|u_i|=|u'_i|=d_0$ for all $i\in[1,\ell]$ (and $d_0$
from~(\ref{eq:Mzero})).
Recalling the discussion around~(\ref{eq:unclearandsink}), we can
present~(\ref{eq:modifplaynext}) in a refined form as
\begin{equation}\label{eq:prelimmodifplayrefined}
\textnormal{
$\lbrac\toppair{T_0}{U_0}\bothgt{v_0}{v'_0}\toppair{T_1}{U_1}\bothgt{u_1}{u'_1}
\toppair{T'_1}{U'_1}\rbrac\econc
\lbrac\toppair{T''_1}{U''_1}\bothgt{v_{11}}{v'_{11}}\toppair{\overline{T}_1}{\overline{U}_1}
\bothgt{v_{12}}{v'_{12}}\toppair{T_2}{U_2}\bothgt{u_2}{u'_2}\toppair{T'_2}{U'_2}\rbrac\econc
\cdots
	\econc
\lbrac\toppair{T''_\ell}{U''_\ell}
\bothgt{v_{\ell 1}}{v'_{\ell 1}}
\toppair{\overline{T}_\ell}{\overline{U}_\ell}
\bothgt{v_{\ell 2}}{v'_{\ell 2}}
\toppair{T_{\ell+1}}{U_{\ell+1}}\rbrac$;
}
\end{equation}
here
$\lbrac\toppair{T''_j}{U''_j}\bothgt{v_j}{v'_j}\toppair{T_{j+1}}{U_{j+1}}\bothgt{u_{j+1}}{u'_{j+1}}\toppair{T'_{j+1}}{U'_{j+1}}\rbrac$,
for $j\in[1,\ell]$,
is presented as
$\lbrac\toppair{T''_j}{U''_j}\bothgt{v_{j1}}{v'_{j1}}\toppair{\overline{T}_j}{\overline{U}_{j}}\bothgt{v_{j2}}{v'_{j2}}\toppair{T_{j+1}}{U_{j+1}}\bothgt{u_{j+1}}{u'_{j+1}}\toppair{T'_{j+1}}{U'_{j+1}}\rbrac$
where $v_{j1}$ is short and both paths
$\overline{T}_j\gt{v_{j2}}T_{j+1}$ and
$\overline{U}_j\gt{v'_{j2}}U_{j+1}$ are $d_0$-sinking (i.e., in each segment
of length $d_0$ of these paths a root-successor in the term that starts the segment is exposed); we can have
$v_{j2}=\varepsilon$.

The first segment $W_0\gt{w_1}W_1$ is one of the paths $T_0\gt{v_0}T_1$ and
$U_0\gt{v'_0}U_1$; each of these two paths is $d_0$-sinking (since
otherwise the first balancing would be possible earlier).
For the segment $W_j\gt{w_j}W_{j+1}$, $j\in[1,\ell]$, we have four
options:
\begin{itemize}
	\item
		$U_j\gt{u'_jv'_{j1}}\overline{U}_{j}\gt{v'_{j2}}U_{j+1}$, if $W_j=U_j$ and
	$W_{j+1}=U_{j+1}$;
	\item
		$T_j\gt{u_jv_{j1}}\overline{T}_{j}\gt{v_{j2}}T_{j+1}$, if $W_j=T_j$ and
		$W_{j+1}=T_{j+1}$;
	\item
		$U_j\gt{\bar{v}}\overline{T}_{j}\gt{v_{j_2}}T_{j+1}$ for some
		$\bar{v}$, $|\bar{v}|<d_0$,
		if $W_j=U_j$ and
		$W_{j+1}=T_{j+1}$;
	\item
		$T_j\gt{\bar{v}}\overline{U}_{j}\gt{v'_{j2}}U_{j+1}$ for some
		$\bar{v}$, $|\bar{v}|<d_0$,
		if $W_j=T_j$ and
		$W_{j+1}=U_{j+1}$.
\end{itemize}
Hence each  $W_j\gt{w_j}W_{j+1}$ has a short ``unclear'' prefix
(it is unclear if it sinks or not), followed
by a $d_0$-sinking suffix (which might be empty, or short, or long
...).

This entails that \emph{if the path  $W_j\gt{w_j}W_{j+1}$ visits a subterm
of $W_0$}, which is surely the case for 
 $W_0\gt{w_0}W_{1}$,
\emph{then $W_{j+1}$ is
 shortly reachable from a subterm of $W_0$}.
 Indeed, 
 if  $W_j\gt{w_j}W_{j+1}$ is  $W_j\gt{w'_j}V\gt{w''_j}W_{j+1}$
 where $V$ is the last subterm of $W_0$ visited by the path 
$W_j\gt{w_j}W_{j+1}$, then $w''_j$ has a short unclear prefix (maybe
empty) followed by a $d_0$-sinking suffix; but if this suffix was not
short, then it would necessarily expose a root-successor in $V$, which is another
subterm of $W_0$; this would contradict the choice of $V$.
(Figure~\ref{fig:pathinLTSG} might be again helpful to realize this
fact.)

For each subterm $V$ of $W_0$ we certainly have only a small number of
terms $W$ that are shortly reachable from $V$.
Since there is only a small number of possible bal-results related to each
concrete pivot $W$, and the bal-results do not repeat, we get that 
\emph{the number of indices $j\in[0,\ell]$ for which $W_j\gt{w_j}W_{j+1}$ visits a subterm
of $W_0$ is bounded by $d\cdot \pressize(T_0,U_0)$} for a small
constant $d$. (We recall that $W_0\in\{T_0,U_0\}$.)

We say that a \emph{segment of the pivot path} of the form 
\begin{center}
$V\gt{w}W_{j+1}\gt{w_{j+1}}W_{j+2}\cdots\gt{w_{j+z-1}}
	W_{j+z}\gt{w'}V'$
\end{center}
is \emph{crucial} if $w$ is a nonempty suffix of $w_j$, $1\leq z\leq\ell{-}j$,
$w'$ is a prefix of $w_{j+z}$,
$V$ is a subterm of $W_0$, and no subterm of $W_0$ is visited
inside the segment; moreover,
either $V'$ is a subterm of $W_0$
or $V'=W_{\ell+1}$ (the end of
the pivot path).
(We can again look at Figure~\ref{fig:pathinLTSG}, and imagine that
$W_0$ is the (maybe large regular) term in the rectangle and $V$ is its subterm 
determined in the third rectangle, whose root is $B$. The next two steps
can be viewed as a prefix of a crucial segment that could finish after
many steps later when some of the root-successors of $B$ in the
rectangle is exposed and becomes the current root.)

Since each crucial segment is non-sinking (until the last
step), it gives rise to an
$(n,s,g)$-sequence (for some small $n,s,g$), as was depicted in
Figure~\ref{fig:twoboundonnonsink} and discussed in the informal
 beginning of
Section~\ref{sec:balance}.
It is thus intuitively clear that \emph{the length of any crucial segment
of the pivot path}, as well as the length of its corresponding segment
of the modified play~(\ref{eq:modifplaynext}), \emph{is bounded by
$d'\cdot\calE_{\calB_{n,s,g}}$} for a small constant $d'$ (and the
$(n,s,g)$-base $\calB_{n,s,g}$, by Lemma~\ref{lem:ELdecreasbound}).

Since each crucial segment is fully determined by the segment
$W_j\gt{w_j}W_{j+1}$ in which it starts (and which visits a subterm of
$W_0$),
there are at most $d\cdot \pressize(T_0,U_0)$ crucial segments, and
their overall length is thus bounded by $d\cdot \pressize(T_0,U_0)\cdot
d'\cdot\calE_{\calB_{n,s,g}}$.
Hence
we are approaching the required bound
\begin{center}
	$\eqlevel(T_0,U_0)\leq c\cdot\big(\calE_{\calB_{n,s,g}}\cdot
\pressize(T_0,U_0)+(\pressize(T_0,U_0))^2\big)$,
\end{center}
for a small constant $c$.
The bound $c\cdot (\pressize(T_0,U_0))^2$ serves for bounding the sum of 
lengths of subpaths of $W_j\gt{w_j}W_{j+1}$ (and the corresponding
subplays in~(\ref{eq:modifplaynext})) when both sides are quickly
sinking ``inside'' the (regular) terms $T_0$ and $U_0$, respectively.
(An extreme case is when there is no balancing
since
both paths from $T_0$ and $U_0$, respectively, are $d_0$-sinking all
the time.)
Since the eq-level drops by one in each step of each play
in~(\ref{eq:modifplaynext}), we cannot have a repeat of a pair there.
Hence there is some small $c$ such that  $c\cdot (\pressize(T_0,U_0))^2$
 bounds the number of those pairs 
 in~(\ref{eq:modifplaynext}) in which 
both members are ``close to'' subterms of $T_0$ or $U_0$.
This bounds the sum of lengths of the respective
 segments of~(\ref{eq:modifplaynext}) that are sinking ``closely to
 $T_0,U_0$'' on both sides.

The claim of Theorem~\ref{th:computingelbound} is now almost
clear; it will be completed 
in Section~\ref{sec:finalproof} where we show that the respective
constant $\calE=\calE_{\calB_{n,s,g}}$ is indeed computable.
In the rest of this section we perform a routine (and somewhat
tedious) analysis to show some concrete numbers $n,s,g,c$ (cf.
Table~\ref{tab:constants} at the end of the paper).

\subparagraph{Refined presentations of balanced modified plays.}

Assuming a given grammar $\calG=(\calN,\act,\calR)$, we fix 
a completed play $\pi$ from some (maybe large regular terms) $(T_0,U_0)$
and its transformation
$\pi_\ell=\mu_0\rho'_1\mu_1\rho'_2\cdots \mu_{\ell-1}\rho'_\ell\pi'_\ell$
in the previous notation; in fact, we also use a finer
form and write
\begin{equation}\label{eq:piellbasic}
\pi_\ell=\mu^\dsink_0\rho'_1\mu^\uncl_1\mu^\dsink_1\rho'_2\mu^\uncl_2\mu^\dsink_2\cdots
\rho'_\ell\mu^\uncl_\ell\mu^\dsink_\ell
\end{equation}
(where the superscript $\uncl$ can be read as ``unclear'' and 
$\dsink$ as ``sinking''). 
We add that $\mu^\dsink_0=\mu_0$ and that we view $\varepsilon$ (the empty sequence)
also as the \emph{empty play}, and we put 
$\mu^\dsink_j=\varepsilon$ in the cases where $\mu^\dsink_j$
has not been defined
explicitly. As expected, we stipulate  $\length(\varepsilon)=0$,
$\pairs(\varepsilon)=\emptyset$, and 
$\mu\varepsilon=\varepsilon\mu=\mu$
for all (modified) plays $\mu$.

The presentation~(\ref{eq:modifplay}) is accordingly refined
(as in~(\ref{eq:prelimmodifplayrefined})) to
\begin{equation}\label{eq:modifplayrefined}
\lbrac\toppair{T_0}{U_0}\bothgt{v_0}{v'_0}\toppair{T_1}{U_1}\bothgt{u_1}{u'_1}
\toppair{T'_1}{U'_1}\rbrac\econc
\lbrac\toppair{T''_1}{U''_1}\bothgt{v_{11}}{v'_{11}}\toppair{\overline{T}_1}{\overline{U}_1}
\bothgt{v_{12}}{v'_{12}}\toppair{T_2}{U_2}\bothgt{u_2}{u'_2}\toppair{T'_2}{U'_2}\rbrac\econc
\cdots\econc
\lbrac\toppair{T''_\ell}{U''_\ell}
\bothgt{v_{\ell 1}}{v'_{\ell 1}}
\toppair{\overline{T}_\ell}{\overline{U}_\ell}
\bothgt{v_{\ell 2}}{v'_{\ell 2}}
\toppair{T_{\ell+1}}{U_{\ell+1}}\rbrac
\end{equation}
where,  for $j\in[1,\ell]$, we have
$\mu^\uncl_j=\toppair{T''_j}{U''_j}
\bothgt{v_{j 1}}{v'_{j 1}}
\toppair{\overline{T}_j}{\overline{U}_j}$,
and either
$\mu^\dsink_j=\toppair{\overline{T}_j}{\overline{U}_j}
\bothgt{v_{j 2}}{v'_{j 2}}
\toppair{{T}_{j+1}}{\overline{U}_{j+1}}$ or
$\mu^\dsink_j=\varepsilon$ in which case 
$v_{j2}=v'_{j2}=\varepsilon$, $\overline{T}_{j}=T_{j+1}$,
 $\overline{U}_{j}=U_{j+1}$.

To explain the use of the superscript $\dsink$ (``sinking'') in
$\mu^\dsink_j$, we
introduce a few notions.

An $(A,i)$-sink word $v\in\calR^+$ 
(satisfying $A(x_1,\dots,x_m)\gt{v}x_i$) is also called 
a \emph{sink-segment};
any 
path of the form  $V\gt{v}V'$ is then also 
understood as a sink-segment
(presentable as $A(x_1,\dots,x_m)\sigma\gt{v}x_i\sigma$).
We say that a \emph{path} $V\gt{v}V'$ is \emph{$d_0$-sinking}, if
$v=v_1v_2\cdots v_{k+1}$ where $|v_j|<d_0$ for all $j\in[1,k{+}1]$
and $v_j$, $j\in[1,k]$, are sink-segments.
A zero-length  path
$V\gt{\varepsilon}V$ is $d_0$-sinking, by putting $k=0$ and
$v_{k+1}=\varepsilon$.

A \emph{play} $\mu=\toppair{T}{U}\bothgt{v}{v'}\toppair{T'}{U'}$
is \emph{$d_0$-sinking} if both its paths $T\gt{v}T'$ and
$U\gt{v'}U'$ are $d_0$-sinking. In particular, a zero-length 
play $\mu=\toppair{T}{U}$ is  $d_0$-sinking, and we also view 
the empty play $\varepsilon$ as $d_0$-sinking.

The above transformation (of $\pi$ to $\pi_\ell$) guarantees that 
all plays $\mu_0$, $\mu^\dsink_{1}$, $\mu^\dsink_{2}$, $\dots$,
$\mu^\dsink_{\ell}$ are $d_0$-sinking (therefore we have put 
$\mu_0=\mu^\dsink_0$). Indeed, if some
$\mu^\dsink_{j}$ ($j\in[0,\ell]$)
was not $d_0$-sinking, then there would be a possibility
to make a ``legal''
balancing step earlier in the respective transformation phase.

The presentations~(\ref{eq:piellbasic}) 
and~(\ref{eq:modifplayrefined})
also yield the corresponding refined version
of the pivot path~(\ref{eq:pivotpath}):
\begin{equation}\label{eq:pivotpathfiner}
W_0\gt{w^\dsink_0}W_1\gt{w^\uncl_{1}}\overline{W}_1\gt{w^\dsink_{1}}W_2
\gt{w^\uncl_{2}}\overline{W}_2\gt{w^\dsink_{2}}
\cdots W_\ell\gt{w^\uncl_{\ell}}\overline{W}_\ell\gt{w^\dsink_{\ell}}W_{\ell+1}
\end{equation}
where each segment
$\overline{W}_j\gt{w^\dsink_{j}}W_{j+1}$
(for $j\in [0,\ell]$ when putting $\overline{W}_0=W_0$) corresponds
to one of the paths in the play 
$\mu^\dsink_j$, and is thus $d_0$-sinking. 
More concretely, $W_0\gt{w^\dsink_0}W_1$
(where $w^\dsink_0=w_0$)
is either $T_0\gt{v_0}T_1$ or
$U_0\gt{v'_0}U_1$, and 
$\overline{W}_j\gt{w^\dsink_{j}}W_{j+1}$ is either 
$\overline{T}_j\gt{v_{j2}}T_{j+1}$ or
$\overline{U}_j\gt{v'_{j2}}U_{j+1}$.
Each (``unclear'') segment 
$W_j\gt{w^\uncl_j}\overline{W}_j$ is one of the following paths:
\begin{itemize}
	\item
		$U_j\gt{u'_jv'_{j1}}\overline{U}_{j}$, if $W_j=U_j$ and
		$W_{j+1}=U_{j+1}$ (in which case $U'_j=U''_j$);
	\item
		$T_j\gt{u_jv_{j1}}\overline{T}_{j+1}$, if $W_j=T_j$ and
		$W_{j+1}=T_{j+1}$ (in which case $T'_j=T''_j$);
	\item
		$U_j\gt{v}\overline{T}_{j}$ for some $v$, $|v|<d_0$,
		if $W_j=U_j$ and
		$W_{j+1}=T_{j+1}$;
	\item
		$T_j\gt{v}\overline{U}_{j}$ for some $v$, $|v|<d_0$,
		if $W_j=T_j$ and
		$W_{j+1}=U_{j+1}$.
\end{itemize}
We now note that the length of each segment 
$\rho'_j\mu^\uncl_{j}=\toppair{T_j}{U_j}\bothgt{u_j}{u'_j}
\toppair{T'_j}{U'_j}\econc\toppair{T''_j}{U''_j}\bothgt{v_{j1}}{v'_{j1}}
\toppair{\overline{T}_{j}}{\overline{U}_{j}}$, and of the respective
pivot-path segment $W_j\gt{w^\uncl_j}\overline{W}_j$, can be bounded
by the small number

\begin{equation}\label{eq:shortpivotgap}
	d_2=d_0+(1+d_0\cdot\hinc)\cdot(d_0-1).
\end{equation}
\begin{proposition}\label{prop:dtwo}
	For each $j\in[1,\ell]$ we have $|w^\uncl_j|\leq
	\length(\rho'_j\mu^\uncl_{j})\leq d_2$.
\end{proposition}
\begin{proof}
We have $|w^\uncl_j|\leq
	\length(\rho'_j\mu^\uncl_{j})$ by the above definitions
(since $\length(\rho'_j\mu^\uncl_{j})\geq d_0$, and either 
$|w^\uncl_j|=\length(\rho'_j\mu^\uncl_{j})$ or $|w^\uncl_j|<d_0$).

W.l.o.g. we suppose $\rho_j\der_L\rho'_j$ (illustrated in Fig.\ref{fig:balstep})
and present $\rho'_j\mu^\uncl_{j}$ accordingly as
\begin{center}
$\rho'_j\mu^\uncl_{j}=\toppair{A(x_1,\dots,x_m)\sigma'}{U_j}\bothgt{u_j}{u'_j}
\toppair{E'\sigma'}{U'_j}\econc\toppair{E'\sigma''}{U'_j}\bothgt{v_{j1}}{v'_{j1}}
\toppair{\overline{T}_j}{\overline{U}_{j}}$ 
\end{center}
where
$A(x_1,\dots,x_m)\gt{u_j}E'$ and $|u_j|=d_0$; hence 
$\height(E')\leq 1+d_0\cdot \hinc$.
We have $\overline{T}_j=T_{j+1}$ if $\mu^\dsink_j=\varepsilon$, and 
$\overline{T}_j=x_i\sigma''$ (for some $i\in[1,m]$) 
if $\mu^\dsink_j\neq\varepsilon$.

The path 
$E'\sigma''\gt{v_{j1}}\overline{T}_j$ must be
$d_0$-sinking (otherwise there would be 
an earlier next balancing step).
Hence 
 $|v_{j1}|\leq \height(E')\cdot(d_0-1)$.
We thus get 
\begin{center} 
$\length(\rho'_j\mu^\uncl_{j})=|u_j|+|v_{j1}|\leq d_0 + (1+d_0\cdot
\hinc)\cdot(d_0-1)=d_2$.
\end{center}
\end{proof}	

Having bounded the parts $\rho'_j\mu^\uncl_j$, we will 
now bound 
the total length of the suffixes of $\mu^\dsink_j$ that are ``close
to'' $T_0,U_0$; then we will finally bound the number and the length
of so-called ``crucial segments'' of $\pi_\ell$  starting with pivots that are also close 
to $T_0,U_0$ in a sense.

\subparagraph{Close sink-parts in $\pi_\ell$.}

Since $\mu_0=\mu^\dsink_0=\toppair{T_0}{U_0}\bothgt{v_0}{v'_0}\toppair{T_1}{U_1}$ 
is $d_0$-sinking, both paths $T_0\gt{v_0}T_1$ and $U_0\gt{v'_0}U_1$
are frequently visiting subterms of the terms $T_0$ and $U_0$.
Using the fact that no pair repeats along $\pi_\ell$
(recall Proposition~\ref{prop:samelengthmodplays}),
we now derive a bound on the length of $\mu_0$ and other 
segments that are ``close'' to $(T_0,U_0)$.

For each $j\in[1,\ell]$ where $\mu^\dsink_j\neq\varepsilon$
we define the presentation
$\mu^\dsink_j=\mu^\usink_j\mu^\csink_j$ 
(the superscript $\usink$ for ``unclear sinking'' and 
$\csink$ for ``close sinking'') as follows:
If some of the paths in the play 
$\mu^\dsink_j=\toppair{\overline{T}_j}{\overline{U}_j}\bothgt{v_{j2}}{v'_{j2}}\toppair{T_{j+1}}{U_{j+1}}$
never visits a subterm of $T_0$ or $U_0$, then 
$\mu^\usink_j=\mu^\dsink_j$ and
$\mu^\csink_j=\varepsilon$.
Otherwise 
we write $\mu^\dsink_j$ as 
$\toppair{\overline{T}_j}{\overline{U}_j}
\bothgt{\overline{v}_{j2}}{\overline{v}'_{j2}}
\toppair{\overline{\overline{T}}_j}{\overline{\overline{U}}_j}
\bothgt{\overline{\overline{v}}_{j2}}{\overline{\overline{v}}'_{j2}}
\toppair{T_{j+1}}{U_{j+1}}$
for the shortest prefix
$\mu^\usink_j=
\toppair{\overline{T}_j}{\overline{U}_j}
\bothgt{\overline{v}_{j2}}{\overline{v}'_{j2}}
\toppair{\overline{\overline{T}}_j}{\overline{\overline{U}}_j}$
such that each of the paths 
$\overline{T}_j\gt{\overline{v}_{j2}}\overline{\overline{T}}_j$ 
and 
$\overline{U}_j\gt{\overline{v}'_{j2}}\overline{\overline{U}}_j$ 
visits a
subterm of $T_0$ or $U_0$; in this case 
$\mu^\csink_j=\toppair{\overline{\overline{T}}_j}{\overline{\overline{U}}_j}
\bothgt{\overline{\overline{v}}_{j2}}{\overline{\overline{v}}'_{j2}}
\toppair{T_{j+1}}{U_{j+1}}$.
(Since $\mu^\dsink_j$ is $d_0$-sinking,
both paths 
$\overline{\overline{T}}_j
\gt{\overline{\overline{v}}_{j2}}
T_{j+1}$ and 
$\overline{\overline{U}}_j
\gt{\overline{\overline{v}}'_{j2}}
U_{j+1}$
are frequently visiting subterms of the terms $T_0$ and $U_0$.)
If $\mu^\dsink_j=\varepsilon$, then we put
$\mu^\usink_j=\mu^\csink_j=\varepsilon$; we also put
$\mu_0=\mu^\dsink_0=\mu^\csink_0$ (while $\mu^\usink_0=\varepsilon$).

The balanced modified play $\pi_\ell$ (\ref{eq:piellbasic}) can be
thus presented in more detail as
\begin{equation}\label{eq:piellfinal}
\pi_\ell=\mu^\csink_0\rho'_1\mu^\uncl_1\mu^\usink_1\mu^\csink_1\rho'_2
\mu^\uncl_2\mu^\usink_2\mu^\csink_2
\cdots \rho'_\ell \mu^\uncl_\ell\mu^\usink_\ell\mu^\csink_\ell.
\end{equation}
We refer to 
$\mu^\csink_{j}$, $j\in[0,\ell]$, as to
\emph{close sink-parts}.
The next proposition bounds the total length of close sink-parts 
in~(\ref{eq:piellfinal}), using 
the small number 
\begin{equation}\label{eq:dthree}
	d_3=(\max\{d_0,|\calR|^{d_0}\})^2
\end{equation}
(determined by $\calG=(\calN,\act,\calR)$).

\begin{proposition}\label{prop:dthree}
	$\sum_{j\in[0,\ell]}\length(\mu^{\csink}_j)\leq 
	d_3\cdot(\pressize(T_0,U_0))^2$.
\end{proposition}
\begin{proof}
The number of  subterms of $T_0$ and $U_0$ is
$\pressize(T_0,U_0)$, and each
term can reach at most $\max\{|\calR|^{d_0},d_0\}$
terms within less than $d_0$ steps 
(since $|\calR|^0+|\calR|^1+\cdots +|\calR|^{d_0-1}\leq |\calR|^{d_0}$
when $|\calR|\geq 2$).
Hence there are at most
$\big(\max\{|\calR|^{d_0},d_0\}\cdot \pressize(T_0,U_0)\,\big)^2$ elements
in
$\bigcup_{j\in[0,\ell]}\pairs(\mu^{\csink}_j)$.
Since there is no repeat of a pair in
$\pi_\ell$, 
the claim follows.
\end{proof}

\subparagraph*{Crucial segments of $\pi_\ell$.}
For $\pi_{\ell}=\mu^\csink_0\rho'_1\mu^\uncl_{1}\mu^\usink_{1}\mu^\csink_{1}
\rho'_2\mu^\uncl_{2}\mu^\usink_{2}\mu^\csink_{2}\cdots
\rho'_{\ell}\mu^\uncl_{\ell}\mu^\usink_{\ell}\mu^\csink_{\ell}$
and the respective pivot path 	
$W_0\gt{w_0}W_1\gt{w_1}W_2\gt{w_2}\cdots W_\ell\gt{w_\ell}W_{\ell+1}$,
assuming $\ell\geq 1$,
we say that $W_j$, $j\in[1,\ell]$ is \emph{close}
(which is another variant of closeness to $(T_0,U_0)$) if the path
$W_{j-1}\gt{w_{j-1}}W_j$ visits a subterm of $T_0$ or $U_0$; in this
case we also write $W_{j-1}\gt{w_{j-1}}W_j$ as 
\begin{center}
$W_{j-1}\gt{w'_{j-1}}V_{j-1}\gt{w''_{j-1}}W_j$
\end{center}
where $V_{j-1}$ is the
last subterm of $T_0$ or $U_0$ in the path (not excluding the cases 
$V_{j-1}=W_{j-1}$ and $V_{j-1}=W_{j}$). We note that $W_1$ is close,
since  $W_0\in\{T_0,U_0\}$.

Let $\{j\in[1,\ell]\mid W_j$ is close$\}=\{k_1,k_2,\dots,k_p\}$
where $1=k_1<k_2<k_3\cdots <k_p\leq\ell$; for technical reasons we
also put $k_{p+1}=\ell{+}1$.
The pivot path can be thus written
\begin{equation}\label{eq:closepivotpath}
	\textnormal{
	$W_0\gt{w'_{0}}[V_0\gt{w''_{0}}W_1\gt{w_1}\cdots
W_{k_2-1}]\gt{w'_{k_2-1}}\cdots[V_{k_{p}-1}\gt{w''_{k_{p}-1}}W_{k_{p}}\gt{w_{k_{p}}}\cdots
	W_{\ell}]\gt{w_\ell}W_{\ell+1}$
}
\end{equation}	
where the brackets are just highlighting the corresponding segments.
We use the segmentation~(\ref{eq:closepivotpath}) of the pivot path to 
induce the following segmentation of $\pi_\ell$:
$$\mu^\csink_0\big[\rho'_{1}\cdots\mu^\usink_{k_2-1}\big]\mu^\csink_{k_2-1} 
\big[\rho'_{k_2}\cdots\mu^\usink_{k_3-1}\big]\mu^\csink_{k_3-1}
\cdots\cdots
\big[\rho'_{k_{p}}\cdots\mu^\usink_{\ell}\big]\mu^\csink_{\ell}\,.
$$
The highlighted segments are called the \emph{crucial segments} (of
$\pi_\ell$). 
The total length of ``non-crucial'' segments
$\mu^\csink_0$, $\mu^\csink_{k_2-1}$, $\mu^\csink_{k_3-1}$, $\cdots$, $\mu^\csink_{\ell}$
is bounded by Proposition~\ref{prop:dthree}. 
We note that $\mu^\csink_{j}$ inside the crucial segments are
empty since otherwise we had a close pivot there.

For bounding the number $p$ of crucial segments and their lengths, it
is useful to use the notions of stairs and their simple-stair
decompositions.

\subparagraph*{Stairs, simple stairs, simple-stair decompositions.}

A word $v\in\calR^*$ is a \emph{stair} if $v=\varepsilon$ or 
$v=rv'$ where $r\in\calR$, let $r$ be
$A(x_1,\dots,x_m)\gt{a}E$, and $E\gt{v'}F$ for some $F\not\in\var$.
If $v$ is a stair,
then any path of the form  $V\gt{v}V'$ is also 
called a stair (in the form $A(x_1,\dots,x_m)\sigma\gt{v}F\sigma$).
Hence no prefix of a stair is a sink-segment.

We say that $v=rv'\in\calR^+$
($r\in\calR$) is 
a \emph{simple stair} if $A(x_1,\dots,x_m)\gt{r}E\gt{v'}F$ 
(for $r$ being $A(x_1,\dots,x_m)\gt{a}E$)
where $F$ is a
subterm of $E$ with a nonterminal root (hence $F\not\in\var$)
and $v'$ is a (possibly empty) concatenation of (possibly long) sink-segments
(hence $v'=u_1u_2\cdots u_k$ where $u_i$, $i\in[1,k]$, are
sink-segments).
If $v$ is a simple stair, then also any path $V\gt{v}V'$
is called a simple
stair.

\begin{proposition}\label{prop:stairdecomp}
	\begin{enumerate}
		\item	Any stair $v\in\calR^*$ has the unique \emph{simple-stair
	decomposition}
 $v=v_1v_2\cdots v_q$ ($q\in\Nat$) where $v_i$,
$i\in[1,q]$, are simple stairs.
\item
If $G\gt{v_1v_2\cdots v_q}G'$ where  $v_i$
are simple stairs, then
$\pressize(G')\leq
\pressize(G)+q\cdot\stepinc$; moreover, if $G$ is finite, then
$\height(G')\leq \height(G)+q\cdot\hinc$.
\end{enumerate}
\end{proposition}
\begin{proof}
1. By induction on $|v|$, for stairs $v$. If $v=\varepsilon$, then $q=0$. 
If $|v|>0$, then we write $v=v_1v'$ for the shortest $v_1\in\calR^+$
such that $v'$ is a stair; $v'$ has the unique simple-stair
	decomposition by the induction hypothesis.
We can easily verify that $v_1$ is a simple stair, and that we cannot have 
$v_1v'=v'_1v''$ where $v'_1$ is
a simple stair, $v''$ is a stair (decomposed into simple stairs), 
and $v'_1\neq v_1$.

2. 
We recall that $A(x_1,\dots,x_m)\gt{r}E$ entails 
$\pressize(E\sigma)\leq \pressize(A(x_1,\dots,x_m)\sigma)+\stepinc$,
and we have  
$\pressize(F\sigma)\leq \pressize(E\sigma)$ for any 
subterm $F$ of $E$; moreover, if $A(x_1,\dots,x_m)\sigma$ is finite,
then $\height(F\sigma)\leq\height(E\sigma)\leq 
\height(A(x_1,\dots,x_m)\sigma)+\hinc$.
\end{proof}

\subparagraph*{Bounding the number of crucial segments.}

To 
bound the number $p$ of crucial segments, we use the small number
\begin{equation}\label{eq:dfour}
	d_4=d_1\cdot (1{+}|\nonvarsubrhs|)^{d_2+d_0-1}
\end{equation}
where $\nonvarsubrhs=\{F\mid F$ is a subterm of the rhs 
	of a rule
in $\calR$ and $F\not\in\var\}$.
\begin{proposition}\label{prop:dfour}
The number $p$ of crucial segments is
at most
$d_4\cdot\pressize(T_0,U_0)$.
\end{proposition}
\begin{proof}
First we note that we can have $W_j=W_{j'}$
for different $j,j'\in [1,\ell]$; but for each $W$ we can have 
$W=W_j$ for at most $d_1$ indices $j\in[1,\ell]$, since there are at
most $d_1$ possible bal-results for each pivot
(Proposition~\ref{prop:fewbaltopivot}) and the 
bal-results $(T''_j,U''_j)$, $j\in[1,\ell]$, are all pairwise
different (Proposition~\ref{prop:samelengthmodplays}).

Hence if we get a bound on the cardinality of the set 
$\textsc{SP}=\{W_{k_1},W_{k_2},\dots,W_{k_{p}}\}$ 
of ``starting pivots'' of the crucial segments
(where $k_1=1$), then multiplying this bound by $d_1$ yields a bound
on $p$.

We fix $j\in[1,p]$, and note that
the stair $V_{k_j-1}\gt{w''_{k_j-1}}W_{k_j}$
is a suffix of the path
	$W_{k_j-1}\gt{w^\uncl_{k_j-1}}\overline{W}_{k_j-1}\gt{w^\dsink_{k_j-1}}W_{k_j}$,
where 	$|w^\uncl_{k_j-1}|\leq d_2$ and 
$\overline{W}_{k_j-1}\gt{w^\dsink_{k_j-1}}W_{k_j}$ can be written
$\overline{W}_{k_j-1}\gt{\overline{w}}\overline{\overline{W}}
\gt{\overline{\overline{w}}}W_{k_j}$ where $\overline{w}$ is a
sequence of sink-segments and $|\overline{\overline{w}}|<d_0$.
The simple-stair decomposition of $V_{k_j-1}\gt{w''_{k_j-1}}W_{k_j}$
is thus a sequence of at most  $d_2{+}(d_0{-}1)$ simple stairs.

Hence a (generous) upper bound on $|\textsc{SP}|$ is 
$\pressize(T_0,U_0)\cdot (1{+}|\nonvarsubrhs|)^{d_2+d_0-1}$.
This yields $p\leq \pressize(T_0,U_0)\cdot
(1{+}|\nonvarsubrhs|)^{d_2+d_0-1}\cdot d_1=d_4\cdot\pressize(T_0,U_0)$ as claimed.
\end{proof}

\subparagraph*{Bounding the lengths of crucial segments.}

For $j\in[1,p]$, we view the number $k_{j+1}-k_{j}$
as the \emph{index length} of the crucial 
segment $\big[\rho'_{k_j}\cdots\mu^\usink_{k_{j+1}-1}\big]$. 
We first bound the index length, defining $n,s,g$ and using the bound
on $(n,s,g)$-sequences (Lemma~\ref{lem:ELdecreasbound}), and then we
bound the standard length.

We first note that \emph{each highlighted segment} 
in~(\ref{eq:closepivotpath}) 
\emph{is a stair}. 
Indeed, if
the path 
$$\big[V_{k_j-1}\gt{w''_{k_j-1}}W_{k_j}\gt{w_{k_j}}W_{k_j+1}\gt{w_{k_j+1}}W_{k_j+2}
\cdots \gt{w_{k_{j+1}-2}}W_{k_{j+1}-1}\big]$$ 
(for $j\in[1,p]$)
had a prefix that is a sink-segment, then one of 
$W_{k_j+1},W_{k_j+2},\dots,W_{k_{j+1}-1}$ would be also close, since 
$V_{k_j-1}$ is the last subterm of $T_0$ or $U_0$ in
$W_{k_j-1}\gt{w_{k_j-1}}W_{k_j}$, and each subterm of $V_{k_j-1}$ is
also a subterm of $T_0$ or $U_0$.

Thus the index length of crucial segments is bounded due to the next
lemma, for which we define 
the following small numbers:

\begin{equation}\label{eq:definingn}
n=m^{d_0}\,;
\end{equation}

\begin{equation}\label{eq:definings}
	s=m^{d_0+1}+(m+2)\cdot
	d_0\cdot\stepinc+(d_2+d_0-1)\cdot\stepinc\,;
\end{equation}

\begin{equation}\label{eq:definingg}
g=(d_2+d_0-1)\cdot\stepinc\,.
\end{equation}

\begin{lemma}\label{lem:nsgforstairs}
We assume a balanced modified play
$\pi_\ell=\mu^\csink_0\rho'_1\mu^\uncl_1\mu^\usink_1\mu^\csink_1\cdots
\rho'_\ell\mu^\uncl_\ell\mu^\usink_\ell\mu^\csink_\ell$
and the respective pivot path $W_0\gt{w_0}W_1\gt{w_1}\cdots
W_\ell\gt{w_\ell}W_{\ell+1}$.
Let 
\begin{equation}\label{eq:pivotstair}
V\gt{w}W_{j+1}\gt{w_{j+1}}W_{j+2}\cdots\gt{w_{j+k-1}}
W_{j+k}
\end{equation}
be a segment of the pivot path that is a stair,
where $j\geq 0$, $k\geq 1$, $j+k\leq\ell$,  
and $w$ is a suffix of $w_j$.
Let $e=1+\eqlevel(\finishofplay(\rho'_{j+1}))$
(where $\finishofplay(\rho'_{j+1})$
is the bal-result related to the pivot $W_{j+1}$, hence
$(T''_{j+1},U''_{j+1})$ in~(\ref{eq:modifplayrefined})).

Then $k\leq \calE_\calB$ 
for each $(n,s,g)$-candidate
$\calB$ that is full below $e$; in particular, $k\leq \calE_{\calB_{n,s,g}}$. 
(Here $n,s,g$ are the numbers defined
by~(\ref{eq:definingn}),~(\ref{eq:definings}),~(\ref{eq:definingg}).) 
\end{lemma}	

\begin{proof}
We will show that the (eqlevel-decreasing)
	sequence $\finishofplay(\rho'_{j+1})$, 
$\finishofplay(\rho'_{j+2})$, $\dots$, $\finishofplay(\rho'_{j+k})$
of the bal-results
corresponding to the pivots $W_{j+1}$, 
$W_{j+2}$, $\dots$, 
$W_{j+k}$
can be presented as an 
$(n,s,g)$-sequence
\begin{equation}\label{eq:balcreatensg}
(E_1\sigma,F_1\sigma), (E_2\sigma,F_2\sigma), \dots, 
(E_k\sigma,F_k\sigma).
\end{equation}
The claim then follows by
Lemma~\ref{lem:ELdecreasbound}.
Hence it remains to show the presentation~(\ref{eq:balcreatensg}) of
the respective bal-results.
By the definition of
stairs, 
we can present~(\ref{eq:pivotstair}) as
\begin{center}
$A(x_1,\dots,x_m)\sigma'\gt{w}G'_1\sigma'\gt{w_{j+1}}G'_{2}\sigma'\cdots
	\gt{w_{j+k-1}}
G'_{k}\sigma'$
\end{center}

where
\begin{equation}\label{eq:stairinppathbasic}
\textnormal{
	$A(x_1,\dots,x_m)\gt{w}G'_1\gt{w_{j+1}}G'_{2}\cdots\gt{w_{j+k-1}}G'_{k}$;
}
\end{equation}
we thus have $W_{j+i}=G'_{i}\sigma'$ (for $i\in[1,k]$) where $G'_i$
	are finite terms with nonterminal roots.

Recalling the refined presentation~(\ref{eq:pivotpathfiner}), we write
the path
$W_{j+i}\gt{w^\uncl_{j+i}}\overline{W}_{j+i}\gt{w^\dsink_{j+i}}W_{j+i+1}$,
for each $i\in[0,k{-}1]$, as 
$W_{j+i}\gt{w^\uncl_{j+i}}\overline{W}_{j+i}\gt{\overline{w}_i}\overline{\overline{W}}_{j+i}\gt{\overline{\overline{w}}_i}W_{j+i+1}$
where
$\overline{W}_{j+i}\gt{\overline{w}_i}\overline{\overline{W}}_{j+i}$
is a sequence of sink-segments of lengths less than $d_0$,
and $|\overline{\overline{w}}_i|<d_0$.
We thus present~(\ref{eq:stairinppathbasic}) as
$$A(x_1,\dots,x_m)\gt{w}G'_1\gt{w^\uncl_{j+1}}\overline{G}_1\gt{\overline{w}_{1}}
\overline{\overline{G}}_1
\gt{\overline{\overline{w}}_{1}}
G'_{2}\cdots
\gt{w^\uncl_{j+k-1}}\overline{G}_{k-1}
\gt{\overline{w}_{k-1}}
\overline{\overline{G}}_{k-1}
\gt{\overline{\overline{w}}_{k-1}}G'_{k}\,.$$
We recall that $|w^\uncl_{j+i}|\leq d_2$ (for all $i\in[0,k{-}1]$).
Since $w$ is a suffix of
$w^\uncl_jw^\dsink_j=w^\uncl_j\overline{w}_0\overline{\overline{w}}_0$,
we note that the simple-stair decomposition of the
stair $A(x_1,\dots,x_m)\gt{w}G'_1$ is a sequence of at most 
$d_2{+}(d_0{-}1)$ simple stairs.
More generally, for each $i\in[1,k]$, the simple-stair
decomposition of the stair
$$A(x_1,\dots,x_m)\gt{w}G'_1\gt{w^\uncl_{j+1}}\overline{G}_1\gt{\overline{w}_{1}}
\overline{\overline{G}}_1
\gt{\overline{\overline{w}}_{1}}
G'_{2}\cdots
\gt{w^\uncl_{j+i-1}}\overline{G}_{i-1}
\gt{\overline{w}_{i-1}}
\overline{\overline{G}}_{i-1}
\gt{\overline{\overline{w}}_{i-1}}G'_{i}$$
is a sequence of at most $i\cdot
(d_2{+}(d_0{-}1))$ simple stairs; hence
\begin{equation}\label{eq:gprimegrow}
	\textnormal{
	$\pressize(G'_i)\leq
\pressize(A(x_1,\dots,x_m))+i\cdot(d_2{+}d_0{-}1)\cdot\stepinc$
}
\end{equation}
(recalling Proposition~\ref{prop:stairdecomp}).
We recall the relation of a pivot, $W_{j+i}=G'_i\sigma'$ in our case, 
and its bal-result, 
as captured by Proposition~\ref{prop:relpivottobalres} (and illustrated
in Figure~\ref{fig:balstep}).
We note that $G'_i\sigma'$ might not be a $d_0$-safe form of $W_{j+i}$
(due to possible short branches of $G'_i$). 
This leads us to present 
$V=A(x_1,\dots,x_m)\sigma'$ in a $d_0$-top form, as
$A(x_1,\dots,x_m)\overline{\overline{\sigma}}\sigma$ where 
$A(x_1,\dots,x_m)\overline{\overline{\sigma}}$ is the respective
$d_0$-top.

Putting $G_i=G'_i\overline{\overline{\sigma}}$, we get
$W_{j+i}=G'_i\sigma'=G'_i\overline{\overline{\sigma}}\sigma=G_i\sigma$, for each
$i\in[1,k]$. We have 
$\varin(G_i)\subseteq\varin(A(x_1,\dots,x_m)\overline{\overline{\sigma}})\subseteq \{x_1,\dots,x_n\}$ (for
	$n=m^{d_0}$), and
any word   $v\in\calR^*$ with $|v|\leq d_0$
that is performable from $W_{j+i}=G_i\sigma$ is 
performable from $G_i$ as well.

Since $G_i\sigma$ is thus a $d_0$-safe form of
$W_{j+i}$,
the bal-result related to $W_{j+i}=G_i\sigma$
can be written as $(E_i\sigma,F_i\sigma)$ where 
$\varin(E_i,F_i)\subseteq
\varin(G_i)\subseteq\{x_1,\dots,x_n\}$,
and $\pressize(E_i,F_i)\leq \pressize(G_i)+(m{+}2)\cdot
d_0\cdot\stepinc$ (by Corollary~\ref{cor:sizebalresult}).
By mimicking the derivation of the bound~(\ref{eq:gprimegrow}), we get 
\begin{center}
$\pressize(G_i)\leq
\pressize(A(x_1,\dots,x_m)\overline{\sigma})+i\cdot(d_2{+}(d_0{-}1))\cdot\stepinc$.
\end{center}
Since $\pressize(A(x_1,\dots,x_m)\overline{\sigma})\leq m^{d_0+1}$,
and $g=(d_2{+}d_0{-}1)\cdot\stepinc$, we get
\begin{center}
$\pressize(G_i)\leq
m^{d_0+1}+i\cdot g$, for all $i\in[1,k]$.
 \end{center}
From $\pressize(E_i,F_i)\leq \pressize(G_i)+(m{+}2)\cdot
d_0\cdot\stepinc$ we derive, for all $i\in[1,k]$, that 
$$\pressize(E_i,F_i)\leq m^{d_0+1}+(m+2)\cdot
d_0\cdot\stepinc+i\cdot g=s+(i-1)\cdot g\,.$$
Hence the sequence $\finishofplay(\rho'_{j+1})$, 
$\finishofplay(\rho'_{j+2})$, $\dots$, $\finishofplay(\rho'_{j+k})$
can be indeed presented as an $(n,s,g)$-sequence
$(E_1\sigma,F_1\sigma), (E_2\sigma,F_2\sigma), \dots, 
(E_k\sigma,F_k\sigma)$.
\end{proof}

\begin{corollary}\label{cor:nsgforstairs}
For each crucial segment $\rho'_{k_j}\cdots\mu^\usink_{k_{j+1}-1}$ we
have $k_{j+1}-k_j\leq \calE_{\calB}$ for each $(n,s,g)$-candidate
$\calB$ that is full below $1+\eqlevel(\finishofplay(\rho'_{k_j}))$;
in particular, $k_{j+1}-k_j\leq \calE_{\calB_{n,s,g}}$.
\end{corollary}

We will now  bound the (standard) length of a crucial segment by 
multiplying its index length, increased by $1$, by 
the small
number 
\begin{equation}\label{eq:dfive}
	d_5=(d_2+d_0-1)\cdot(1+(d_0-1)\cdot\hinc)\,.
\end{equation}

\begin{proposition}\label{prop:nsgforstairstwo}
	For each $j\in[1,p]$ we have
	$\length(\rho'_{k_j}\cdots\mu^\usink_{k_{j+1}-1})\leq
	d_5\cdot(1{+}k_{j+1}{-}k_j)$.
\end{proposition}	

\begin{proof}
We fix 
a crucial segment $\rho'_{k_j}\cdots\mu^\usink_{k_{j+1}-1}$.
We make a convenient notational change 
(using $j{+}1$ for the previous $k_j$, and $k$ for $k_{j+1}{-}k_j$)
and present 
this segment as 
\begin{center}
$\rho'_{j+1}\mu^\uncl_{j+1}\mu^\usink_{j+1}
\rho'_{j+2}\mu^\uncl_{j+2}\mu^\usink_{j+2}
\cdots
\rho'_{j+k}
\mu^\uncl_{j+k}\mu^\usink_{j+k}$ 
\end{center}
(for $i\in[1,k{-}1]$ we have $\mu^\usink_{j+i}=\mu^\dsink_{j+i}$ since
$\mu^\csink_{j+i}=\varepsilon$).
In a more detailed presentation, the segment is a prefix of
\begin{equation}\label{eq:detailedcrucial}
\textnormal{	
	$\toppair{T_{j+1}}{U_{j+1}}\bothgt{u_{j+1}}{u'_{j+1}}
\toppair{T'_{j+1}}{U'_{j+1}}\econc
\toppair{T''_{j+1}}{U''_{j+1}}\bothgt{v_{j+1,1}}{v'_{j+1,1}}
\toppair{\overline{T}_{j+1}}{\overline{U}_{j+1}}
\bothgt{v_{j+1,2}}{v'_{j+1,2}}
\toppair{T_{j+2}}{U_{j+2}}
\cdots
\toppair{T_{j+k}}{U_{j+k}}\bothgt{u_{j+k}}{u'_{j+k}}
\toppair{T'_{j+k}}{U'_{j+k}}\econc
\toppair{T''_{j+k}}{U''_{j+k}}\bothgt{v_{j+k,1}}{v'_{j+k,1}}
\toppair{\overline{T}_{j+k}}{\overline{U}_{j+k}}
\bothgt{v_{j+k,2}}{v'_{j+k,2}}\toppair{T_{j+k+1}}{U_{j+k+1}}$
}
\end{equation}
finishing somewhere in the part 
$\toppair{\overline{T}_{j+k}}{\overline{U}_{j+k}}
\bothgt{v_{j+k,2}}{v'_{j+k,2}}\toppair{T_{j+k+1}}{U_{j+k+1}}$,
as determined by $\mu^\csink_{j+k}$ (which might be empty or
nonempty).
We also consider the related pivot-path stair
\begin{equation}\label{eq:crucialsegstairdetail}
V\gt{w}W_{j+1}\gt{w^\uncl_{j+1}}\overline{W}_{j+1}\gt{w^\dsink_{j+1}}W_{j+2}\cdots
\gt{w^\uncl_{j+k-1}}\overline{W}_{j+k-1}\gt{w^\dsink_{j+k-1}}W_{j+k}
\end{equation}
where $V\gt{w}W_{j+1}$ is related to the part
$\rho'_j\mu^\uncl_j\mu^\dsink_j$
that precedes our crucial segment: the path $V\gt{w}W_{j+1}$
is the suffix 
$V_{j}\gt{w''_{j}}W_{j+1}$ of $W_{j}\gt{w^\uncl_{j}}\overline{W}_{j}
\gt{w^\dsink_{j}}W_{j+1}$ 
for the
respective last subterm
$V_{j}$ of $T_0$ or $U_0$.
We present the stair~(\ref{eq:crucialsegstairdetail}) 
similarly as 
the stair~(\ref{eq:pivotstair}) in the proof of
Lemma~\ref{lem:nsgforstairs}.
We get $V=A(x_1,\dots,x_m)\sigma'$ and 
\begin{center}
$A(x_1,\dots,x_m)\gt{w}G'_1\gt{w^\uncl_{j+1}}\overline{G}_1\gt{\overline{w}_1}
\overline{\overline{G}}_1
\gt{\overline{\overline{w}}_{1}}G'_{2}
\cdots
\gt{w^\uncl_{j+k-1}}\overline{G}_{k-1}\gt{\overline{w}_{k-1}}
\overline{\overline{G}}_{k-1}
\gt{\overline{\overline{w}}_{k-1}}G'_{k}$\,.
\end{center}
We will show that
\begin{equation}\label{eq:firstpartlength}
\length(\rho'_{j+1}\cdots\mu^\usink_{j+k-1})\leq 
d_5\cdot k- (d_0{-}1)\cdot\height(G'_k)
\end{equation}
and
\begin{equation}\label{eq:lastpartlength}
\length(\rho'_{j+k}
\mu^\uncl_{j+k}\mu^\usink_{j+k})\leq d_5+(d_0{-}1)\cdot\height(G'_k)\,,
\end{equation}
which yields $\length(\rho'_{j+1}\cdots\mu^\usink_{j+k})\leq
	d_5\cdot(1{+}k)$ and thus finishes the proof.

We show~(\ref{eq:firstpartlength}):
Similarly as~(\ref{eq:gprimegrow}), we derive 
$\height(G'_i)\leq
1+i\cdot(d_2{+}d_0{-}1)\cdot\hinc$, for all $i\in[1,k]$.
Since $|w^\uncl_{j+i}|\leq\length(\rho'_{j+i}\mu^\uncl_{j+i})\leq d_2$,
$|\overline{w}_i\overline{\overline{w}}_i|=\length(\mu^\dsink_{j+i})$,
$\overline{G}_{i}\gt{\overline{w}_i}\overline{\overline{G}}_{i}$ 
is a sequence of sink-segments of lengths less than $d_0$,
and $|\overline{\overline{w}}_i|<d_0$, we also derive
\begin{center}
$|\overline{w}_{i}|\leq (d_0{-}1)\cdot 
(\height(\overline{G}_{i})-\height(\overline{\overline{G}}_{i}))$,
and
\end{center}
\begin{center}
$\height(G'_{i+1})\leq\height(G'_i)+(d_2+d_0{-}1)\cdot\hinc
- (\height(\overline{G}_i)-\height(\overline{\overline{G}}_i))
$.
\end{center}
For $\textsc{Sum}=\sum^{k-1}_{i=1}
\big(\height(\overline{G}_i-\height(\overline{\overline{G}}_i))\big)$
we thus get
\begin{equation}\label{eq:boundsum}
\height(G'_k)\leq 1+k\cdot(d_2+d_0{-}1)\cdot\hinc-\textsc{Sum}, 
\textnormal{ and}
\end{equation}
\begin{equation}\label{eq:firstversion}
\length(\rho'_{j+1}\cdots\mu^\usink_{j+k-1})
\leq
(k{-}1)\cdot(d_2{+}d_0{-}1)+(d_0{-}1)\cdot\textsc{Sum}.
\end{equation}
Replacing  $\textsc{Sum}$ in~(\ref{eq:firstversion}) with its upper
bound $1+k\cdot(d_2+d_0{-}1)\cdot\hinc-\height(G'_k)$
(derived from~(\ref{eq:boundsum})), we get 
\begin{center}
$\length(\rho'_{j+1}\cdots\mu^\usink_{j+k-1})
\leq k\cdot\big(
d_2{+}d_0{-}1+(d_0{-}1)(d_2{+}d_0{-}1)\cdot\hinc\big) -
(d_0{-}1)\cdot\height(G'_{k})$.
\end{center}
This yields~(\ref{eq:firstpartlength}).

We show~(\ref{eq:lastpartlength}):
We recall that $\length(\rho'_{j+k}
\mu^\uncl_{j+k})\leq d_2$, and aim to bound $\mu^\usink_{j+k}$,
assuming $\mu^\usink_{j+k}\neq\varepsilon$.
In this case $\startofplay(\mu^\usink_{j+k})=
(\overline{T}_{j+k},\overline{U}_{j+k})$, and both paths
$\overline{T}_{j+k}\gt{v_{j+k,2}}$, $\overline{U}_{j+k}\gt{v_{j+k,2}}$
of the play $\mu^\usink_{j+k}\mu^{\csink}_{j+k}$ 
(recall~(\ref{eq:detailedcrucial}))
are $d_0$-sinking. 
In the worst case the play $\mu^\usink_{j+k}$
finishes when each of these two paths visits a subterm of $T_0$ or $U_0$
(in which case $\mu^\csink_{j+k}\neq\varepsilon$ follows).
Due to the construction of $\rho'_{j+k}\mu^\uncl_{j+k}$ we have that
both $\overline{T}_{j+k}$ and $\overline{U}_{j+k}$ are reachable from
the pivot $W_{j+k}=G'_k\sigma'\in\{T_{j+k},U_{j+k}\}$ in at most $d_2$
steps (in fact, one even in less than $d_0$ steps).

We recall that $\varin(G'_k)\subseteq\{x_1,\dots,x_m\}$ and
	that $x_q\sigma'$ is a subterm of $T_0$ or $U_0$, for each
	$q\in[1,m]$ (since $V=A(x_1,\dots,x_m)\sigma'$ is a subterm of
	$T_0$ or $U_0$). Thus if the respective paths
	$G'_k\sigma'\gt{\overline{v}}\overline{T}_{j+k}$
	and
	$G'_k\sigma'\gt{\overline{\overline{v}}}\overline{U}_{j+k}$,
	where $|\overline{v}|\leq d_2$ and $|\overline{\overline{v}}|\leq d_2$, 
 ``sink inside'' 
	the
 terms $x_q\sigma'$, they visit subterms of  $T_0$ or $U_0$ at such
 moments.
 The pair
 $(\overline{T}_{j+k},\overline{U}_{j+k})$
can be thus surely presented as 
$(\overline{E}\sigma_1,\overline{F}\sigma_2)$
where $\varin(\overline{E})$ and $\varin(\overline{F})$ are subsets
of $\{x_1,\dots,x_m\}$, the terms
$x_q\sigma_1$ and $x_q\sigma_2$ are subterms of
$T_0$ or $U_0$, for each
$q\in[1,m]$, and both
$\height(\overline{E})$ and $\height(\overline{F})$ are bounded by
$\height(G'_k)+d_2\cdot\hinc$.

Therefore $\mu^\usink_{j+k}$ cannot be longer than 
$(d_0{-}1)\cdot(\height(G'_k)+d_2\cdot\hinc)$. This yields 
$\length(\rho'_{j+k}
\mu^\uncl_{j+k}\mu^\usink_{j+k})\leq 
d_2+(d_0{-}1)\cdot(\height(G'_k)+d_2\cdot\hinc)$, which
implies~(\ref{eq:lastpartlength}). 
\end{proof}

\begin{table}
	\begin{center}
	\begin{tabular}{ l | c | l }
		$m$ & (\ref{eq:marity}) & maximum arity of nonterminals\\
	$\hinc$ &  (\ref{eq:hinc}) & height-increase in one step\\
		$\stepinc$ & (\ref{eq:stepinc}) & size-increase in one step\\
		$d_0$ & (\ref{eq:Mzero})
		& lengths of shortest $(A,i)$-sink words (plus $1$)\\
		$d_1$ & (\ref{eq:done})
		& number of bal-results related to one pivot \\
		$d_2$ & (\ref{eq:shortpivotgap})  & 
		length of ``unclear'' part after a pivot, followed by $d_0$-sinking\\
		$d_3$ &  (\ref{eq:dthree})
		& $d_3\cdot(\pressize(T_0,U_0))^2$ bounds the total
		length of close sink-parts\\	
		$n=m^{d_0}$ & (\ref{eq:definingn}) & number of variables in 
		``$(n,s,g)$-tops'' 
		$(E_i,F_i)$ of the bal-results \\ 
		& & 		related to pivots on a pivot-path stair\\

		$s$ &   (\ref{eq:definings})
		& $\pressize(E_1,F_1)$ 
		for the first such $(n,s,g)$-top\\
		$g$ &  (\ref{eq:definingg})
		& maximal
		growth-rate of $(n,s,g)$-tops\\		
		$d_4$ & (\ref{eq:dfour})
		& $d_4\cdot\pressize(T_0,U_0)$ bounds the number of
		crucial segments\\
		$d_5$ &  (\ref{eq:dfive})
		& $d_5\cdot(1{+}\calE_{\calB_{n,s,g}})$ bounds the length
		of each crucial segment\\		
		$c$ &   (\ref{eq:definingc})
		& $\max\big\{d_3\,, 2\, d_4\, d_5\big\}$,
		the number $c$ in~(\ref{eq:elbound}) in Theorem~\ref{th:computingelbound}   \\
	
	\end{tabular}
\end{center}
\caption{Small upper bounds determined by a given grammar
$\calG$}\label{tab:constants}
\end{table}

\section{Completing the Proof of
Theorem~\ref{th:computingelbound}}\label{sec:finalproof}

Below we repeat the statement of
Theorem~\ref{th:computingelbound},
and show a proof based on the previous results. 
In fact, it remains to prove that  $\calE=\calE_{\calB_{n,s,g}}$
is computable. The idea is that we stepwise increase an under-approximation
of $\calB_{n,s,g}$ and of the respective $\calE$; the pairs
$(E,F)$ (of the respective sizes) that are in this process so far deemed to be equivalent 
(i.e., assumed to satisfy $E\sim F$) are verified by using 
the assumption~(\ref{eq:elboundtwo}) for the current
(under-approximation of) $\calE$. If we find that
$\eqlevel(E,F)\leq c\cdot\big(\calE\cdot
\pressize(E,F)+(\pressize(E,F))^2\big)$
for some of such pairs $(E,F)$
(which can be checked 
Proposition~\ref{prop:stratdecid}),
 then we adjust (increase) the
under-approximation. This process must clearly terminate, and at the
end the claim~(\ref{eq:elboundtwo}) holds for all $T\not\sim U$, as
can be easily shown by contradicting the existence of a violating pair
 $T\not\sim U$ with the least eq-level.

\medskip
\noindent
	\textbf{Theorem 7.}	
\emph{
	For any grammar $\calG=(\calN,\act,\calR)$
there is a small number $c$ and a computable (not
necessarily small) number $\calE$ such that for all
$T,U\in\trees_\calN$ we have:
\begin{equation}\label{eq:elboundtwo}
	\textnormal{
	if $T\not\sim U$ then $\eqlevel(T,U)\leq c\cdot\big(\calE\cdot
\pressize(T,U)+(\pressize(T,U))^2\big)$.
}
\end{equation}
}

\begin{proof}
We fix a grammar $\calG=(\calN,\act,\calR)$, which determines the
small numbers in Table~\ref{tab:constants}, 
and two terms $T_0, U_0$ such that $T_0\not\sim U_0$.
Let $$\pi_\ell=\mu^\csink_0\rho'_1\mu^\uncl_1\mu^\usink_1\mu^\csink_1\rho'_2
\mu^\uncl_2\mu^\usink_2\mu^\csink_2
\cdots \rho'_\ell \mu^\uncl_\ell\mu^\usink_\ell\mu^\csink_\ell$$
be a respective balanced modified play
for which we
use the above developed notions and
notation;
we recall that
$\length(\pi_\ell)=\eqlevel(T_0,U_0)$.
Highlighting the crucial segments, we write $\pi_\ell$ as
$$\mu^\csink_0\big[\rho'_{k_1}\cdots\mu^\usink_{k_2-1}\big]\mu^\csink_{k_2-1} 
\big[\rho'_{k_2}\cdots\mu^\usink_{k_3-1}\big]\mu^\csink_{k_3-1}
\cdots\cdots
\big[\rho'_{k_{p}}\cdots\mu^\usink_{\ell}\big]\mu^\csink_{\ell}\,.
$$
We have $p=0$ (and $\ell=0$) if $\pi_\ell=\mu^\csink_0$; otherwise
$1=k_1<k_2<k_3\cdots<k_p<k_{p+1}=\ell+1$.
The close sink-segments $\mu^\csink_{k_j-1}$, for $j\in[1,p{+}1]$, 
might be empty or
nonempty, but all close sink-segments inside the crucial segments are
empty. 
The total length of the close sink-segments 
is bounded by
$d_3\cdot(\pressize(T_0,U_0))^2$ (by Proposition~\ref{prop:dthree}),
the number $p$ of the crucial segments is bounded by $d_4\cdot \pressize(T_0,U_0)$ 
(by Proposition~\ref{prop:dfour}),
and the length of each crucial segment is bounded by 
$d_5\cdot (1+\calE_{\calB_{n,s,g}})$
(by Corollary~\ref{cor:nsgforstairs} and
Proposition~\ref{prop:nsgforstairstwo}).

Hence $\length(\pi_\ell)$ (and thus $\eqlevel(T_0,U_0)$) is bounded by
\begin{center}
$d_3\cdot(\pressize(T_0,U_0))^2+d_4\cdot\pressize(T_0,U_0)\cdot
d_5\cdot (1+\calE_{\calB_{n,s,g}})$.
\end{center}
Putting 
\begin{equation}\label{eq:definingc}
c=\max\big\{\,d_3\,,\, 2\cdot d_4\cdot d_5\,\big\},
\end{equation}
and recalling that 
$\calE_{\calB}\geq 1$ for any $(n,s,g)$-candidate $\calB$,
we get 
\begin{center}
	$\eqlevel(T_0,U_0)\leq c\cdot\big(\calE_{\calB_{n,s,g}}\cdot
\pressize(T_0,U_0)+(\pressize(T_0,U_0))^2\big)$.
\end{center}
It remains to show that $\calE_{\calB_{n,s,g}}$ is computable.
We first recall that $\calE_{\calB_{n,s,g}}$ in the bound 
$d_5\cdot (1+\calE_{\calB_{n,s,g}})$
on the length of each crucial segment can be refined, as stated in
Corollary~\ref{cor:nsgforstairs}.
For all terms $T,U$ we thus get
the following implication:
\begin{equation}\label{eq:finermainclaim}
\textnormal{
if $T\not\sim U$, then $\eqlevel(T,U)\leq c\cdot\big(\calE_{\calB}\cdot
\pressize(T,U)+(\pressize(T,U))^2\big)$
}
\end{equation}
for any $(n,s,g)$-candidate $\calB$ that is full below
$\eqlevel(T,U)$.  (In this case $\calB$ is surely full below
$1+\eqlevel(E\sigma,F\sigma)$ for the first, and each further, bal-result
$(E\sigma,F\sigma)$ 
in any balanced modified play from $(T,U)$, if there is any
balancing step there at all.)

For $k\in\Nat$ we define the (reflexive and symmetric) relation $\speceq_k$ 
on $\trees_\calN$
as follows:
\begin{center}
	$T\speceq_k U$ $\Leftrightarrow_{df}$
$\eqlevel(T,U)>c\cdot\big(k\cdot
\pressize(T,U)+(\pressize(T,U))^2\big)$;
\end{center}
hence $\sim\mathop{\subseteq}\speceq_k$ for all $k\in\Nat$.
We say that an $(n,s,g)$-candidate $\calB$ 
is \emph{$k$-sound} (for $k\in\Nat$) 
if 
	$(\pairsvs{n}{s}\smallsetminus\calB)\mathop{\subseteq}
	\mathop{\speceq_k}$
	and, moreover, in the case $n>0$ the $(n{-}1,s',g)$-candidate	
	$\calB'$ is $k$-sound (we use the
	notation~(\ref{eq:nextsize})).
	An \emph{$(n,s,g)$-candidate} $\calB$ is \emph{sound} if it is
$\calE_\calB$-sound.
We note that the full candidate $\calB_{n,s,g}$ is sound (since 
all relevant
pairs outside $\calB_{n,s,g}$ are in $\sim$, and
thus in $\speceq_k$ for all $k$).

There is an obvious algorithm that constructs a sound
$(n,s,g)$-candidate $\calB$, for the above defined small $n,s,g$, and
$c$. (Just a systematic brute-force search would do.)

We will now observe that for each sound $(n,s,g)$-candidate $\calB$ we 
have $\speceq_{\calE_\calB}\mathop{=}\sim$ (on the set $\trees_\calN$),
and thus
$\calB=\calB_{n,s,g}$; by this the proof will be finished. 
For the sake of contradiction we suppose a sound $(n,s,g)$-candidate $\calB$
and some
$(T,U)\mathop{\in}\speceq_{\calE_\calB}\cap\not\sim$ where
$\eqlevel(T,U)$ is the least possible.
Then $\calB$ is full below
$\eqlevel(T,U)$ (for any $(T',U')$ with
$\eqlevel(T',U')<\eqlevel(T,U)$ we have $T'\not\speceq_{\calE_\calB}U'$,
hence all relevant $(T',U')$ 
with $\eqlevel(T',U')<\eqlevel(T,U)$
must be in $\calB$ since $\calB$ is
sound).
But then~(\ref{eq:finermainclaim}), applied to our $T,U,\calB$, 
contradicts the assumption $T\speceq_{\calE_\calB}U$.
\end{proof}

\subparagraph*{Acknowledgements.}
The author acknowledges the support of the Grant Agency of the Czech
Rep., GA\v{C}R 18-11193S, and thanks 
Sylvain Schmitz for a detailed discussion 
and his comments helping to improve the form of the paper. 
Also the useful comments of anonymous reviewers are gratefully acknowledged.

\bibliographystyle{elsarticle-num}
\bibliography{pj}

\begin{thebibliography}{10}
\expandafter\ifx\csname url\endcsname\relax
  \def\url#1{\texttt{#1}}\fi
\expandafter\ifx\csname urlprefix\endcsname\relax\def\urlprefix{URL }\fi
\expandafter\ifx\csname href\endcsname\relax
  \def\href#1#2{#2} \def\path#1{#1}\fi

\bibitem{Senizergues:TCS2001}
G.~S\'{e}nizergues, {L(A)=L(B)?} {D}ecidability results from complete formal
  systems, Theor. Comput. Sci. 251~(1--2) (2001) 1--166.

\bibitem{Milner89}
R.~Milner, Communication and concurrency, Prentice-Hall, Inc., 1989.

\bibitem{BBK2}
J.~Baeten, J.~Bergstra, J.~Klop, Decidability of bisimulation equivalence for
  processes generating context-free languages, J.ACM 40~(3) (1993) 653--682.

\bibitem{Srba:Roadmap:04}
J.~Srba, Roadmap of infinite results, in: Current Trends In Theoretical
  Computer Science, The Challenge of the New Century, Vol.~2, World Scientific
  Publishing Co., 2004, pp. 337--350, (updated version at
  http://users-cs.au.dk/srba/roadmap/).

\bibitem{Seni05}
G.~S\'{e}nizergues, The bisimulation problem for equational graphs of finite
  out-degree, SIAM J.Comput. 34~(5) (2005) 1025--1106, (preliminary version at
  FOCS'98).

\bibitem{DBLP:journals/toct/Schmitz16}
S.~Schmitz, Complexity hierarchies beyond elementary, {TOCT} 8~(1) (2016) 3.

\bibitem{Stir:DPDA:prim}
C.~Stirling, Deciding {DPDA} equivalence is primitive recursive, in: Proc.
  ICALP'02, Vol. 2380 of LNCS, Springer, 2002, pp. 821--832.

\bibitem{DBLP:conf/fossacs/Jancar14}
P.~Jan\v{c}ar, Equivalences of pushdown systems are hard, in: Proc. FoSSaCS'14,
  Vol. 8412 of LNCS, Springer, 2014, pp. 1--28.

\bibitem{BGKM12}
M.~Benedikt, S.~G{\"o}ller, S.~Kiefer, A.~S. Murawski, Bisimilarity of pushdown
  automata is nonelementary, in: Proc. LICS 2013, IEEE Computer Society, 2013,
  pp. 488--498.

\bibitem{Cau95}
D.~Caucal, Bisimulation of context-free grammars and of pushdown automata, in:
  A.~Ponse, M.~de~Rijke, Y.~Venema (Eds.), CSLI volume 53 "Modal logic and
  process algebra", Stanford, 1995, pp. 85--106.

\bibitem{CourcelleHandbook}
B.~Courcelle, Recursive applicative program schemes, in: Handbook of
  Theoretical Computer Science, vol. B, Elsevier, MIT Press, 1990, pp.
  459--492.

\bibitem{JancarLICS12}
P.~Jan\v{c}ar, Decidability of {DPDA} language equivalence via first-order
  grammars, in: Proc. LICS 2012, IEEE Computer Society, 2012, pp. 415--424.

\bibitem{DBLP:conf/icalp/Jancar14}
P.~Jan\v{c}ar, Bisimulation equivalence of first-order grammars, in: Proc.
  ICALP'14 (II), Vol. 8573 of LNCS, Springer, 2014, pp. 232--243.

\bibitem{Stirling:TCS2001}
C.~Stirling, Decidability of {DPDA} equivalence, Theor. Comput. Sci. 255~(1--2)
  (2001) 1--31.

\bibitem{Kiefer13}
S.~Kiefer, {BPA} bisimilarity is {EXPTIME}-hard, Inf. Proc. Letters 113~(4)
  (2013) 101--106.

\bibitem{DBLP:conf/mfcs/BurkartCS95}
O.~Burkart, D.~Caucal, B.~Steffen, An elementary bisimulation decision
  procedure for arbitrary context-free processes, in: Proc. of MFCS'95, Vol.
  969 of LNCS, Springer, 1995, pp. 423--433.

\bibitem{Jan12b}
P.~Jan\v{c}ar, Bisimilarity on basic process algebra is in 2-{E}xp{T}ime (an
  explicit proof), Logical Methods in Computer Science 9~(1:10) (2013) 1--19.

\bibitem{HiJeMo96}
Y.~Hirshfeld, M.~Jerrum, F.~Moller, A polynomial algorithm for deciding
  bisimilarity of normed context-free processes, Theor. Comput. Sci. 158 (1996)
  143--159.
\newblock \href {http://dx.doi.org/10.1016/0304-3975(95)00064-X}
  {\path{doi:10.1016/0304-3975(95)00064-X}}.

\bibitem{CzLa10}
W.~Czerwi{\'n}ski, S.~Lasota, Fast equivalence-checking for normed context-free
  processes, in: Proc. FSTTCS'10, Vol.~8 of LIPIcs, Schloss Dagstuhl -
  Leibniz-Zentrum f{\"u}r Informatik, 2010.

\bibitem{DBLP:journals/jacm/JancarS08}
P.~Jan\v{c}ar, J.~Srba, Undecidability of bisimilarity by {D}efender's forcing,
  J. ACM 55 (2008)~(1).

\bibitem{DBLP:journals/corr/YinFHHT14}
Q.~Yin, Y.~Fu, C.~He, M.~Huang, X.~Tao, Branching bisimilarity checking for
  {PRS}, in: Proc. ICALP'14 (II), Vol. 8573 of LNCS, Springer, 2014, pp.
  363--374.

\bibitem{DBLP:conf/fsttcs/BroadbentG12}
C.~H. Broadbent, S.~G{\"o}ller, On bisimilarity of higher-order pushdown
  automata: Undecidability at order two, in: FSTTCS 2012, Vol.~18 of LIPIcs,
  Schloss Dagstuhl - Leibniz-Zentrum f{\"u}r Informatik, 2012, pp. 160--172.

\bibitem{DBLP:conf/lics/Ong15}
L.~Ong, Higher-order model checking: An overview, in: Proc. LICS 2015, {IEEE}
  Computer Society, 2015, pp. 1--15.

\bibitem{DBLP:journals/siglog/Walukiewicz16}
I.~Walukiewicz, Automata theory and higher-order model-checking, {ACM SIGLOG}
  News 3~(4) (2016) 13--31.

\bibitem{DBLP:conf/concur/Stirling06}
C.~Stirling, Second-order simple grammars, in: Proc. {CONCUR}'06, Vol. 4137 of
  LNCS, Springer, 2006, pp. 509--523.

\bibitem{DBLP:conf/lics/JancarS19}
P.~Jan\v{c}ar, S.~Schmitz, Bisimulation equivalence of first-order grammars is
  {A}ckermann-complete, in: 34th Annual {ACM/IEEE} Symposium on Logic in
  Computer Science, {LICS} 2019, Vancouver, BC, Canada, June 24-27, 2019,
  {IEEE}, 2019, pp. 1--12.
\newblock \href {http://dx.doi.org/10.1109/LICS.2019.8785848}
  {\path{doi:10.1109/LICS.2019.8785848}}.

\bibitem{DBLP:journals/jcss/Jancar20}
P.~Jan\v{c}ar, Deciding semantic finiteness of pushdown processes and
  first-order grammars w.r.t. bisimulation equivalence, J. Comput. Syst. Sci.
  109 (2020) 22--44, (a preliminary version at MFCS'16).
\newblock \href {http://dx.doi.org/10.1016/j.jcss.2019.10.002}
  {\path{doi:10.1016/j.jcss.2019.10.002}}.

\end{thebibliography}

\end{document}